\documentclass[a4paper,12pt]{article}
\usepackage{amsmath}
\usepackage{amssymb}
\usepackage{amsthm}
\usepackage{enumerate}
\usepackage{latexsym}
\usepackage{mathrsfs}
\usepackage[dvips]{graphicx}
\usepackage[colorlinks=true, citecolor=blue]{hyperref}
\usepackage{pstricks}
\usepackage{color}
\usepackage{tikz}
\usetikzlibrary{decorations.pathreplacing}
\def\f(#1){{\mathop{f}^{(#1)}}}
\def\m(#1){{\mathop{m}^{(#1)}}}
\def\p(#1){{\mathop{p}^{(#1)}}}

\def\ben{\begin{equation}}
\def\een{\end{equation}}
\def\bena{\begin{eqnarray}}
\def\eena{\end{eqnarray}}
\def\non{\nonumber}

\def\d{{\rm d}}

\def\S{{\bf S}}
\def\I{{\cal I}}
\def\C{{\cal C}}

\def\L{{\cal L}}
\def\D{D}

\def\mr{{\mathbb R}}

\def\G{{\cal G}}
\def\F{{\cal F}}

\def\T{{\mathbb T}}

\newcommand{\bfD}{{\bf D}}

\newcommand{\mn}{{\mathbb N}}

\newcommand{\id}{id}

\renewcommand{\D}{{\mathcal D}}
\newcommand{\R}{{\mathcal R}}

\newcommand{\la}{{\hat a }}

\renewcommand{\S}{{\mathscr S}}
\renewcommand{\O}{{\mathcal O}}

\newcommand{\e}{{\rm e}}

\renewcommand{\atop}[2]{\genfrac{}{}{0pt}{}{#1}{#2}}
\newcommand{\varp}{\frac{\delta}{\delta\varphi}}
\newcommand{\Lsc}{\mathcal{L}^{\Lambda, \Lambda_{0}}}

\newcommand{\LscO}{\mathcal{L}^{0, \Lambda_{0}}}
\newcommand{\LscIr}{\mathcal{L}^{\Lambda_{0}, \Lambda_{0}}}
\def\bra{\langle }
\def\ket{\rangle}
\newcommand{\otimesu}[1]{\underset{#1}{\otimes}}

\renewcommand{\bigotimes}{\otimes}

\newcommand{\La}{\Lambda}
\newcommand{\Lao}{\Lambda_0}

\newcommand{\ka}{\kappa}
\newcommand{\vp}{\varphi}
\newcommand{\pa}{\partial}
\newcommand{\de}{\delta}

\theoremstyle{theorem}

\newtheorem{thm}{Theorem}

\newtheorem{lemma}{Lemma}

\newtheorem{prop}{Proposition}

\newtheorem{corollary}{Corollary}

\newtheorem{defn}{Definition}

\begin{document}

\title{Operator product expansion algebra} 

\author{
Jan Holland\thanks{\tt HollandJW1@Cardiff.ac.uk}\:
and
Stefan Hollands\thanks{\tt HollandsS@Cardiff.ac.uk}\:
\\ \\
{\it School of Mathematics, Cardiff University} \\
{\it Cardiff, United Kingdom} \medskip \\
}

\date{\today}

\maketitle

\begin{abstract}
We establish conceptually important properties of the operator product expansion (OPE)
in the context of perturbative, Euclidean $\varphi^{4}$-quantum field theory.
 First, we demonstrate, generalizing earlier results and techniques of~hep-th/1105.3375, that the 3-point OPE, $\bra \O_{A_{1}}\O_{A_{2}}\O_{A_{3}} \ket = \sum_{C}\C_{A_{1}A_{2}A_{3}}^{C}\bra\O_{C}\ket$, usually interpreted only as an asymptotic  short distance expansion, actually {\em converges} at {\em finite}, and even large, distances. We further show that the factorization identity $\C_{A_{1}A_{2}A_{3}}^{B}=\sum_{C}\C_{A_{1}A_{2}}^{C}\C_{CA_{3}}^{B}$ is satisfied for suitable configurations of the spacetime arguments. Again, the infinite sum is shown to be convergent. Our proofs rely on explicit bounds on the remainders of these expansions, obtained using refined versions, mostly due to Kopper et al., of the renormalization group flow equation method. These bounds also establish that each OPE coefficient is a real analytic function in the spacetime arguments for non-coinciding points. Our results hold for arbitrary but finite loop orders. They lend support
 to proposals for a general axiomatic framework of quantum field theory, based on such `consistency conditions' and akin to vertex operator algebras, wherein the OPE is promoted to the defining structure of the theory.
\end{abstract}



\section{Introduction}

The operator product expansion
(OPE)~\cite{Wilson:1971bg, Wilson:1971dh, zimbrand}
\ben
\label{ope1}
\O_{A_{1}}(x_{1})\cdots \O_{A_{N}}(x_{N})
\sim \sum_C \C_{A_{1}\ldots A_{N}}^C(x_{1},\ldots, x_{N}) \, \O_C(x_{N}),
\een
is an important structure in quantum field theory, both for practical
calculations, as well as also from a conceptual point of view. It is normally understood
as a statement about operator insertions into any, suitably well-behaved, quantum state
(in the Lorentzian context), or as an insertion into the vacuum functional with any number of additional
``spectator fields'' (in the Euclidean context). The symbols $\O_A$ denote the composite
fields that appear in the given theory, where the label $A$ also incorporates the
tensor or spinor character of the field. In the Euclidean context, which we will
focus on in this paper, the points $x_i \in \mr^4$ are mutually distinct.

In the original papers, it was suggested that the OPE should be an asymptotic expansion, after
insertion into the vacuum functional with spectator fields. Namely, the remainder
in such an expansion, truncated at a sufficiently high dimension $D$ for the operators
$\O_C$ on the right side, should go to zero as $x_1, \dots, x_{N-1} \to x_N$, at a rate that
improves as we increase $D$. However, it has recently been shown~\cite{Hollands:2011gf} that,
at least to any fixed loop order in perturbation theory, and $N=2$, the expansion is even {\em convergent}
for {\em finite} separation of $x_1, x_2$, i.e. we can put ``$=$'' instead of
``$\sim$'' in~\eqref{ope1}. This result is complemented by
results of~\cite{Bostelmann:2005fa} obtained in an axiomatic setting of quantum field theory (QFT), as well as also by results in the
context of vertex operator algebras~\cite{Huang:2007fa} describing $2d$ conformal field theories.

The {\em first result} of this paper is to generalize the convergence statement in perturbative QFT to more than just $N=2$ fields in the product~\eqref{ope1}. As in~\cite{Hollands:2011gf},
we will restrict attention for simplicity and concreteness to the case of $g\varphi^4$-theory
in $4d$ with mass $m>0$  (this restriction is made to simplify our proofs; we believe that it can be
removed using methods along the lines of~\cite{guida}), and to keep the discussion reasonably simple, we will also restrict to $N=3$ points,
where the key technical differences in the proof are already visible. The composite fields in the theory are
\ben\label{compfields}
\O_A = \partial^{w_1} \varphi \cdots \partial^{w_n} \varphi \ ,
\quad A = \{n, w\}
\een
with $w$ a multiindex, see our Notations and Conventions section for more on multi-indices. The engineering dimension of such a field is defined as usual by

\ben\label{compop}
[A]=n+\sum_{i}|w_{i}|\, .
\een
 The result is:

\begin{thm}\label{OPEbound}
Let $f_p(x)$ be any smooth function on $x \in \mr^4$ such that the support of the
Fourier transform $\hat f_p(q)$ is contained in a ball $|p-q| \le 1$ around $p \in \mr^4$. Define the
smeared spectator fields by
$$
\varphi(f_p) \equiv \int d^4 x \ \varphi(x)  \ f_p(x) \ .
$$
Then the remainder of the OPE, carried out up to operators of dimension $D=\sum_{i=1}^{3}[A_{i}]+\Delta$, at $l$-loops, is bounded by
\ben
\begin{split}
&\Big| \Big\bra \O_{A_{1}}(x_{1})\O_{A_{2}}(x_{2})\O_{A_{3}}(x_3)\, \varphi(f_{p_{1}})\cdots\varphi(f_{p_{n}}) \Big\ket\\
 &\qquad\qquad- \sum_{[C] \leq D}\C_{A_{1}A_{2} A_{3}}^{C}(x_{1},x_{2},x_{3})\ \Big\bra \O_{C}(x_3)\, \varphi(f_{p_{1}})\cdots\varphi(f_{p_{n}}) \Big\ket  \Big|\\
 \leq\quad& m^{n-1}\, \prod_{i=1}^{3}[A_{i}]!\ \prod_{j} \sup|\hat{f}_{p_{j}}|\, \sup(1,\frac{|\vec{p}|_{n}}{m})^{(4\sum [A_i]+2\Delta)(n+l+9/2)+3n}
 \\ &\quad \times
 \sum_{\lambda=0}^{2l+n/2+1}\frac{\log^{\lambda}\sup(1,\frac{|\vec{p}|_{n}}{m})}{2^{\lambda}\lambda!}
 \\
  & \quad \times \frac{1}{\sqrt{\Delta!}}
  \ \ \frac{(\tilde{K}\ m\ {\rm max}(|x_1-x_2|,|x_2-x_3|,|x_1-x_3|) )^{\sum [A_{i}]+1+\Delta}}{( {\rm min}(|x_1-x_2|,|x_2-x_3|,|x_1-x_3|) )^{\sum [A_i]+1}}
\end{split}
\label{ope3conv}
\een
in $g\varphi^4$-theory.
Here 
$[A]$
denotes the canonical dimension of a composite field ${\mathcal O}_A$ as in eq.~\eqref{compop}.
$\tilde{K}$ is a constant depending on $n$ and $l$.
\end{thm}
 From the bound one can draw the conclusion  that the OPE, i.e. the sum over $C$ in~\eqref{ope3conv}, converges for arbitrary finite distances $|x_i-x_j| > 0$ for $i \neq j$, including large distances (!), because evidently $ ({\rm any \ \ number})^{\Delta}/\sqrt{\Delta !} \to 0$. In this regard, our result is a generalization of our earlier result~\cite{Hollands:2011gf} for $N=2$ points. Another important conclusion, which has no analog for $N=2$ points, is the following. Suppose we expand the OPE to operators up to a fixed dimension $D=[A_{1}]+[A_{2}]+[A_{3}]+\Delta$. Suppose that $|x_i - x_j| = O(\epsilon)$ for $i \neq j$, and let $\epsilon \to 0$. Then the remainder goes as $O(\epsilon^{\Delta})$, i.e. goes to zero e.g. for $\Delta =1$.
However, suppose now instead that $|x_2-x_3|=|x_1-x_3|=O(\epsilon)$, but $|x_1-x_2|=O(\epsilon^2)$, i.e. one pair of points is much closer together than the remaining ones. Then, as one can see from the structure of the last
line in equation \eqref{ope3conv}, the bound will {\em not} go to zero as $\epsilon \to 0$,
unless we make $\Delta$ considerably bigger. In that case, convergence of the remainder is not guaranteed by our bound.  Thus, we see that the sense in which the OPE approximates a correlator
for short distances may depend on exactly how the relative distances between the points are taken to zero, when more than two insertions are present in the correlator.

\medskip
\noindent
The {\em second result} of this paper clarifies the nature of the algebraic relations that are satisfied
by the OPE coefficients $\C_{A_1 \dots A_N}^C$, again in the context of perturbation theory. An obvious general expectation is that, since these coefficients resemble the structure constants of an algebra, there should
hold certain associativity conditions. For example, consider a product of three fields $\O_{A_1}(x_1)
\O_{A_2}(x_2) \O_{A_3}(x_3)$, which we may expand as in eq.~\eqref{ope1}. Now, suppose that $|x_1-x_2|$
is smaller than $|x_2-x_3|$. Then it seems natural to first expand the product $\O_{A_1}(x_1) \O_{A_2}(x_2)$,
regarding $\O_{A_3}(x_3)$ as merely a spectator, and then to expand the result of this OPE times the spectator in a subsequent OPE. One would expect these expansions to agree, and our theorem shows that this expectation is correct:
\begin{thm}\label{thmfact}
Up to any arbitrary but fixed loop order $l$ in $g\varphi^4$-theory, the identity
\ben\label{ope2}
\C_{A_1 A_2 A_3}^B(x_1, x_2, x_3) = \sum_C \C_{A_1 A_2}^C(x_1, x_2) \ \C_{CA_3}^B(x_2, x_3)
\een
holds for all configurations satisfying
$$0<\frac{|x_{1}-x_2|}{|x_2-x_{3}|} < \frac{1}{\tilde K}$$
for some (sufficiently large) constant $\tilde K>0$ (depending on $l, B$).
\end{thm}
In particular, implicit in this statement is the claim that the infinite sum over $C$ on the right side of formula~\eqref{ope2}
converges.  In the free field theory, our result is fairly trivial
and we may in fact take $\tilde K=1$. However, in the presence of interaction the result appears to us rather non-trivial
and we can only show that convergence occurs if $\tilde K$ is a rather large numerical constant, i.e.
it appears that $|x_1-x_2|$ must be {\em much} smaller than $|x_2-x_3|$.

\medskip
\noindent
As in the previous paper~\cite{Hollands:2011gf}, our proof is based on the use of the Wilson-Wegner-Polchinski renormalization group flow equation method \cite{Polchinski:1983gv, Wegner:1972ih, Wilson:1971bg,Wilson:1971dh,Wetterich:1992yh} and the powerful refinements of this method
due to Kopper et al., see e.g.~\cite{Muller:2002he} for a review. In this method, one first introduces an infrared cutoff called $\Lambda$, and an ultraviolet
cutoff called $\Lambda_0$. One then defines the quantities of interest such as correlation functions, OPE coefficients etc. for finite values of the cutoffs, and derives for them a flow
equation as a function of $\Lambda$. This equation is then integrated subject to boundary conditions
which play the role of the renormalization conditions. The quantities of interest may be bounded
inductively and uniformly  in the ultraviolet cutoff $\Lambda_0$.
The last fact makes it
possible to remove the cutoff and
at the same time provides non-trivial bounds for them. These bounds
are shown to imply the above two theorems.
A side result that might be of independent interest is that the Schwinger functions with up to three insertions, as well as the OPE coefficients of up to three composite operators, are analytic functions of $x_i$ away from
coinciding configurations (where they are of course singular). These results are discussed briefly below the corresponding bounds, see corollaries~\ref{CAG2cor} and~\ref{corCAG3}.

It is evident that~\eqref{ope2} is only one in an entire hierarchy of more general consistency conditions: We can consider more
than $N=3$ points, and we can consider, out of $N$ points, a proper subset of points $(x_{i_1}, \dots, x_{i_r})$
where $\{i_1, \dots, i_r\} \subset \{1, \dots, N\}$, such that the distances $|x_{i_k} - x_{i_l}|$ within the group are much smaller than the distance to any point outside the group. The, fairly obvious, consistency conditions arising in this way have been written out e.g. in~\cite{Hollands:2008vq}, and are reviewed for completeness
in the next section of this paper. The methods developed in this paper
are capable also of proving the validity of such generalized consistency conditions in perturbative QFT. However, it seemed to us that
the proof for three points was already illustrative of all the essential ingredients of our methods, while being more easy to follow. We therefore will not consider the general case in this paper. We also strongly believe that our results are still valid if Euclidean space $\mr^4$ is replaced by a general real analytic Riemannian 4-manifold $(M,g)$, with the distance between points defined in terms of the metric in the usual way.

The paper is organized as follows: To put our results into context, we first recall the general nature of the consistency conditions that one
expects to hold in QFT. Then we briefly introduce our notation and the general framework based on the renormalization group flow equations, which will be used throughout this paper. In section \ref{secBounds} we will then put this framework to use and derive bounds on various quantities of interest. With these estimates at hand, we are able to derive our two main results, theorems \ref{OPEbound} and \ref{thmfact}, in section \ref{secOPEconv} and \ref{secFactor}, respectively.

\paragraph*{Notations and conventions:} Our convention for the Fourier
transform in $\mr^4$ is
\ben
f(x) = \int_p  \hat f(p) \e^{ipx} := \int_{\mr^4} \frac{d^4 p}{(2\pi)^4}
\e^{ipx} \hat f(p) \ .
\een
 We also use a
standard
multi-index notation.
Our multi-indices are elements $w = (w_1, \dots, w_n) \in \mn^{4n}$,
so
that each  $w_i \in \mn^4$
is a four-dimensional multiindex whose entries are $w_{i,\mu} \in \mn$
and $\mu=1,\dots,4$. If $f(\vec p)$ is a smooth function on $\mr^{4n}$, we set
\ben
\pa^{w} f(\vec p) = \prod_{i,\mu}
\left( {\pa \over \pa p_{i,\mu}} \right)^{w_{i,\mu}} f(\vec p)
\een
and
\ben
w! = \prod_{i,\mu} w_{i,\mu}! \, , \quad |w|=\sum_{i,\mu} w_{i,\mu} \, .
\een
We often need to take derivatives $\partial^w$ of a product of
functions
$f_1 \dots f_r$.
Using the Leibniz rule, such derivatives get distributed over the
factors
resulting
in the sum of all terms of the form $c_{\{v_i\}} \ \partial^{v_1} f_1
\dots \partial^{v_r} f_r$, where
each $v_i$ is now a $4n$-dimensional multi-index, where
$v_1+\dots+v_r=w$,
and where
\ben
c_{\{v_i\}} = \frac{(v_1+\dots+v_r)!}{v_1! \dots v_r!} \le r^{|w|}
\een
is the associated weight factor. We will denote sets of indices by $I=\{i_{1},\ldots,i_{k}\}$ with $i_{j}\in\mathbb{N}$ and we denote their cardinality by $|I|$.

For a given set of momenta $(p_{1},\ldots, p_{n})\in\mathbb{R}^{4n}$ we use the shorthand notation
\ben\label{pshort}
\vec{p}:= (p_{1},\ldots, p_{n})\quad , \quad |\vec{p}|_{n}:= \sup_{J\subseteq \{1,\ldots, n\}}\, \Big|\sum_{i\in J} p_{i} \Big|\quad , \quad \vec{p}_{n+2}:= (\vec{p},k,-k)
\een
Later we will often simply write $|\vec{p}|$ instead of $|\vec{p}|_{n}$. Further we define $\kappa:=\sup(\Lambda,m)$ for later convenience. We also often use the notation $\log_{+}(x)=\log(\sup(1,x))$.

If $F(\varphi)$ is a differentiable function (in the Frechet space
sense)
of the
Schwartz space function $\varphi \in \S(\mr^4)$, we denote its
functional
derivative as
\ben
\frac{\d}{\d t} F(\varphi + t\psi) |_{t=0} = \int d^4 x \
\frac{\delta F(\varphi)}{\delta \varphi(x)} \ \psi(x) \ ,
\quad \psi \in \S(\mr^4)\ ,
\een
where the right side is understood in the sense of distributions in
$\S'(\mr^4)$. Multiple functional derivatives are
denoted in a similar way and define in general distributions on
multiple Cartesian copies of $\mr^4$.

\paragraph*{Note added in proof:} After this paper was completed, we found an improved proof of the factorization property that also 
generalizes thm.~\eqref{thmfact} to the case of an arbitrary number of factors in the operator product. These results will be included in a forthcoming paper by the 
same authors.

\section{Consistency conditions}\label{secPerspective}

Before we start with the proofs of our two theorems, we would like to pause for
a moment to indicate in somewhat more detail the status of these results within
the general framework of quantum field theory (QFT). One perspective on QFT is
to emphasize its algebraic structure. The oldest mathematical framework of
this viewpoint is Algebraic Quantum Field Theory~\cite{haag}, wherein the observables
are members of a net of $C^*$-algebras, whereas the states (`Hilbert space') are vectors
in the various representations of such a net. Another manifestation of this idea,
developed for $2d$ conformal field theories, are ``Vertex Operator Algebras'', see e.g.~\cite{Huang:2007fa}, and their representation theory. A generalization
of this type of algebraic framework, wherein the OPE is raised to the status of
the defining structure of the theory, but which is applicable in principle more generally to
non-conformally invariant Euclidean QFT's in $d$ dimensions, was proposed e.g. in~\cite{Hollands:2008vq}.

In this framework, one considers the collection of composite fields $\O_A$ within a given theory as being in one-to-one correspondence with basis elements of an abstract vector space, $V$, which is graded by the dimensions of the fields
\ben
V = \bigoplus_{\Delta} V_\Delta \quad ,
\een
with $\Delta \in \mr_+$ from the (countable) set of possible dimensions of the fields,
and $\dim V_\Delta < \infty$ (one might further wish to complement this condition by a more stringent
finiteness condition of the type $\sum_\Delta \dim V_\Delta \ q^{\Delta} < \infty$ for $0<q<1$). The OPE-coefficients $\C_{A_1 \dots A_N}^B(x_1, \dots, x_N)$
may then be viewed, for each collection of distinct points $x_i \in \mr^d$, as the components of a linear map
\ben
\C(x_1, \dots, x_N) : V^{\otimes N} \to V \ ,
\een
which depends (real) analytically on $(x_1, \dots, x_N)$ within the set
$C_N=\{(x_1, \dots, x_N) \ \mid \ x_i \neq x_j \}$. Natural conditions on this collection
(all $N$) of linear maps were proposed in~\cite{Hollands:2008vq}. In particular, the hierarchy of $\C$'s should
reflect the algebraic nature of QFT. This is formulated by imposing consistency conditions,
of the kind described in the introduction, on the $\C$'s. For this, one considers
a partition $\I = (I_j)_{j=1}^n, I_j \subset \{1, \dots, N\}$ of the set $\{1,\dots,N\}$ into disjoint subsets (for a single point, we put $C(x)=\id_{V}$) and a subset of $C_\I$ of the form
\ben
\begin{split}
C_\I = \bigg\{
& (x_1, \dots, x_N) \in C_N \ \  \bigg| \ \  \frac{|x_i-x_j|}{|x_k-x_l|} < \frac{1}{K} \\
& \text{if $i,j \in I_m$, $k \in I_s$, $l \in I_r$, $s \neq r$}
\bigg\} \ ,
\end{split}
\een
where $K$ is some given large number.
The set contains those configurations for which the distance between points $x_i, x_j$ for $i,j$ from the
same subset in the partition is much smaller than for different partitions (the sets
are reminiscent, and in fact closely related, to those defining the ``little disc operad'').
The factorization law is then that for $(x_1, \dots, x_N) \in C_\I$, and for $K$ sufficiently large, we should have
\ben
\C(x_1, \dots, x_N) = \C(x_{j_1}, \dots, x_{j_n}) \circ \bigg(
\bigotimes_{i=1}^n \C(x_k)_{k \in I_i}
\bigg) \ .
\een
Here, each $j_k$ is the element of largest index in the subset $I_k$.
The circle $\circ$ denotes the composition of maps. In view of the infinite dimensional
nature of $V$, one effectively assumes here, in particular, that the composition(s) be well defined. Implicit is therefore a statement about convergence as in eq.~\eqref{ope2} of Thm.~\ref{thmfact}, which corresponds to $N=3$ and the partition $\I: I_1=\{1,2\}, I_2=\{3\}$. The theorem derives its conceptual importance from this context.

Consistency conditions now arise because the above identity must hold for {\em all} possible partitions. In particular, by considering, for $N=3$ points e.g. the alternative
partition ${\mathcal J}: I_1=\{1,3\},I_2=\{2\}$, one in effect has a kind of Jacobi-identity
for the OPE. Note, however, that this identity is more subtle than for ordinary algebras, because the domains $C_{\I}$ resp. $C_{\mathcal J}$ in which the respective factorization identity holds, are not guaranteed to have any overlap. Both identities are rather related by analytic continuation in the points $(x_1,x_2,x_3)$, since the 3-point coefficient $\C(x_1, x_2, x_3)$ is analytic in its arguments.

After these general points, the remainder of the paper is devoted to proving thms.~\ref{OPEbound} and \ref{thmfact} in a concrete (perturbative) model in $d=4$ dimensions, thereby
lending support to the type of axiomatic framework which we have just outlined.

\section{The flow equation framework}\label{secFrame}

The model that we consider is
a hermitian scalar field with self-interaction $g \varphi^4$ and
mass $m > 0$ on flat 4-dimensional Euclidean space. The quantities of
interest in this (perturbative) quantum field theory will be defined
in this section via the flow equation method.
Renormalization theory based on the  flow equation (FE) \cite{Polchinski:1983gv, Wegner:1972ih, Wilson:1971bg,Wilson:1971dh}
of the renormalization group
has been reviewed quite often in
the literature, so we will be relatively brief.
The first presentation in the form we use it here is in \cite{Keller:1990ej}.
Reviews are in \cite{Muller:2002he} and in \cite{Kopper:1997vg} (in German).

\subsection{Connected amputated Green functions (CAG's)}

To begin, we introduce an infrared cutoff $\Lambda$, and an ultraviolet cutoff $\Lambda_0$. The IR cutoff is of
course not necessary
in a massive theory. The IR behavior is substantially modified only
for $\La$ above
$m$. These
cutoffs enter the
definition of the theory through the propagator
$C^{\Lambda,\Lambda_0}$,
which is defined in
momentum space by
\ben\label{propreg}
C^{\La,\Lao}(p)\,=\, {1 \over  p^2+m^2}
\left[ \exp \left(- {p^2+m^2 \over \Lao^2} \right) - \exp
\left(- {p^2+m^2 \over \La^2} \right) \right] \, .
\een
The full propagator is recovered for $\La \to 0$ and $\Lao \to \infty$,
and we always assume
\ben
\label{ka}
0<\Lambda\ , \quad \ka := \sup(\La,m) < \Lao \ .
\een
Other choices of regularization are of course admissible. The one
chosen in (\ref{propreg}) has the advantage of being analytic in $p^2$
for $\Lambda>0$.
The propagator defines a corresponding
Gaussian measure
$\mu^{\La,\Lao}$, whose covariance is $\hbar C^{\La,\Lao}$.
The factor of $\hbar$ is inserted to obtain a consistent
loop expansion in the following.
The interaction is taken to be
\ben
L^{\Lambda_0}(\varphi) = \int d^4 x \ \bigg( a^{\Lambda_0}
\, \varphi(x)^2
+b^{\Lambda_0} \, \partial \varphi(x)^2+c^{\Lambda_0}
\, \varphi(x)^4 \bigg) \ .
\label{ac}
\een
It contains suitable counter terms satisfying
$a^{\Lambda_0} = O(\hbar),\  b^{\Lambda_0} = O(\hbar^2)$
and $c^{\Lambda_0} = \frac{g}{4!} +
O(\hbar)$. They will be
adjusted--and actually diverge--when $\Lambda_0 \to \infty$
in order to obtain a well
defined limit of the quantities of interest to us. We have anticipated this
by making them ``running couplings'', i.e. functions of the ultra
violet cutoff $\Lambda_0$.
The correlation ($=$ Schwinger- $=$ Greens- $=$ $n$-point-) functions of $n$ basic fields with
cutoff are defined by the expectation values
\ben\label{pathint}
\begin{split}
 \langle \varphi(x_1) \cdots \varphi(x_n) \rangle &\equiv  \mathbb{E}_{\mu^{\Lambda,\Lambda_{0}}} \bigg[\exp \bigg( -\frac{1}{\hbar}
L^{\Lambda_0}\bigg) \, \varphi(x_1) \cdots \varphi(x_n) \bigg] \bigg/ Z^{\Lambda,\Lambda_0} \\
& =
\int d\mu^{\Lambda,\Lambda_0} \ \exp \bigg( -\frac{1}{\hbar}
L^{\Lambda_0}\bigg) \, \varphi(x_1) \cdots \varphi(x_n)\bigg/ Z^{\Lambda,\Lambda_0} \, .
\end{split}
\een
This is just the standard Euclidean path-integral, but note that
the free part in the Lagrangian
has been absorbed into the Gaussian measure $d\mu^{\La,\Lao}$.
The normalization factor
is chosen so that $\langle 1 \rangle = 1$. This factor is finite only
as long as we impose an additional volume cutoff. But the infinite volume limit
can be taken without difficulty once we pass to perturbative
connected correlation functions which we will do in a moment.
For more details on this limit see~\cite{Kopper:2000qm,Muller:2002he}.
 The path integral will be analyzed in
the perturbative sense, i.e. the exponentials are expanded out and the
Gaussian integrals are then performed.
The full theory is obtained by sending the cutoffs $\Lambda_0 \to
\infty$
and $\Lambda \to 0$,
for a suitable choice of the running couplings.
In the FE technique, the
correct behavior of the running
couplings, necessary for a well-defined limit, is obtained by deriving
first a differential equation for the Schwinger functions in
$\Lambda$,
and by then defining the running couplings implicitly through the boundary
conditions for this equation.

These FEs are written more conveniently in
terms of the hierarchy of ``connected, amputated
Schwinger functions'' (CAG's).
Their generating functional is defined  through the
convolution of the Gaussian measure with the exponentiated interaction,
\ben\label{CAGdef}
-L^{\Lambda, \Lambda_0} := \hbar \, \log \, \mu^{\Lambda,\Lambda_0}
\star \exp \bigg(-\frac{1}{\hbar} L^{\Lambda_0}
\bigg)- \hbar \log  Z^{\La,\Lao} \ .
\een
The convolution is defined in general by
$(\mu^{\Lambda,\Lambda_0} \star F)(\varphi) =
\int d\mu^{\Lambda,\Lambda_0}(\varphi') \ F(\varphi+\varphi')$.
The functional
$L^{\Lambda,\Lambda_0}$ has an expansion
as a formal power series
in terms of Feynman diagrams with precisely $l$ loops, $n$ external
legs, and propagator $C^{\Lambda,\Lambda_0}(p)$. As the name suggests,
only connected diagrams contribute,
and the (free) propagators on the external legs are removed. We will
not use  decompositions in terms of Feynman diagrams.
But we will also analyze the functional (\ref{CAGdef})
in the sense of formal power
series
\ben\label{genfunc}
L^{\Lambda, \Lambda_0}(\varphi) := \sum_{n>0}^\infty
\sum_{l=0}^\infty {\hbar^l}
\int d^4x_1 \dots d^4 x_n\ \L^{\Lambda,\Lambda_0}_{n,l}(x_1, \dots, x_n)
\,
\varphi(x_1) \cdots \varphi(x_n) \, ,
\een
where $\varphi \in \S(\mr^4)$ is any Schwartz space function. No
statement
is made about the
convergence of the series in $\hbar$.
The objects on the right side, the CAG's,
 are the basic quantities in our analysis
because they are
easier to work with than the full Schwinger functions. But the
latter can of course be recovered from the CAG's.

Because the connected amputated functions in position space are translation
invariant, their Fourier transforms, denoted
$\L^{\La,\Lao}_{n,l}(p_1, \dots, p_n)$, are supported
at $p_1+\dots+p_n=0$.
We consequently write, by
abuse of notation
\ben
\L^{\Lambda,\Lambda_0}_{n,l}(p_1, \dots, p_n) = \delta^{4}{(\sum_{i=1}^n
p_i)}
\, \L^{\Lambda,\Lambda_0}_{n,l}(p_1, \dots, p_{n-1}) \, ,
\een
i.e. one of the momenta is determined in terms of the remaining $n-1$
independent momenta by momentum conservation. It is straightforward to
see that, as
functions of these remaining independent momenta, the connected
amputated
Green functions
are smooth for $\La_{0}<\infty$, $\L^{\Lambda,\Lambda_0}_{n,l}(p_1, \dots, p_{n-1})
\in C^\infty(\mr^{4(n-1)})$.

The FEs are
obtained by taking a $\La$-derivative of
eq.\eqref{CAGdef}:
\ben
\partial_{\La} L^{\La,\Lao} \,=\,
\frac{\hbar}{2}\,
\langle\frac{\delta}{\delta \vp},\dot {C}^{\La}\star
\frac{\delta}{\delta \vp}\rangle L^{\La,\Lao}
\,-\,
\frac{1}{2}\, \langle \frac{\delta}{\delta
  \vp} L^{\La,\Lao},
\dot {C}^{\La}\star
\frac{\delta}{\delta \vp} L^{\La,\Lao}\rangle  +
\hbar \partial_\Lambda \log Z^{\La,\Lao} \ .
\label{fe}
\een
Here we use the shorthand $\,\dot {C}^{\La}\,$ for
$\partial_{\La} {C}^{\La,\Lao}\,$, which, as we note,
does not depend on $\Lao$.
By $\langle\ ,\  \rangle$ we denote the standard scalar product in
$L^2(\mathbb{R}^4, d^4 x)\,$, and $\star$ denotes
convolution in $\mr^4$. For example
\ben
\langle\frac{\delta}{\delta \vp},\dot {C}^{\La}\star
\frac{\delta}{\delta \vp}\rangle = \int d^4x d^4y \ \dot {C}^{\La}(x-y)
\frac{\delta}{\delta \varphi(x)} \frac{\delta}{\delta \varphi(y)}
\een
is the ``functional Laplace operator''.

To define the  CAG's through the FEs,
we have to impose  boundary
conditions. These are, using
the multi-index convention introduced above in ``Notations and Conventions'' (we restrict to BPHZ renormalization
conditions in their simplest form, more general choices are of
course equally admissible):
\ben
\partial^w_{\vec{p}} \L^{0,\Lambda_0}_{n,l}(\vec 0) = \de_{w,0}\ \de_{n,4}\
 \de_{l,0}\ \frac{g}{4!} \quad \text{for $n+|w|
\le 4$,}
\een
 as well as
\ben
\partial^w_{\vec{p}} \L^{\Lambda_0,\Lambda_0}_{n,l}(\vec p) = 0 \quad \text{for
$n+|w| > 4$.}
\een
The  CAG's are then determined by integrating the
FE's subject to these boundary conditions, see e.g.~\cite{Keller:1990ej,Muller:2002he}.

\subsection{Insertions of composite fields}

For the purposes of this paper, and also in many applications,
one would like to define not only
Schwinger functions of products of the basic field, but also ones
containing composite operators.
These are obtained by replacing the action $L^{\Lambda_0}$ with an
action containing additional sources, expressed through smooth functionals.
Particular examples of such functionals are {\em local} ones. Any
such local functional can by definition be
written as
\ben
F (\varphi)= \sum_A \int d^4 x \  \O_A(x) \ f^A(x) \,\, ,
\quad f^A \in C^\infty_0(\mr^4) \, ,
\een
where $\O_A$ are composite operators as in eq.~\eqref{compop} and
where the sum is finite. To simplify our discussion below, we
will restrict attention to composite fields~\eqref{compfields} with
an even number of factors of $\varphi$ (note that these are the
fields which are invariant under the $\mathbb{Z}_2$-symmetry $\varphi \to -\varphi$,
hence the OPE closes in this sector).
We now consider instead of $L^{\Lambda_0}$ a modified action
containing
sources $f^A$, given by
replacing
\ben
L^{\Lambda_0}\to L^{\Lao}_F:=L^{\Lambda_0}+ F + \sum_{j=0}^\infty
B^{\Lambda_0}_j(\underbrace{F \otimes \cdots \otimes F}_j) \, ,
\een
where the last term represents the counter terms and is for each $j$ a
suitable linear functional
\ben
B_j^{\Lambda_0}: [C^\infty(\S(\mr^4))]^{\otimes j}
\to C^\infty(\S(\mr^4))
 \  ,
\een
that is symmetric, and of order $O(\hbar)$.
These counter terms
are designed to eliminate the divergences arising from composite field
insertions in the
Schwinger functions when one takes $\Lambda_0 \to \infty$.
The Schwinger functions with insertions of $r$~composite operators are
defined with the aid of  functional
derivatives with respect to the sources, setting the sources
$f^{A_i} =0$ afterwards:
\ben
\langle \O_{A_1}(x_1) \cdots \O_{A_r}(x_r)  \rangle :=
\een
\[
\hbar^r
\frac{\delta^r}{\delta f^{A_1}(x_1) \dots \delta f^{A_r}(x_r)}
\ (Z^{\La,\Lao})^{-1} \int d\mu^{\Lambda,\Lambda_0}
\exp \bigg(-\frac{1}{\hbar}  L^{\Lambda_0}_F(\varphi)
\bigg)\biggr|_{ f^{A_i}=0} \, .
\]
The previous definition of the CAG's is a special case
of this; there we
take $F = \int d^4 x \ f(x) \ \varphi(x)$, and we have
$B^{\Lambda_0}_j(F^{\otimes j})=0$, because no extra counter terms
are required for this simple insertion. As above, we can define
a corresponding effective action
as
\ben\label{genfuncins}
-L^{\Lambda,\Lambda_0}_F := \hbar \, \log \, \mu^{\Lambda,\Lambda_0}
\star \exp \bigg(-\frac{1}{\hbar} ( L^{\Lambda_0}
+ F + \sum_{j=0}^\infty B^{\Lambda_0}_j(F^{\otimes j})) \bigg)
- \log Z^{\La,\Lao}
\een
which is now a functional of the sources $f^{A_i}$, as well as of $\varphi$.
Differentiating $r$ times with respect to the sources, and setting them
to zero afterwards, gives the generating functionals of the CAG's with
$r$ operator insertions, namely:
\ben
L^{\Lambda,\Lambda_0}(\O_{A_1}(x_1) \otimes \dots \otimes \O_{A_r}(x_r))
=
\frac{\delta^r \ L^{\Lambda,\Lambda_0}_F}{\delta f^{A_1}(x_1) \dots \delta
  f^{A_r}(x_r)}
\,  \Bigg|_{f^{A_i} =  0} \, .
\een
The CAG's with insertions satisfy a number of obvious properties,
e.g.
they are multi-linear--as indicated by the
tensor product notation--and symmetric in the insertions.

As the CAG's without insertions, the CAG's with insertions can be further
expanded in $\varphi$ and $\hbar$, and this is denoted as
\ben
L^{\Lambda,\Lambda_0} \bigg( \bigotimes_{i=1}^r \O_{A_i}(x_i) \bigg) =
\sum_{n,l \ge 0} {\hbar^l} \int d^4p_1\dots d^4p_n \
\L^{\Lambda,\Lambda_0}_{n,l}\bigg( \bigotimes_{i=1}^r \O_{A_i}(x_i);
p_1,
\dots, p_n \bigg)
\prod_{j=1}^n \vp(p_j) \, .\non
\een
Due to the insertions in $\L_{n,l}^{\La,\Lao} (\otimes_j
\O_{A_j}(x_j),
\vec p)$, there is no restriction on the momentum set
$\vec p$. However it follows from translation invariance that functions
with insertions at a translated set of points $x_j + y\,$ are
obtained
from those
at $y=0\,$ upon multiplication by $\e^{iy\sum_{i=1}^n p_i}$, i.e.
\ben\label{CAGtrans}
\L^{\Lambda,\Lambda_0}_{n,l}\bigg( \bigotimes_{i=1}^r \O_{A_i}(x_i+y);
p_1,
\dots, p_n \bigg)=
\e^{iy\sum_{i=1}^n p_i}\, \L^{\Lambda,\Lambda_0}_{n,l}\bigg( \bigotimes_{i=1}^r \O_{A_i}(x_i);
p_1,
\dots, p_n \bigg)\, .
\een
Also, since we only consider composite fields with even powers of $\varphi$,
only moments with $n$ an even number are non-zero.
The CAG's with insertions satisfy a FE similar in nature to the one for CAG's without insertions:
\ben\label{FEN}
\begin{split}
\partial_{\Lambda}L^{\Lambda,\Lambda_{0}}(\bigotimes_{i=1}^{N}\O_{A_{i}})=&\frac{\hbar}{2}\bra \varp \, ,\, \dot{C}\star \varp  \ket\, L^{\Lambda,\Lambda_{0}}(\bigotimes_{i=1}^{N}\O_{A_{i}})\\
-&
\sum_{I_1 \cup I_2 = \{1,...,N \} }\bra \varp  L^{\Lambda,\Lambda_{0}}(\bigotimes_{i\in I_{1}}\O_{A_{i}}) \, ,\, \dot{C}\star \varp L^{\Lambda,\Lambda_{0}}(\bigotimes_{j\in I_{2}}\O_{A_{j}}) \ket\quad ,
\end{split}
\een
Here we suppressed the coordinate space variables $(x_{1},\dots, x_{N})$ and simply wrote $\O_{A_{i}}$ instead of $\O_{A_{i}}(x_{i})$. We will often use this convention in the following for the sake of brevity.
When expanded out in $\varphi$, this equation can be expressed in momentum space as
\ben\label{FEexpand}
\begin{split}
&\partial_{\Lambda}\L^{\Lambda,\Lambda_{0}}_{2n,l}(\bigotimes_{i=1}^{N}\O_{A_{i}}; p_{1},\ldots,p_{2n})= \left(\atop{2n+2}{2}\right) \, \int_{k} \dot{C}^{\Lambda}(k)\L^{\Lambda,\Lambda_{0}}_{2n+2,l-1}(\bigotimes_{i=1}^{N}\O_{A_{i}}; k, -k,  p_{1},\ldots,p_{2n})\\
&-4\sum_{\atop{l_{1}+l_{2}=l}{n_{1}+n_{2}=n+1}} n_{1}n_{2}\,  \mathbb{S}\, \Bigg[  \L^{\Lambda,\Lambda_{0}}_{2n_{1},l_{1}}(\bigotimes_{i=1}^{N}\O_{A_{i}}; q,  p_{1},\ldots,p_{2n_{1}-1})\,  \dot{C}^{\Lambda}(q)\,   \L^{\Lambda,\Lambda_{0}}_{2n_{2},l_{2}}(  p_{2n_{1}},\ldots,p_{2n}) \\
&+
\sum_{\atop{I_1 \cup I_2 = \{1,...,N \}}{I_{1} \neq \emptyset\neq I_{2} } } \int_{k}\, \L^{\Lambda,\Lambda_{0}}_{2n_{1},l_{1}}(\otimesu{i\in I_{1}}\O_{A_{i}}; k,  p_{1},\ldots,p_{2n_{1}-1})\,  \dot{C}^{\Lambda}(k)\,   \L^{\Lambda,\Lambda_{0}}_{2n_{2},l_{2}}(\otimesu{j\in I_{2}}\O_{A_{j}}; -k,   p_{2n_{1}},\ldots,p_{2n}) \Bigg]
\end{split}
\een
with $q=p_{2n_{1}}+\ldots+p_{2n}$ and where $\mathbb{S}$ is the symmetrization operator acting on functions of the momenta $(p_{1},\ldots, p_{2n})$ by taking the mean value over all permutations $\pi$ of $1,\ldots, 2n$ satisfying $\pi(1)<\pi(2)<\ldots<\pi(2n_{1}-1)$ and $\pi(2n_{1})<\ldots<\pi(2n)$.
It is crucial to note that the FE for the CAG's with $N\geq 2$ insertions is not
linear homogeneous, but involves a ``source term''
which is quadratic in the CAG's with less than $N$ insertions. If we want to
integrate the FEs with insertions, we therefore have to ascent in the number
of insertions.
The CAG's with insertions are thereby uniquely defined once we impose suitable boundary conditions on the corresponding FE. The simplest choice in the case of $N\geq 2$ insertions is
\ben\label{BCunsub}
\partial^{w}_{\vec{p}}\LscIr_{n,l}(\bigotimes_{i=1}^{N}\O_{A_{i}}(x_{i}); \vec{p})=0\quad \text{for all } w,n,l.
\quad
\een
For CAG's with one insertion we choose (''normal ordering'')
\ben\label{BCL1}
\partial^{w}_{\vec{p}}\LscO_{n,l}(\O_{A}(0); \vec{0})= i^{|w|}w! \delta_{w,w'}\delta_{n,n'}\delta_{l,0} \quad \text{ for }n+|w|\leq [A]
\een
\ben\label{BCL2}
\partial^{w}_{\vec{p}}\LscIr_{n,l}(\O_{A}(0); \vec{p})=0\quad \text{ for }n+|w|>[A] \quad .
\een

\subsection{Regularized CAG's with $N\ge 2$ insertions}\label{sec:regCAG}

It is not hard to see that, as long as we keep the UV cutoff $\Lao$ finite, the CAG's with insertions depend smoothly on the points
$x_1, \dots, x_N$, as well as on the momenta $p_1, \dots, p_n$. In the limit as $\Lao \to \infty$, smoothness in
the $x_i$'s however is lost, and the CAG's develop singularities for configurations such that some of the points $x_i$
coincide.  This is of course not a problem, nor unexpected--the Greens functions in QFT are usually singular for coinciding points--reflecting the singular nature of the operators themselves.

One of the main technical advances of this paper, generalizing methods of \cite{Hollands:2011gf, Keller:1992by}, is to consider, besides the original CAG's, certain ``regularized'' (sometimes called ``oversubtracted'') versions thereof, which possess more regularity
in the $x_i$'s as $\Lao \to \infty$, i.e. which remove, in a certain sense, some or all of the singularities at coincident points.  In the case of $N=2$ insertions, these divergences can be removed by simply changing the boundary conditions for the CAG's \cite{Keller:1992by}. These regularized CAG's are parametrized by a single integer parameter $D\geq -1$ and defined by the same FE, eq.\eqref{FEN}, but subject to the boundary conditions
\ben\label{BC2rel}
\partial^{w}_{\vec{p}}\L^{0,\Lambda_{0}}_{D,n,l}(\O_{A_{1}}(x)\otimes\O_{A_{2}}(0); \vec{0})=0\quad \text{ for } n+|w|\leq D
\een
as well as
\ben\label{BC2irrel}
\partial^{w}_{\vec{p}}\LscIr_{D,n,l}(\O_{A_{1}}(x)\otimes\O_{A_{2}}(0); \vec{p})=0\quad \text{ for } n+|w|> D
\een
instead of eq.\eqref{BCunsub}. It has been shown in \cite{Keller:1992by} that these functions are of differentiability class $C^{D-[A_{1}]-[A_{2}]}$ in $x$. This justifies the interpretation of $D$ as a regularization parameter. For more than two insertions, one has a choice for which subset of the arguments $x_1, \dots, x_N$ one wants to remove the singularities, and up to what degree. In order to specify this, we consider all subsets $I \subset \{1, \dots, N\}$ with $|I| \ge 2$ and define the corresponding ``partial diagonal'' by
\ben
{\rm Diag}_I = \{(x_1, \dots, x_N) \in (\mr^4)^N \mid x_i=x_j \ \ \ {\rm for} \ \ \ i,j \in I\} \ .
\een
The CAG's with insertions $L^{\La,\Lao}(\otimes_{i=1}^N \O_{A_i}(x_i))$, or the versions thereof which are expanded in
$\varphi$, see eq.~\eqref{FEexpand}, develop singularities at those diagonals in the limit as $\Lao \to \infty$.
Next, we specify a collection $\I $ of subsets of $\{1, \dots, N\}$. We agree that $\{1, \dots, N\}$ is always in $\I$. For each corresponding partial diagonal an integer $D_{I} \ge -1$ is assigned which specifies the degree of regularization of the CAG at that diagonal. The tuple of all such numbers is denoted
\ben
\bfD=( D_I \mid I \in \I )\, \in (\mathbb{Z}_{\geq -1})^{| \I  |}
\een
and given such a tuple, one defines the corresponding regulated CAG's by the FE:
\ben\label{FE3reg}
\begin{split}
\partial_{\Lambda}L^{\Lambda,\Lambda_{0}}_{\bfD}(\otimes_{i=1}^N\O_{A_{i}})=&\frac{\hbar}{2}\bra \varp \, ,\, \dot{C}\star \varp  \ket\, L^{\Lambda,\Lambda_{0}}_{\bfD}(\otimes_{i=1}^N \O_{A_{i}})- \bra \varp L^{\Lambda,\Lambda_{0}}_{\bfD}(\otimes_{i=1}^N \O_{A_{i}})  \, ,\, \dot{C}\star \varp L^{\Lambda,\Lambda_{0}} \ket \\
-&\sum_{\atop{I_1 \cup I_2 = \{1, ..., N\}} {I_{i}\in \I \text{ or }|I_{i}|=1 } } \bra \varp L^{\Lambda,\Lambda_{0}}_{\bfD_1}(\otimes_{i \in I_1} \O_{A_{i}})  \, ,\, \dot{C}\star \varp L^{\Lambda,\Lambda_{0}}_{\bfD_2} (\otimes_{j \in I_2} \O_{A_j} ) \ket \ .
\end{split}
\een
Here, we have set $\bfD_i = \{D_J \mid J \in \I, J \subseteq I_i\}$ for $i=1,2$.
When $|I_1|=1$ (or when $|I_2|=1$) then the corresponding $D$-symbol on $L^{\La,\Lao}(\O_{A_i})$ is absent--the 1-point functions are already smooth in $x_i$.
The boundary conditions are taken as
\ben\label{BC3a}
\partial^{w}_{\vec{p}}\L^{0,\Lambda_{0}}_{\bfD}(\O_{A_{1}}(x_{1})\otimes\cdots \otimes\O_{A_{N}}(0); \vec{0})=0\quad \text{ for } n+|w|\leq D_{\{1, \dots, N\}}
\een
and
\ben\label{BC3b}
\partial^{w}_{\vec{p}}\L^{\Lambda_{0},\Lambda_0}_{\bfD}(\O_{A_{1}}(x_{1})\otimes\cdots\otimes\O_{A_{N}}(0); \vec{p})=0\quad \text{ for } n+|w|> D_{\{1,\dots,N\}} .
\een
Our definition of the regulated CAG's is recursive in nature; first those with $N=2$ insertions are defined,
then those with $N=3$, etc.

The regularized multi-point CAG's are the main new technical tool in this paper. For $N=2$, we only have to specify the value of $D_{\{1,2\}}$, as $I=\{1,2\}$ is the only subset of $\{1,2\}$ with at least
2 elements. In that case, the definition of the regularized two-point CAG hence depends on the specification of only
one number, and it reduces to the one already given above in eqs.\eqref{BC2rel} and \eqref{BC2irrel}. In this paper, we will only need additionally the case of
$N=3$ points. In that case, $\I$ may be any collection of the subsets $I=\{1,2,3\}, \{1,2\}, \{2,3\}, \{1,3\}$, and the corresponding
numbers $D_I$ control the degree of regularity of the CAG $L_{\bf D}^{\La,\Lao}(\O_{A_1}(x_1)
\otimes \O_{A_2}(x_2) \otimes \O_{A_3}(x_3))$ at the respective diagonals $\{x_1=x_2=x_3\}, \{x_1=x_2\}, \{x_2=x_3\}, \{x_1=x_3\}$.
Of particular interest to us will be the following special cases:

\begin{itemize}

\item The collection $\I=\{\{1,2,3\}, \{1,2\}, \{2,3\}, \{1,3\}\}$ together with the choice $D_{\{1,2,3\}}=D_{\{1,2\}}=D_{\{1,3\}}=D_{\{2,3\}}=-1$ corresponds to no regularization at all. We will  denote these CAG's without regularization as before simply by $L^{\Lambda,\Lambda_{0}}(\bigotimes_{i=1}^{3}\O_{A_{i}})$.

\item The collection $\I=\{\{1,2,3\}, \{1,2\}, \{2,3\}, \{1,3\}\}$ together with the choice $D_{I}=\sum_{i\in I} [A_{i}]$ for all $I \in \I$ can be shown to be regular (continuous) at all partial diagonals, uniformly as $\Lao \to \infty$.

\item The collection $\I=\{\{1,2,3\}, \{1,2\}, \{2,3\}, \{1,3\}\}$ together with the choice $D_{\{1,2,3\}}=D_{\{1,2\}}=D_{\{1,3\}}=-1, D_{\{2,3\}}>-1$ corresponds to no regularization at all except at the diagonal $\{x_2=x_3\}$. These objects appear when analysing the OPE of two fields localized
    at $x_2, x_3$ in a Schwinger function of three fields, with the field at $x_1$ playing the role of a ``spectator''.

\item It is not necessary that $\I$ contains all subsets of $\{1,2,3\}$ (of order $>1$), as in the above three examples. The sum in the defining FE then has fewer terms, as described above. We will also consider these objects.
\end{itemize}
As we will later be interested in spacetime derivatives of the CAG's, the following relations will be useful.
They are generalizations of the \emph{Lowenstein rules}, given e.g. in \cite{Keller:1991bz,Keller:1992by}.
For $N=1$ insertion, the rule is simply that
\ben\label{loew1}
\partial_{x}^{w}\, L^{\Lambda,\Lambda_{0}}(\O_{A}(x))= L^{\Lambda,\Lambda_{0}}(\partial^{w}_x \O_{A}(x))
\een
where $\partial^w_x$ is a multiple partial derivative acting on $x$, see our notations and conventions for
more on multi-index notation. Note that the identity is (slightly) non-trivial, because the $\partial^w_x$
symbol on the left side acts on the CAG, while on the right side it acts on the operator $\O_A$, i.e.
a monomial in $\varphi$ and its derivatives, which is then inserted into a CAG. So, a priori, both sides mean
potentially different things, but they are in fact equal. A generalization also holds for CAG's with multiple
insertions.
\ben\label{loew2}
\partial_{x_{j}}^{w}\, L^{\Lambda,\Lambda_{0}}_{\bfD}(\bigotimes_{i=1}^{N}\O_{A_{i}}(x_{i}))= L^{\Lambda,\Lambda_{0}}_{\bfD}(\O_{A_1}(x_1) \otimes \cdots \partial^w_{x_j} \O_{A_j}(x_j) \otimes
\cdots \O_{A_N}(x_N))
\een
Derivatives in directions of the ``center of mass'' on the total diagonal, or a partial diagonal associated
with a subset $I \subset \{1, \dots, N\}$, have a special status. For a multiindex $v$ such a
derivative is given by $(\sum_{i \in I} \partial_{x_i})^v$
\ben\label{loew3}
\begin{split}
(\sum_{j \in I}\partial_{x_{j}})^{v}\, L^{\Lambda,\Lambda_{0}}_{\bfD}(\bigotimes_{i=1}^{N}\O_{A_{i}}(x_{i}))
=&\sum_{\sum_{i \in I} w_i = v} c_{\{w_{i}\}} \  L^{\Lambda,\Lambda_{0}}_{\bfD(I,\{w_i\})}( \ \prod_{i \in I}
\partial_{x_i}^{w_i} \ \otimes_{i=1}^{N}\O_{A_{i}}(x_{i}) \ ) \ .
\end{split}
\een
Here $\bfD(I,\{w_i\})=\{D_J' \mid J \in \I\}$ is defined as
\ben
D'_J =
\begin{cases}
D_J & \text{if $J$ is not a subset of $I$},\\
D_J + \sum_{j \in J} |w_j| & \text{if $J$ is a subset of $I$.}
\end{cases}
\een
The symbols $c_{\{w_i\}}$ are binomial type coefficients, see our notations and conventions
section. The last condition is most straightforward to understand when $I = \{1, \dots, N\}$. The point is that
the CAG is always smooth in the center of mass variable $x_1+\dots+x_N$, so derivatives with respect to this variable
never make the divergences of the CAG's worse. The degree of regularity of the right side is hence the same as that
without the derivatives taken, which, due to the higher dimensions of the operator insertions, results in a
higher regularity degree $D_{\{1,...,N\}}$ on the total diagonal.

All of the generalized Lowenstein rules, although complicated looking, can be proved via a simple but powerful
{\em principle in the FE method,} which we will use several times over in this article. The general principle simply says that if we have two hierarchies of functions satisfying the same {\em FE} and same {\em boundary condition}, then
they must coincide. In the case of the Lowenstein identities, it is easy to find what FE's each side has to satisfy, and what boundary conditions. For each identity, these are easily seen to coincide, thus proving the identity.

\subsection{Amputated Greens functions (AG's) with insertions}

Although the connected amputated Greens functions (CAG's) are the basic building blocks of the correlation functions
which will be mostly considered, we will occasionally also need the non-connected version of these, called AG's.
Their definition is
\ben\label{GD}
G^{\Lambda,\Lambda_{0}}(\bigotimes_{i=1}^{N}\O_{A_{i}})
=\sum_{\alpha=1}^{N} (-1)^{\alpha+1}\sum_{ I_1 \cup ... \cup I_\alpha = \{1,...,N\} }\prod_{i=1}^{\alpha}\,\hbar^{N-\alpha} L^{\Lambda, \Lambda_{0}}(\bigotimes_{j\in I_i}\O_{A_{j}})
\een
As usual, we also consider the expanded quantities in $\hbar$ and $\varphi$; these are denoted in
the present case as ${\mathcal G}^{\La,\Lao}_{n,l}(\otimes_{i=1}^N \O_{A_i}, \vec p)$, where
as usual, $l$ indicates the power of $\hbar$, and $n$ the power of $\varphi$. As the name suggests, these are the
amputated versions of the Schwinger (=Greens) functions,
\ben
\begin{split}
&\Big\bra\prod_{i=1}^N \O_{A_i}(x_i) \  \prod_{j=1}^n \hat \varphi(p_i) \Big\ket_{l  \ {\rm loops}}\ \prod_{k=1}^n (C^{\Lambda,\Lambda_{0}}(p_{k}))^{-1}
\\
&= \sum_{j=1}^{n} (-1)^{j+1}\sum_{\substack{I_{1}\cup\ldots\cup I_{j}=\{1,\ldots, n\} \\ l_{1}+\ldots+l_{j}=l }}\, \hbar^{n+l+1-j}\,  {\mathcal G}^{\La,\Lao}_{|I_{1}|,l_{1}}(\otimes_{i=1}^N \O_{A_i}(x_i), \vec p_{I_{1}}) \, \bar{\L}_{|I_{2}|,l_{2}}^{\Lambda,\Lambda_{0}}(\vec{p}_{I_{2}}) \cdots  \bar{\L}_{|I_{j}|,l_{j}}^{\Lambda,\Lambda_{0}}(\vec{p}_{I_{j}})
\end{split}
\een
where $\bar\L^{\La,\Lambda_{0}}_{n,l}$ are the expansion coefficients of the generating functional $\bar{L}^{\La,\Lambda_{0}}(\varphi)=L^{\La,\Lambda_{0}}(\varphi)+\frac{1}{2}\bra \varphi,\, (C^{\La,\Lambda_{0}})^{-1}\star\varphi \ket$ without the momentum conservation delta functions taken out. We will use this relation later.

By contrast to the CAG's, the AG's satisfy linear {\em homogeneous} FE's
which are
\ben\label{GFE}
\partial_{\Lambda}G^{\Lambda,\Lambda_{0}}(\bigotimes_{i=1}^{N}\O_{A_{i}})=\frac{\hbar}{2}\bra \varp , \dot{C}\star\varp \ket\, G^{\Lambda,\Lambda_{0}}(\bigotimes_{i=1}^{N}\O_{A_{i}})- \frac{\hbar}{2}\bra \varp G^{\Lambda,\Lambda_{0}}(\bigotimes_{i=1}^{N}\O_{A_{i}}), \dot{C}\star\varp L^{\Lambda,\Lambda_{0}} \ket \ .
\een
The fact that the AG's satisfy a homogeneous FE is a welcome simplification, which is unfortunately
counterbalanced by the fact that the boundary conditions for the AG's are more non-trivial.
Therefore, as a compromise between simple FE and simple boundary conditions,  we will not work with the full AG's in the following, but instead define the slightly modified objects
 \ben\label{GTild}
F^{\Lambda,\Lambda_{0}}(\bigotimes_{i=1}^{N}\O_{A_{i}}):=G^{\Lambda,\Lambda_{0}}(\bigotimes_{i=1}^{N}\O_{A_{i}})+ (-1)^{N}\prod_{i=1}^{N} L^{\Lambda,\Lambda_{0}}(\O_{A_{i}})\, .
\een
Using the definitions of the CAG's given above, these functionals are seen to obey the FE
\ben\label{GTildFE}
\begin{split}
&\partial_{\Lambda}F^{\Lambda,\Lambda_{0}}(\bigotimes_{i=1}^{N}\O_{A_{i}})\\
=&\frac{\hbar}{2}\bra \varp , \dot{C}\star\varp \ket\, F^{\Lambda,\Lambda_{0}}(\bigotimes_{i=1}^{N}\O_{A_{i}})- \frac{\hbar}{2}\bra \varp F^{\Lambda,\Lambda_{0}}(\bigotimes_{i=1}^{N}\O_{A_{i}}), \dot{C}\star\varp L^{\Lambda,\Lambda_{0}} \ket\\
-&(-1)^{N}\hbar\sum_{i\neq j \in\{1,\ldots, N\}} \bra \varp L^{\Lambda,\Lambda_{0}}(\O_{A_{i}}), \dot{C}^{\Lambda}\star\varp L^{\Lambda,\Lambda_{0}}(\O_{A_{j}})   \ket \prod_{r\in\{1,\ldots,N\}\setminus\{i,j\}} L^{\Lambda,\Lambda_{0}}(\O_{A_{r}})
\end{split}
\een
and the trivial boundary conditions
\ben
\partial^{w}_{\vec{p}} \F^{\Lambda_{0},\Lambda_{0}}_{n,l}(\bigotimes_{i=1}^{N}\O_{A_{i}}; \vec{p})=0\quad \text{for all }n,l,w,
\een
with a calligraphic letter $\F^{\La,\Lao}_{n,l}$ denoting as usual the objects appearing in the expansion
of $F^{\La,\Lao}$ in powers of $\hbar, \varphi$.
Like the CAG's with multiple insertions, the $F$ functionals are divergent on the partial diagonals, i.e. whenever two or more spacetime arguments coincide. To analyze the OPE we also need regularized $F$ functionals, which will be called $F_{D}$. They are defined to satisfy the same FE as $F$, eq.\eqref{GTildFE}, but the boundary conditions are set to be
\ben\label{Gbound1}
\partial^{w}_{\vec{p}}\F^{0,\Lambda_{0}}_{n,l,D}(\bigotimes_{i=1}^{N}\O_{A_{i}}; \vec{0})\Big|_{x_{N}=0}=0\quad \text{ for }n+|w|\leq D
\een
\ben\label{Gbound2}
\partial^{w}_{\vec{p}}\F^{\Lambda_{0},\Lambda_{0}}_{n,l,D}(\bigotimes_{i=1}^{N}\O_{A_{i}}; \vec{p})\Big|_{x_{N}=0}=0\quad \text{ for }n+|w|>D\, .
\een
Note that in the $N=2$ case $F_{D}$ reduces to the (regularized) bilocal insertions, i.e.
\ben
F^{\Lambda,\Lambda_{0}}_{D}(\O_{A}(x)\otimes\O_{B}(0))=\hbar L^{\Lambda,\Lambda_{0}}_{D}(\O_{A}(x)\otimes\O_{B}(0))
\een
since both share the same FE and boundary conditions. For $N\geq 3$, however, such a simple relation does not exist.
We also define the complete regularized AG's as
 \ben\label{GDdef}
{G}_{D}^{\Lambda,\Lambda_{0}}(\bigotimes_{i=1}^{N}\O_{A_{i}}):=
F_{D}^{\Lambda,\Lambda_{0}}(\bigotimes_{i=1}^{N}\O_{A_{i}})- (-1)^{N}\prod_{i=1}^{N} L^{\Lambda,\Lambda_{0}}(\O_{A_{i}})
\een
since the CAG's with single insertion are smooth without any regularization.

\subsection{OPE coefficients}\label{sec:OPEcoefs}

We next give the definition of the  OPE coefficients. To have more compact formulas, let us define the operator $\D^{A}$ acting on differentiable functionals $F(\varphi)$ of Schwartz space functions $\varphi\in\S(\mathbb{R}^{4})$ by
\ben\label{defD}
\D^{A} F(\varphi) = \left. \frac{1}{n!\,w!}\, \partial_{\vec{p}}^{w}\frac{\delta^{n}}{\delta\hat\varphi(p_{1})\cdots\delta\hat\varphi(p_{n})}\, F(\varphi)\, \right|_{\hat\varphi=0, \vec{p}=0}\quad ,
\een
where $A=\{n,w\}$. Further, let us also define the multivariate Taylor expansion operator 
through
\ben\label{defT}
\T^{j}_{\vec{x}\to\vec{y}}\, f(\vec{x})=\T^{j}_{(x_{1},\ldots,x_{N})\to(y_{1},\dots, y_{N})} \, f(x_{1},\ldots, x_{N})=\sum_{|w|=j}\, \frac{(\vec{x}-\vec{y})^{w}}{w!}\, \partial^{w}f(\vec{y})
\een
where $\vec{x}=(x_{1},\ldots,x_{N})$ and
where $f$ is a sufficiently smooth function on $\mathbb{R}^{4N}$. For expansions around zero will use the shorthand $\T^{j}_{\vec{x}\to\vec{0}}=:\T^{j}_{\vec{x}}$\ .
Then the OPE coefficients are defined as follows:
\begin{defn}\label{defOPE}
Let $\Delta:=[C]-([A_{1}]+\ldots+[A_{N}])$. Then we define the OPE coefficients by
\ben\label{OPEhigh}
\C_{A_{1}\ldots A_{N}}^{C}(x_{1},\ldots,x_{N-1},0)\,
:=\, \D^{C}\left\{ G^{0,\Lambda_{0}}_{[C]-1}\left((1-\sum_{j < \Delta}\T^{j}_{\vec{x}}\,) \bigotimes_{i=1}^{N}\O_{A_{i}}(x_i)\right) \right\}\, ,
\een
 where it is understood that $x_N=0$.
\end{defn}
\noindent
{\bf Remark:} In the case $N=2$ this definition is equivalent to the one given in \cite{Hollands:2011gf}. Note
also that the OPE coefficients are translation invariant, so we may e.g. put the last point
to zero by a translation, as we have done above to get a simpler formula.

\section{Bounds on Green's functions with insertions}\label{secBounds}

In the previous section, we have introduced the quantities of interest in this paper, the (connected, amputated)
Schwinger functions, and the OPE coefficients. In order to prove thm.~\ref{OPEbound} and thm.~\ref{thmfact}, we
will need suitable bounds for these objects. These bounds will be presented and proved in the following sections.
Our bounds are obtained using an inductive scheme based on the renormalization group FEs, relying heavily on the bounds already established in \cite{Hollands:2011gf} and also \cite{Kopper:1997vg,Keller:1991bz,Kopper:2009um}. For the sake of brevity we will refer the reader to the existing papers for technical details whenever possible. In the following we will also set $\hbar=g=1$.

\subsection{CAG's with up to one insertion}\label{sec:CAG1}

Bounds on CAG's without insertions and on those with one insertion were derived in~\cite{Kopper:2009um, Hollands:2011gf}. These bounds are a crucial input for the
subsequent bounds on CAG's with multiple insertions,  because the CAG's with one insertion
enter the FE for those with two insertions as ``inhomogeneities'' [cf. eq.~\eqref{FEN}], those with up to
two insertions enter the FE for those with three etc. In this way, many of the estimates will be seen to carry over from the case of one insertion  in that way.

Let us recall the bound for the CAG's without insertions first~\cite{Kopper:2009um, Hollands:2011gf}. There exists a constant $\,K>0\,$ such that for $2n+|w|\ge 5\,$  (recall also the definitions of $|\vec{p}|$ and $\kappa$ from our notations and conventions section above)
\ben
|\pa^{w}_{\vec{p}} {\cal L}_{2n,l}^{\La,\Lao}(p_1,\ldots,p_{n-1})|\leq \sqrt{|w|!}
\,\La^{4-2n-|w|}\
K^{(2n+4l-4)(|w|+1)}\ (n+l-2)! \
\sum_{\la=0}^{\ell(n,l)}
\frac{\log^{\lambda}(\sup(\frac{|\vec{p}|}{\kappa},\frac{\ka}{m}))}{2^{\lambda}\,\lambda!}
\ ,
\label{propout}
\een
where $\ell(n,l) = l\ \mbox{ if }\ n\ge 2$ and
$\ell(n,l) = l-1\ \mbox{ if }\  n=1\ $.
For $2n+|w| \le 4\,$ one has the estimates
\ben
  | {\cal L}_{4,l}^{\La,\Lao}(\vec{p})|\ \leq\
\frac{K^{2l}}{(l+1)^2\, 2^4}
\ (1+l)!\, \sum_{\la=0}^{l}
\frac{\log^{\lambda}\bigr(\sup({|\vec{p}|\over \ka},{\ka\over
    m})\bigl)}{2^{\lambda}\,\lambda!}\
\ ,
\label{prop40}
\een
\ben
  |\pa^{w}_{\vec{p}} {\cal L}_{2,l}^{\La,\Lao}(p)|\leq
\sup(|p|,\ka)^{2-|{w}|}\ \frac{K^{2l-1}}{(l+1)^2}
\ l!\,\sum_{\lambda=0}^{l-1}
\frac{\log^{\lambda}\bigr(\sup({|p| \over \ka},{\ka\over m})\bigl)}{2^{\lambda}\,\lambda!}
\ .
\label{prop20}
\een
We also recall the following bound for the CAG's with one insertion \cite{Hollands:2011gf}. Fix any $A=\{r, v\}$. Then
\ben\label{boundCAG1}
\begin{split}
|\partial^{w}_{\vec{p}}\L^{\Lambda,\Lambda_{0}}_{ 2n,l}(\O_{A}(0); \vec{p})|\leq\, & \Lambda^{[A]-2n-|w|}\, K^{(4n+8l-4)|w|}\, K^{[A](n+2l)^{3}}\, \sqrt{|w|!\, |v|!}  \\
\times & \sum_{\mu=0}^{{d}(N=1,n,l,w,[A])}\frac{1}{\sqrt{\mu!}}\left(\frac{|\vec{p}|}{\Lambda}\right)^{\mu}\, \sum_{\lambda=0}^{2l+n-1}\frac{\log^{\lambda}(\sup(\frac{|\vec{p}|}{\kappa}), \frac{\kappa}{m})}{2^{\lambda}\lambda!}\, ,
\end{split}
\een
where $K>0$ is some constant not depending on $n$ and $l$,  and where we defined $d(N,n,l,w,D'):= 2D'(n+l+2(N-1))+\sup(D'+1-2n-|w|, 0)$. Eventually we would of course like to remove the cutoffs, i.e. take the limits $\La\to 0$ and $\Lao\to\infty$. In this respect, the following bounds, which hold for $\Lambda\leq m$, will be useful \cite{Hollands:2011gf}
\ben\label{boundCAG0m}
\begin{split}
  |\pa^{w}_{\vec{p}} {\cal L}_{2n,l}^{\La,\Lao}(\vec{p})| &\leq\, m^{4-2n-|w|}\
\frac{K^{(2n+4l-4)(|w|+1)}}{n!} \ (n+l-1)! \\
& \times \ \sqrt{|w|!\,
(|w|+2n-4)!}\,
\sum_{\lambda=0}^{ l}
\frac{\log_+^{\lambda}(\frac{|\vec{p}|}{m})}{2^{\lambda}\,\lambda!}\quad \text{for }2n+|w|\geq 5
\end{split}
\een
\ben\label{boundCAG1m}
\begin{split}
  |\pa^{w}_{\vec{p}} {\cal L}_{2n,l}^{\La,\Lao}(\O_A(0);\vec{p})| \
&\leq
m^{[A]-2n -|w|}\
K^{(4n+8l-4)|w|}\
 K^{[A](n+2l)^3} \ \sqrt{|w|!\, |v|!}  \\
&\times\ \sqrt{[2n+|w| - [A]]_+!}
\sum_{\mu=0}^{d(N=1,n,l,w,[A])}
 \left(\frac{|\vec p|}{m}\right)^{\mu}\,
\sum_{\lambda=0}^{2l+n-1}
\frac{\log_+^{\lambda}(\frac{|\vec{p}|}{m})}{2^{\lambda}\,\lambda!}\
\, ,
\end{split}
\een
where by $[\, \cdot\,]_{+}$ we mean the positive part of the respective expression.

\subsection{CAG's with two insertions}

Our first new estimate is a generalization of the bound on the CAG's with two insertions given in~thm.~2 of~\cite{Hollands:2011gf}.
\begin{thm}\label{thmCAG2}
For any $D\leq D'=[A_{1}]+[A_{2}]$ there exists a constant $K>0$ such that
\ben\label{boundCAG2}
\begin{split}
&|\partial^{w}_{\vec{p}}\L^{\Lambda,\Lambda_{0}}_{D, 2n,l}(\O_{A_{1}}(x)\otimes\O_{A_{2}}(0); \vec{p})|\leq\,  \Lambda^{D-2n-|w|}\, K^{(4n+8l-3)(|w|+D'-D)}\, K^{D'(n+2l)^{3}}\, \sqrt{|v_1|!|v_2|!}\\
&\hspace{2cm} \times\ \frac{\sqrt{(|w|+D'-D)!}}{|x|^{D'-D}} \sum_{\mu=0}^{{d}(2,n,l,w,D')}\frac{1}{\sqrt{\mu!}}\left(\frac{|\vec{p}|}{\Lambda}\right)^{\mu}\, \sum_{\lambda=0}^{2l+n}\frac{\log^{\lambda}(\sup(\frac{|\vec{p}|}{\kappa}), \frac{\kappa}{m})}{2^{\lambda}\lambda!}
\end{split}
\een
with $d(N,n,l,w,D'):= 2D'(n+l+2(N-1))+\sup(D'+1-2n-|w|, 0)$. Here $v_i$ refers to
the multi-indices in $A_{i}\equiv \{n_{i},v_i\}$.
\end{thm}
\noindent
{\bf Remarks:}
\begin{enumerate}
\item  We can see explicitly that the parameter $D$ improves regularity at $x=0$.

\item The CAG's with two insertions have been bounded in \cite{Hollands:2011gf} for the particular choice $D=[A_{1}]+[A_{2}]$ (full regularization) and $|w|\leq D+1$.
\end{enumerate}
\begin{proof}
The strategy is to integrate the differentiated FE
\ben\label{FE2}
\begin{split}
&\hspace{2cm}\partial_{\Lambda}\partial_{\vec{p}}^{w}\L^{\Lambda,\Lambda_{0}}_{D,2n,l}(\O_{A_{1}}\otimes\O_{A_{2}}; p_{1},\ldots,p_{2n})=\\
\vspace{0.4cm}\\
=& \left(\atop{2n+2}{2}\right) \, \int_{k} \dot{C}^{\Lambda}(k)\ \partial_{\vec{p}}^{w}\L^{\Lambda,\Lambda_{0}}_{D,2n+2,l-1}(\O_{A_{1}}\otimes\O_{A_{2}}; k, -k,  p_{1},\ldots,p_{2n})-4\sum_{\substack{l_{1}+l_{2}=l \\ n_{1}+n_{2}=n+1}}\!\!\!\!\! n_{1}n_{2}\\
\times&\,  \mathbb{S}\, \Bigg[\sum_{w_{1}+w_{2}+w_{3}=w}\!\!\!\!\!\!\! c_{\{w_{j}\}}  \partial_{\vec{p}}^{w_{1}}\L^{\Lambda,\Lambda_{0}}_{D,2n_{1},l_{1}}(\O_{A_{1}}\otimes\O_{A_{2}}; q,  p_{1},\ldots,p_{2n_{1}-1})\,  \partial_{\vec{p}}^{w_{2}}\dot{C}^{\Lambda}(q)\,   \partial_{\vec{p}}^{w_{3}}\L^{\Lambda,\Lambda_{0}}_{2n_{2},l_{2}}(  p_{2n_{1}},\ldots,p_{2n}) \\
&\vspace{0.05cm}\\
+& \sum_{w_{1}+w_{2}=w} c_{\{w_{j}\}}\ \int_{k}\, \partial_{\vec{p}}^{w_{1}}\L^{\Lambda,\Lambda_{0}}_{2n_{1},l_{1}}(\O_{A_{1}}; k,  p_{1},\ldots,p_{2n_{1}-1})\,  \dot{C}^{\Lambda}(k)\,  \partial_{\vec{p}}^{w_{2}} \L^{\Lambda,\Lambda_{0}}_{2n_{2},l_{2}}(\O_{A_{2}}; -k,   p_{2n_{1}},\ldots,p_{2n}) \Bigg]
\end{split}
\een
over $\Lambda$ and bound each term on the right hand side separately. For the first two terms on the right hand side of this equation,  the bound is verified to hold inductively as one goes up in $n+l$ and for fixed $n+l$ goes up in $l$. This general procedure is very similar to the one employed in \cite{Hollands:2011gf}. To bound the expression in the last line of eq.\eqref{FE2}, we will make use of the known bounds on the CAG's with one insertion, \eqref{boundCAG1}.

When integrating eq.\eqref{FE2} over $\Lambda$ we have to distinguish three cases:

 \begin{enumerate}[(a)]

\item  Contributions with $2n+|w|> D$ are referred to as \emph{irrelevant}. Here the boundary conditions are given at $\Lambda=\Lambda_{0}$, see eq.\eqref{BC2irrel}. Therefore, we integrate over $\Lambda'$ from $\Lambda$ to $\Lambda_{0}$ in this case.

 \item  Contributions with $2n+|w|\leq D$ are referred to as \emph{relevant}.  The boundary conditions for relevant terms, eq.\eqref{BC2rel}, are given at $\Lambda=0$ and at vanishing external momentum, $\vec{p}=\vec{0}$. Thus, we will integrate over $\Lambda'$ from $0$ to $\Lambda$ in this case.

\item Contributions with $2n+|w|\leq D$ and $\vec{p}\neq \vec{0}$ will be obtained from (A),(B) with the help of a Taylor expansion in $\vec{p}$.
 \end{enumerate}
\textbf{(a) Irrelevant terms ($2n+|w|>D$):}

 \paragraph*{First term on the r.h.s. of the flow equation:}

Substituting our inductive bound, \eqref{boundCAG2}, and
\ben\label{cdotest}
\dot{C}^{\Lambda}(k)=-\frac{2}{\Lambda^{3}}\, \e^{-\frac{k^{2}+m^{2}}{\Lambda^{2}}}
\een
into the first term on the r.h.s. of eq.\eqref{FE2} and integrating over $\Lambda'$ from $\Lambda$ to $\Lambda_{0}$ yields the inequality [recall the definition of $|\vec{p}|_{2n+2}$ from eq.\eqref{pshort}]
\ben\label{FE21st}
\begin{split}
\Big|\int_{\Lambda}^{\Lambda_{0}}\d\Lambda' \left(\atop{2n+2}{2}\right)\int\d^{4}k \ \dot{C}^{\Lambda'}(k) \, \partial_{\vec{p}}^{w}\L_{D,2n+2,l-1}^{\Lambda',\Lambda_{0} }(\O_{A_{1}}\otimes\O_{A_{2}}; k,-k, p_{1},\ldots, p_{2n})\Big| &\\
\leq \int_{\Lambda}^{\Lambda_{0}}\d\Lambda'\,\left(\atop{2n+2}{2}\right)\, \int\d^{4} k\,\frac{2}{\Lambda'^{3}}\, \e^{-\frac{k^{2}+m^{2}}{\Lambda'^{2}}}  \Lambda'^{D-(2n+2)-|w|}\, K^{(4n+8l-7)(|w|+D'-D)}\, K^{D'(n+2l-1)^{3}}&\\
\times\sqrt{|v_1|!|v_2|!}\, \frac{\sqrt{(|w|+D'-D)!}}{|x|^{D'-D}} \sum_{\mu=0}^{{d}(2,n+1,l-1,w,D')}\frac{1}{\sqrt{\mu!}}\left(\frac{|\vec{p}|_{2n+2}}{\Lambda'}\right)^{\mu}\, \sum_{\lambda=0}^{2l+n-1}\frac{\log^{\lambda}(\sup(\frac{|\vec{p}|_{2n+2}}{\kappa'}), \frac{\kappa'}{m})}{2^{\lambda}\lambda!}&\\
\leq \, \left(\atop{2n+2}{2}\right)\,  K^{(4n+8l-7)(|w|+D'-D)}\, K^{D'(n+2l-1)^{3}}\, \sqrt{|v_1|!|v_2|!}\, \frac{\sqrt{(|w|+D'-D)!}}{|x|^{D'-D}} \sum_{\lambda=0}^{2l+n-1} \frac{2}{2^{\lambda}\lambda!}&\\
\times \int_{\Lambda}^{\Lambda_{0}}\!\!\!\d\Lambda'\,  \Lambda'^{D-2n-1-|w|} \e^{-\frac{m^{2}}{\Lambda'^{2}}} \sum_{\mu=0}^{{d}(2,n,l,w,D')}\frac{1}{\sqrt{\mu!}}\int\d^{4}(\frac{k}{\Lambda'})\left(\frac{|\vec{p}|_{2n+2}}{\Lambda'}\right)^{\mu}\, \log^{\lambda}(\sup(\frac{|\vec{p}|_{2n+2}}{\kappa'}), \frac{\kappa'}{m})\, \e^{-\frac{k^{2}}{\Lambda'^{2}}}&
\end{split}
\een
One can show \cite{Kopper:2009um, Hollands:2011gf} that the momentum integral in the last line is bounded by
\ben\label{kint}
\begin{split}
&\sum_{\mu=0}^{{d}}\frac{1}{\sqrt{\mu!}}\int\d^{4}(k/\Lambda')\left(\frac{|\vec{p}|_{2n+2}}{\Lambda'}\right)^{\mu}\, \log^{\lambda}(\sup(\frac{|\vec{p}|_{2n+2}}{\kappa'}), \frac{\kappa'}{m})\, \e^{-\frac{k^{2}}{\Lambda'^{2}}}\\
&\qquad \leq\, 2^{d}\, \sum_{\mu=0}^{{d}}\frac{1}{\sqrt{\mu!}} \left(\frac{|\vec{p}|}{\Lambda'}\right)^{\mu}\, \left(\log^{\lambda}(\sup(\frac{|\vec{p}|}{\kappa'}), \frac{\kappa'}{m}) + \sqrt{\lambda!} \right)\quad .
\end{split}
\een
The $\Lambda'$ integral can then be estimated using the formula \cite{Kopper:2009um}
\ben\label{Lambdaint}
\sum_{\lambda=0}^{2l+n-1} \frac{1}{2^{\lambda}\lambda!}\int_{\Lambda}^{\Lambda_{0}}\d\Lambda'\, \Lambda'^{-s-1}  \left(\log^{\lambda}(\sup(\frac{|\vec{p}|}{\kappa'}), \frac{\kappa'}{m}) + \sqrt{\lambda!} \right)\, \leq\, 5\frac{\Lambda^{-s}}{s}\sum_{\lambda=0}^{2l+n-1} \frac{1}{2^{\lambda}\lambda!}\log^{\lambda}(\sup(\frac{|\vec{p}|}{\kappa}), \frac{\kappa}{m})
\een
which holds for any $s\in\mathbb{N}$. Using \eqref{kint} and \eqref{Lambdaint} in formula \eqref{FE21st} and also noting the relation $(2l+n-1)^{3}=(2l+n)^{3}-3(2l+n)(2l+n-1)-1$, we find the bound
\begin{equation}
\begin{split}
&\Big|\int_{\Lambda}^{\Lambda_{0}}\d\Lambda' \left(\atop{2n+2}{2}\right)\int\d^{4}k \ \dot{C}^{\Lambda'}(k) \, \partial_{\vec{p}}^{w}\L_{D,2n+2,l-1}^{\Lambda',\Lambda_{0} }(\O_{A_{1}}\otimes\O_{A_{2}}; k,-k, p_{1},\ldots, p_{2n})\Big| \\
&\leq \Lambda^{D-2n-|w|}\, K^{D'(2l+n)^3}\, K^{(4n+8l-3)(|w|+D'-D)}\, \frac{\sqrt{|v_{1}|! |v_{2}|! (|w|+D'-D)!}}{|x|^{D'-D}}\\
&\times\, \sum_{\mu=0}^{d(2,n,l,w,D')}\frac{1}{\sqrt{\mu!}}\left(\frac{|\vec{p}|}{\Lambda}\right)^{\mu}\, \sum_{\lambda=0}^{2l+n-1}\frac{\log^{\lambda}(\sup(\frac{|\vec{p}|}{\kappa}, \frac{\kappa}{m}))}{2^{\lambda} \lambda!}\\
&\times\, K^{-3D'(2l+n)(2l+n-1)-D'-4(|w|+D'-D)}\, \left(\atop{2n+2}{2}\right)\, 2^{d(2,n,l,w,D')}\, \frac{10}{2n+|w|-D}
\end{split}
\end{equation}
We see that this contribution satisfies the inductive bound, \eqref{boundCAG2} multiplied by  $1/8 \ \mathcal{A}$, where [this factor will be crucial below in the discussion of the relevant terms at non-zero momentum, see in particular eq.\eqref{Aappl}]
\ben\label{Amult}
\mathcal{A}(D',D,n,l,w)\quad := K^{-3(|w|+D'-D)}\, K^{-2D'(n+2l)(n+2l-1)-D'}
\een
if $K$ is chosen large enough such that
\ben\label{Kcond1}
10(n+1)(2n+1)\ 2^{d(2,n,l,w,D')}K^{-D'(n+2l)(n+2l-1)-(|w|+D'-D)}\, \leq\, 1/8\quad ,
\een
for all $n,l$. To see that this is indeed possible, it is helpful to note that $l\geq 1$ in this case, since otherwise the first term on the right side of the FE is simply zero.
\vspace{.2cm}

\paragraph*{Second term on the r.h.s. of the flow equation:}

We integrate the second term in the FE \eqref{FE2} over $\Lambda'$ from $\Lambda$ to $\Lambda_{0}$ and insert our inductive bound as well as the known bound for the CAG's without insertions [see \eqref{propout}]
to obtain
\ben\label{487}
\begin{split}
&\Big|\int_{\Lambda}^{\Lambda_{0}}\d\Lambda'\sum_{\substack{l_{1}+l_{2}=l \\ n_{1}+n_{2}=n+1\\ w_{1}+w_{2}+w_{3}=w }}\!\!\!\!\! 4n_{1}n_{2} c_{\{w_{j}\}} \\
&\qquad \times\ \partial_{\vec{p}}^{w_{1}}\L^{\Lambda',\Lambda_{0}}_{D,2n_{1},l_{1}}(\O_{A_{1}}
\otimes\O_{A_{2}}; q,  p_{1},\ldots,p_{2n_{1}-1})\,  \partial_{\vec{p}}^{w_{2}}\dot{C}^{\Lambda'}(q)\,   \partial_{\vec{p}}^{w_{3}}\L^{\Lambda',\Lambda_{0}}_{2n_{2},l_{2}}(  p_{2n_{1}},\ldots,p_{2n}) \Big| \\
&\leq \sum_{\substack{l_{1}+l_{2}=l \\ n_{1}+n_{2}=n+1\\ w_{1}+w_{2}+w_{3}=w }}\!\!\!\!\! 4n_{1}n_{2}\ c_{\{w_{j}\}} \int_\La^{\Lao} d\La' \  \La'^{D-2n_{1}-|w_1|+4-2n_{2}-|w_2|}
\ K^{(4n_1+8l_1-3+2n_2+4l_2-4)(|w_1|+D'-D+|w_2|+1)}  \\
&\times  K^{D'(n_1+2l_1)^3}
 \sqrt{
|v_{1}|!\ |v_{2}|!}\ \frac{\sqrt{(|w_{1}|+D'-D)!}}{|x|^{D'-D}}
\sum_{\mu=0}^{d(2,n_1,l_1,w_1,D')}\frac{1}{\sqrt{\mu!}}\
 (\frac{|\vec p|}{\La'})^{\mu}\\
&\times\ \sum_{\lambda_1=0}^{2l_{1}+n_{1}}
\frac{\log^{\lambda_1}\bigr(\sup({|\vec{p}|
\over \ka '},{\ka '\over m})\bigl)}{2^{\lambda_1}\,\lambda_1!}
\frac{2}{\La'^3}\ |\pa^{w_3} \  \e ^{-\frac{q^2+m^2}{\La '^2}}|\
\sqrt{|w_2|!}\
(n_2 + l_2- 2)!\
\sum_{\lambda_2=0}^{l_{2}}
\frac{\log^{\lambda_2}
\bigr(\sup({|\vec{p}|\over \ka '},{\ka'\over m})\bigl)}
{2^{\lambda_2}\,\lambda_2!}
\end{split}
\een
We use the inequality $2l_{1}+n_{1}\leq 2l+n-1$, which holds due to the fact that the CAG without insertion vanishes unless $n_{2}+2l_{2}\geq 2$, to obtain
\ben\label{eqlhs}
\begin{split}
&\text{''r.h.s. of \eqref{487}''}\leq K^{D'(2l+n-1)^{3}}\ K^{(4n+8l-7)(|w_{1}|+|w_{2}|+D'-D+1)}\!\!\!\sum_{\substack{l_{1}+l_{2}=l \\ n_{1}+n_{2}=n+1\\ w_{1}+w_{2}+w_{3}=w }}\!\!\!\! 4n_{1}n_{2}\ c_{\{w_{j}\}}\ (n_{2}+l_{2}-2)!\\
&\times \sqrt{|v_{1}|!|v_{2}|!}\ \frac{\sqrt{(|w_{1}|+D'-D)! |w_{	2}|!}}{|x|^{D'-D}}\ \int_{\Lambda}^{\Lambda_{0}}\d\Lambda' \Lambda'^{D-2n-|w_{1}|-|w_{2}|-1}\ 2|\pa^{w_3} \  \e ^{-\frac{q^2+m^2}{\La '^2}}|\\
&\times \sum_{\mu=0}^{d(2,n_1,l_1,w_1,D')}\frac{1}{\sqrt{\mu!}}\
 (\frac{|\vec p|}{\La'})^{\mu}
  \sum_{\lambda_{1}=0}^{2l_{1}+n_{1}}\sum_{\lambda_{2}=0}^{l_{2}}\frac{\log^{\lambda_1+\lambda_{2}}\bigr(\sup({|\vec{p}|
\over \ka '},{\ka '\over m})\bigl)}{2^{\lambda_1+\lambda_{2}}\,(\lambda_1+\lambda_{2})!}\ \frac{(\lambda_1+\lambda_{2})!}{\lambda_1!\ \lambda_{2}!} \ .
\end{split}
\een
Then, using the bound \cite{Hollands:2011gf, sansone}
\ben
\Bigl| \pa ^{w} \e^{-\frac{q^2+m^2}{\La ^2}} | \ \le\
k\ \La ^{-|w|}\ \sqrt{|w|!}\ 2^{\frac{|w|}{2}}\ \e^{-\frac{q^2}{2\La ^2}}\
 \e^{-\frac{m^2}{\La ^2}} \quad, \,k=1.086\ldots
\label{w}
\een
and the identity $\lambda_1 + \lambda_2 \le 2l_{1}+n_{1}+l_{2}\leq 2l+n$ (since
$n_1+n_2=n+1$ and the right side of the above bound is $=0$ for $n_2=0$), we find the bound
\ben
\begin{split}
\text{''r.h.s. of \eqref{eqlhs}''}\leq& \,  \sum_{\begin{array}{c}\\[-.8cm]_{l_1+l_2=l},\\[-.2cm]
_{n_1+n_2=n+1}
 \end{array} }\!\!\!\!
(n_2+l_2)!\, 4n_1 (2l+n)  2^{2l+n} K^{(4n+8l-7)(|w|+D'-D+1)} K^{D'(n+2l-1)^3}\\
\times&\sum_{w_i}
 c_{\{w_i\}}\, 2^{\frac12 |w_3|}\ k
\frac{\sqrt{(|w_{1}|+D'-D)!}}{|x|^{D'-D}}
 \sqrt{ |w_3|!\, |w_2|!\ |v_{1}|!\ |v_{2}|!}\\
\times&\sum_{\mu=0}^{d(2,n_1,l_1,w_1,D')}\!\!\!\!\frac{1}{\sqrt{\mu!}}\
 \int_\La^{\Lao} d\La' \  \La'^{D-2n-|w|-1} (\frac{|\vec p|}{\La})^{\mu}
\sum_{\lambda=0}^{2l+n} \frac{\log^{\lambda}
\bigr(\sup({|\vec{p}|\over \ka '},{\ka '\over m})\bigl)}{2^{\lambda}\,\lambda!}\ .
\end{split}
\een
Now the $\Lambda'$ integral can again be estimated using the inequality \eqref{Lambdaint}.  Noting that also
\ben\label{dprop2}
d(2,n_1,l_1,w_1,D')\leq d(2,n,l,w,D')
\een
which holds as a result of the inequality $2(n_{1}+l_{1})\leq 2(n+l)-2$
(using $n_1+n_2=n+1$, and $n_2+l_2 \ge 2$), and using also
$\sum c_{\{w_i\}} 2^{|w_3|/2} = (2+\sqrt{2})^{|w|}$, one obtains the bound claimed in \eqref{boundCAG2} multiplied by $1/8 \ \mathcal{A}$,  which was defined in eq.\eqref{Amult}, provided that $K$ is chosen large enough that
\ben
K^{-D'(n+2l)(n+2l-1)+(4n+8l-7)-(|w|+D'-D)}\, 5k \!\!\!\!\!\!\!\!\!\!
\sum_{\begin{array}{c}\\[-.8cm]_{l_1+l_2=l}\\[-.2cm]
_{n_1+n_2=n+1}
 \end{array} } \!\!\!\!\!\!
(n_2+l_2)!\,  4\,n_1\, (2l+n)\ 2^{2l+n}\ (2+\sqrt{2})^{|w|}
  \le
\frac18
\een
To see that such a $K$ can always be found, it is helpful
to note the inequality $(n+2l)(n+2l-1)\geq 4n+8l-6$ .

\paragraph*{Third term on the r.h.s. of the flow equation (source term):}

Note that this term is a momentum integral over the  CAG's with one insertion, which have been bounded already, see \eqref{boundCAG1}. This integral obeys the following bound:

\begin{lemma}\label{lemmakint2N2}
Let
   $n_{1}+n_{2}=n+1$ and $l=l_{1}+l_{2}$. Then we have for any $D\leq D'$
\ben\label{eq:}
\begin{split}
&\left| \partial^{w}_{\vec{p}}\int_{k}\, \Lsc_{ 2n_{1},l_{1}}(\O_{A_{1}}(x);k,p_{1},\ldots, p_{2n_{1}-1})\, \dot{C}^{\Lambda}(k)\,\Lsc_{ 2n_{2},l_{2}}(\O_{A_{2}}(0);-k,p_{2n_{1}},\ldots, p_{2n}) \right|\\
\leq&\quad  M\, K_{0}^{(4n+8l-4)(|w|+D'-D)} K_{0}^{D'(n+2l)^{3}}\, \sqrt{|v_1|!|v_2|!}\, \Lambda^{D-2n-|w|-1}\, \e^{-m^{2}/\Lambda^{2}}\\
&\times\frac{\sqrt{(|w|+D'-D)!}}{|x|^{D'-D}} \sum_{\mu=0}^{{d}(N=2,n,l,w,D')}\frac{1}{\sqrt{\mu}!}\, \left(\frac{|\vec{p}|}{\Lambda}\right)^{\mu}\, \sum_{\lambda=0}^{2l+n-1}\frac{\log^{\lambda}(\sup(\frac{|\vec{p}|}{\kappa}), \frac{\kappa}{m})+\sqrt{\lambda!}}{2^{\lambda}\lambda!}
\end{split}
\een
where $M=5^{|w|+D'-D}\, 2^{2l+n+1}2^{2d}(2l+n+1)\ d\ 2^{[|w|+(D'-D)]/2}$ and where $K_{0}$ is the constant (called $K$ there) appearing in our bound on the CAG's with one insertion, see \eqref{boundCAG1}.
\end{lemma}
\begin{proof}
 Note that we are bounding the same integral as in lemma 3.2 in \cite{Hollands:2011gf}. Comparing both estimates one finds that in our bound we have effectively traded powers of $\Lambda$ for inverse powers of $|x|$.
This is achieved as follows: Using  the translation covariance property of the CAG's [see discussion preceding eq.\eqref{CAGtrans}], pulling momentum derivatives into the integral and performing partial integrations (the integrand decays sufficiently rapidly for large $k$, due to the exponential damping factor in $\dot{C}^{\Lambda}$) 
we obtain (cf. \cite{Keller:1992by} A.2 p.275)
\ben\label{lemmakint2PI}
\begin{split}
&\left| \partial^{w}_{\vec{p}}\int_{k}\, \Lsc_{ 2n_{1},l_{1}}(\O_{A_{1}}(x);k,p_{1},\ldots, p_{2n_{1}-1})\, \dot{C}^{\Lambda}(k)\,\Lsc_{ 2n_{2},l_{2}}(\O_{A_{2}}(0);-k,p_{2n_{1}},\ldots, p_{2n}) \right|\\
\leq&\,\Big|\int_{k} \sum_{w_{1}+w_{2}+w_{3}=w} c_{\{w_{i}\}} \  \e^{ikx} \ \e^{ix(p_{1}+\ldots+p_{2n_{1}-1})}  \\
&\qquad  \partial_{k}^{w_{3}}\Big( \partial_{\vec{p}}^{w_{1}} \Lsc_{ 2n_{1},l_{1}}(\O_{A_{1}}(0);k,p_{1},\ldots, p_{2n_{1}-1}) \dot{C}^{\Lambda}(k)\, \partial_{\vec{p}}^{w_{2}}\Lsc_{ 2n_{2},l_{2}}(\O_{A_{2}}(0);-k,p_{2n_{1}},\ldots, p_{2n}) \Big) \Big|
\end{split}
\een
Denote by $\|x\|=\max_{\alpha\in\{1,\ldots, 4\}} |x_{\alpha}|$ the maximal component of $x$. We can multiply both sides of \eqref{lemmakint2PI} by $(2\|x\|)^{D'-D}$, pull this factor into the integral, and we
use the elementary fact that $\|x\|^{D'-D} \ e^{ikx}$  can be written as a $(D-D')$-fold $k$-derivative on $\e^{ikx}$ with respect to a suitable component of the 4-vector $k$. That $k$-derivative may then be moved onto the term
in the last line of \eqref{lemmakint2PI} via a partial integration. The resulting expression
can now be bounded just as in lemma 3.2 in~\cite{Hollands:2011gf} by substituting the bounds for the CAG's with one insertion given in that paper. After some bookkeeping we obtain the following bound
for the right side of \eqref{lemmakint2PI}
\ben\label{lemmakintsubst}
\begin{split}
&\eqref{lemmakint2PI} \ \times \ (2\|x\|)^{D'-D} \le M_{0}\, K_{0}^{(4n+8l-4)(|w|+D'-D)} K_{0}^{D'(n+2l)^{3}}\, \sqrt{(|w|+D'-D)!|v_1|!|v_2|!}  \\ & \times\, \Lambda^{D-2n-|w|-5}\ \e^{-m^{2}/\Lambda^2} \
\int_{k} \e^{-k^{2}/2\Lambda^{2}} \sum_{\mu=0}^{d_{1}+d_{2}}\frac{1}{\sqrt{\mu}!}\, \left(\frac{|\vec{p}|+|k|}{\Lambda}\right)^{\mu}\, \sum_{\lambda=0}^{2l+n-1}\frac{\log^{\lambda}(\sup(\frac{|\vec{p}|+|k|}{\kappa}), \frac{\kappa}{m})}{2^{\lambda}\lambda!}
\end{split}
\een
where $M_{0}=5^{|w|+D'-D}\, 2^{2l+n}2^{d_{1}+d_{2}}(2l+n)(d_{1}+d_{2})2^{[|w|+D'-D]/2}$ and where we used the shorthand $d_{i}=2[A_{i}](n_{i}+l_{i})+\sup([A_{i}]+1-2n_{i}-|w_{i}|, 0)$. We would now like to replace $d_{1}+d_{2}$ by $d(2,n,l,w,D')$ in the summation index. Here it is  crucial that $d$ satisfies the property
\ben\label{dreplace}
d(2,n,l,w,D')\geq d_{1}+d_{2}
\een
with $d_{i}=2[A_{i}](n_{i}+l_{i})+\sup([A_{i}]+1-2n_{i}-|w_{i}|, 0)$ and where $n_{1}+n_{2}=n+1, l_{1}+l_{2}=l, w_{1}+w_{2}=w$ and $[A_{1}]+[A_{2}]=D'$.
Substituting the definition of $d$ given in the statement of theorem \ref{thmCAG2} we indeed find
\ben
d(2,n,l,w,D')= 2D'(n+l+2)+\sup(D'+1-2n-|w|, 0)\geq 2D'(n+l+1)+2D'\geq d_{1}+d_{2}
\een
for arbitrary $w$. The momentum integral can then be bounded by an application of the inequality \eqref{kint}.
%
%
\ben\label{momentumint}
\begin{split}
&\sum_{\mu=0}^{{d}}\frac{1}{\sqrt{\mu !}}\int_{k/\Lambda}\e^{-|k|^{2}/2\Lambda^{2}}\,  \left(\frac{|\vec{p}|+|k|}{\Lambda}\right)^{\mu}\sum_{\lambda=0}^{2l+n-1}\frac{\log^{\lambda}(\sup(\frac{|\vec{p}|+|k|}{\kappa}, \frac{\kappa}{m}))}{2^{\lambda}\lambda!}\\
\leq\quad& 2^{{d}}\sum_{\mu=0}^{{d}}\frac{1}{\sqrt{\mu !}}\,  \left(\frac{|\vec{p}|}{\Lambda}\right)^{\mu}\sum_{\lambda=0}^{2l+n-1}\frac{\log^{\lambda}(\sup(\frac{|\vec{p}|}{\kappa}, \frac{\kappa}{m}))+\sqrt{\lambda!}}{2^{\lambda}\lambda!}
\end{split}
\een
Finally, we divide again by $(2\|x\|)^{D'-D}$ and note the elementary inequality
\ben\label{normsinq}
|x|=\sqrt{x_{1}^{2}+\ldots+x_{4}^{2}}\leq  2\|x\|
\een
to obtain the bound claimed in the lemma.
\end{proof}
Now, to obtain a bound for the third term in the FE we have to integrate the bound derived in lemma \ref{lemmakint2N2} over $\Lambda' $ between $\Lambda_{0}$ and $\Lambda$. Using the inequality \eqref{Lambdaint} for the $\Lambda'$ integral, we obtain the bound
\ben\label{sourceAbound}
\begin{split}
&\Bigg| \int_\La^{\Lao} d\La' \partial_{\vec p}^w \int_k {\cal
  L}^{\La',\Lao}_{2n_1,l_1}(\O_{A_{1}}; k, p_1,\ldots,p_{2n_1-1})\,
{\dot C}^{\La'}(k)\,\,
{\cal L}^{\La',\Lao}_{2n_2,l_2}(\O_{A_{2}}; -k, p_{2n_1},\ldots,p_{2n})
\Bigg|\\
&\leq  M\ K_0^{(4n+8l-4)(|w|+D'-D)}\
 K_0^{D'(n+2l)^3 }\ \frac{\sqrt{(|w|+D'-D)!}}{|x|^{D'-D}} \ \sqrt{\ |v_{1}|! \ |v_{2}|!}   \\
&\times \ \int^{\Lao}_\La
 d\La' \La'^{-1+D-2n-|w|}
\ \e^{-m^2/\Lambda'^2} \sum_{\mu=0}^{d(2, n,l,w,D')}\frac{1}{\sqrt{\mu!}}\
 (\frac{|\vec p|}{\La'})^{\mu}\,
\sum_{\la=0}^{2l+n-1}
\frac{\log^{\lambda}(\sup(\frac{|\vec{p}|}{\kappa'},\frac{\kappa'}{m}))
+ \sqrt{\lambda !}}
{2^{\lambda}\,\lambda!}\\
&\le 5\ M\ K_0^{(4n+8l-4)(|w|+D'-D)}\
 K_0^{D'(n+2l)^3 }\ \frac{\sqrt{(|w|+D'-D)!}}{|x|^{D'-D}} \ \sqrt{ |v_{1}|! \ |v_{2}|!}  \\
 &\times \ \La^{D-2n-|w|} \ \sum_{\mu=0}^{d(2, n,l,w,D')}\frac{1}{\sqrt{\mu!}}\
 (\frac{|\vec p|}{\La})^{\mu}\,
\sum_{\lambda=0}^{2l+n-1}
\frac{\log^{\lambda}(\sup(\frac{|\vec{p}|}{\kappa},\frac{\kappa}{m}))}
{2^{\lambda}\,\lambda!}
\end{split}
\een
where $K_{0}$ is the constant from lemma \ref{lemmakint2N2}.
Recall from the FE, eq.\eqref{FE2}, that this term is to be multiplied by $4n_{1}n_{2}$, and we also have to apply the operator $\mathbb{S}$ and sum over the configurations $n_{1}+n_{2}=n+1$ and $l_{1}+l_{2}=l$. In total we find that the inductive bound is reproduced under the condition that $K$ satisfies
\ben
20 (l+1)(n+1)^3 \,M\, K_0^{(4n+8l-4)(|w|+D'-D)}\, K_0^{D'(n+2l)^3} \le \frac{1}{8} K^{(4n+8l-3)(|w|+D'-D)}\, K^{D'(n+2l)^3}\ .
\een
Here it is useful to note that both $(2l+n)$ and $(4n+8l-3)$ are always positive (since $n+l\geq 1$), which helps one to see that the inequality can be satisfied by making $K$ large enough.

\noindent \textbf{(b) Relevant terms ($2n+|w|\leq D$) at $\vec{p}=\vec{0}$:}

As the boundary conditions for the relevant terms are given at zero momentum, we will first derive bounds for $\vec{p}=\vec{0}$ and then proceed to arbitrary momentum with the help of the Taylor expansion formula
\ben\label{reltermstaylor}
\partial_{\vec p }^w\ f_{2n}(\vec p)\,=\!\!\!\!
\!\!\sum_{|\tilde{w}| \le D-2n-|w|}\! \frac{{\vec p}^{\tilde{w}}}{\tilde{w}!}\,
[\partial_{\vec p }^{\tilde{w}+w} f_{2n}](0)\, +\!\!\!
\!\!\!\!\sum_{|\tilde{w}|=D+1-2n-|w|}\!\!\!\!\!\! {\vec p}^{\tilde{w}}
\int_0^1 d\tau  \frac{(1-\tau)^{|\tilde{w}|-1}}{(|\tilde{w}|-1)!}\,
[\partial_{\vec p }^{\tilde{w}+w} f_{2n}](\tau {\vec p})\ .
\een

\paragraph*{First term on the r.h.s. of the FE:}
In view of eq.\eqref{reltermstaylor}, let us consider the first term on the r.h.s. of the FE with momentum derivatives $\partial_{\vec{p}}^{\tilde{w}+w}$ and with $2n+|\tilde{w}+w|\leq D$ at zero momentum. Integrating over $\Lambda'$ from $0$ to $\Lambda$ we find
 \ben\label{1strel0}
 \begin{split}
 &\Big|\int_{0}^{\Lambda}\d\Lambda' \left(\atop{2n+2}{2}\right)\int\d^{4}k \ \dot{C}^{\Lambda'}(k) \, \partial_{\vec{p}}^{w+\tilde{w}}\L_{D,2n+2,l-1}^{\Lambda',\Lambda_{0} }(\O_{A_{1}}\otimes\O_{A_{2}}; k,-k, 0,\ldots, 0)\Big| \\
&\leq  \left(\atop{2n+2}{2} \right)\
K^{(4n+8l-7)(|\tilde{w}+w|+D'-D)} \, K^{D'(n+2l-1)^3}\ \sqrt{|v_{1}|!|v_{2}|!}\
\int_0^{\Lambda} d\Lambda'\  \Lambda'^{D-(2n+2)-|\tilde{w}+w|}\ \frac{2}{\Lambda'^3}\\
&\times \frac{\sqrt{(|w+\tilde{w}|+D'-D)!}}{|x|^{D'-D}}
\int\d^{4}k\,
\e^{-{{k^2+m^2}\over \Lambda'^2}}
\sum_{\mu=0}^{d(2,n,l,\tilde{w}+w, D')} \frac{1}{\sqrt{\mu!}}
(\frac{|k|}{\Lambda'})^{\mu}\,\sum_{\lambda=0}^{2l+n-1}
\frac{\log^{\lambda}(\sup(\frac{|k|}{\kappa'},\frac{\kappa '}{m}))}
{2^{\lambda}\,\lambda!}\\
&\leq  \left(\atop{2n+2}{2} \right)
K^{(4n+8l-7)(|\tilde{w}+w|+D'-D)} \, K^{D'(n+2l-1)^3}\frac{ \sqrt{|v_{1}|!|v_{2}|!\,( |\tilde{w}+w|+D'-D)!} }{|x|^{D'-D}}  \sum_{\mu=0}^{d(2,n,l,\tilde{w}+w,D')}\\
&\times\sum_{\lambda=0}^{2l+n-1}\frac{2}{2^{\lambda}\lambda!} \int_0^{\Lambda} d\Lambda'\  \Lambda'^{D-(2n+1)-|\tilde{w}+w|}e^{-\frac{m^{2}}{\Lambda'^{2}}}\int\d^{4}(k/\Lambda')\, \frac{1}{\sqrt{\mu!}}(\frac{|k|}{\Lambda'})^{\mu}\,\log^{\lambda}(\sup(\frac{|k|}{\kappa'},\frac{\kappa '}{m})) e^{-\frac{k^{2}}{\Lambda'^{2}}}
 \end{split}
 \een
The momentum integral in the last line can be estimated by (cf. the inequality (76) in \cite{Hollands:2011gf})
\ben
\int\d^{4}(k/\Lambda')\, \frac{1}{\sqrt{\mu!}}(\frac{|k|}{\Lambda'})^{\mu}\,\log^{\lambda}(\sup(\frac{|k|}{\kappa'},\frac{\kappa '}{m})) e^{-\frac{k^{2}}{\Lambda'^{2}}}
\leq 
2^{\frac{\mu}{2}}\ \frac{1}{\sqrt{\mu!}}\
(\frac{\mu}{2})!\
[\log^\lambda (\frac{\kappa'}{m}) \,+\,(\lambda!)^{1/2}] 
\een
 and for the subsequent sum over $\mu$ we can use the bound (cf. inequality (88) in \cite{Hollands:2011gf})
 \ben
 \sum_{\mu=0}^{d} 2^{\frac{\mu}{2}}\ \frac{1}{\sqrt{\mu!}}(\frac{\mu}{2})!
 \, \leq\, 2\, d^{3/2}\, .
 \een
For the $\Lambda'$ integral we make use of (cf. inequality (89) in \cite{Hollands:2011gf})
\ben\label{lambdaint}
\begin{split}
&\int_0^{\Lambda} \d\Lambda'\ \Lambda'^{D-(2n+1)-|\tilde{w}+w|}\ \log^\lambda (\frac{\kappa'}{m})\,
\e^{-{{m^2}\over \Lambda'^2}}  \\
&\leq\, \Lambda^{D-2n-|\tilde{w}+w|}\begin{cases} \log^\lambda(\frac{\ka}{m}) & \text{if $D-2n-|\tilde{w}+w|>0$,}\\
2 (\lambda+1)^{-1}\log^{\lambda+1}(\frac{\ka}{m}) & \text{if $D-2n-|\tilde{w}+w|=0$.}
\end{cases}
\end{split}
\een
Using these bounds to estimate the r.h.s. of \eqref{1strel0} shows that this contribution satisfies the inductive bound, formula \eqref{boundCAG2}, multiplied by $1/8 \ \mathcal{A}$, defined in eq.\eqref{Amult}, provided that $K$ is chosen such that
\ben\label{beforeTaylor}
12 \ \left(\atop{ 2n+2}{ 2} \right)\ d(2,n,l,\tilde{w}+w,D')^{3/2}\
  K^{-D'(n+2l)(n+2l-1)-(|w+\tilde{w}|+D'-D)}\
\le\ 1/8\ .
\een
It can be seen that this is indeed possible to find such a $K$. This is easy to see for large values of $n+l$. To see that it is also true for small $n+l$, it is useful to recall that for the first term on the r.h.s. of the FE, \eqref{FE2}, we can assume $l\geq 1$, so $K$ will always have a negative exponent.

\paragraph*{Second term on the r.h.s of the FE:}

Inserting the induction hypothesis for the CAG's with two insertions and the known bounds for the CAG's without insertion, formula \eqref{propout}, \eqref{prop40} and \eqref{prop20}, yields
\ben
\begin{split}
&\Big|\int_{0}^{\Lambda}\d\Lambda'\sum_{\substack{l_{1}+l_{2}=l \\ n_{1}+n_{2}=n+1\\ w_{1}+w_{2}+w_{3}=w+\tilde{w} }}\!\!\!\!\! 4n_{1}n_{2}\  c_{\{w_{j}\}} \\
&\qquad \times\ \partial_{\vec{p}}^{w_{1}}\L^{\Lambda',\Lambda_{0}}_{D,2n_{1},l_{1}}(\O_{A_{1}}
\otimes\O_{A_{2}}; q,  p_{1},\ldots,p_{2n_{1}-1})\,  \partial_{\vec{p}}^{w_{3}}\dot{C}^{\Lambda'}(q)\,   \partial_{\vec{p}}^{w_{2}}\L^{\Lambda',\Lambda_{0}}_{2n_{2},l_{2}}(  p_{2n_{1}},\ldots,p_{2n}) \Big|_{\vec{p}=\vec{0}} \\
&\leq\sum_{\begin{array}{c}\\[-.8cm]_{l_1+l_2=l,}\\[-.2cm]
_{w_1+w_2+w_3=w+\tilde{w},}\\[-.2cm]
_{n_1+n_2=n+1} \end{array} }\!\!\!\!\!\!\!\!\!\!\!\!
4\, n_1 n_2\,\int_0^\La d\La' \  \La'^{D+4-2n-2-|w_1|-|w_2|}
\ K^{(4n_1+8l_1-3+2n_2+4l_2-4)(|w_1|+D'-D+|w_2|+1)}   \\
&\qquad\times\ K^{D'(n_1+2l_1)^3}\ \sqrt{|v_{1}|!\ |v_{2}|!}\ \frac{\sqrt{(|w_{1}|+D'-D)!}}{|x|^{D'-D}}\
c_{\{w_i\}}\,
\sum_{\lambda_1=0}^{2l_{1}+n_{1}}
\frac{\log^{\lambda_1}({\ka '\over m})}{2^{\lambda_1}\,\lambda_1!}\\
&\qquad\times\ \frac{2}{\La'^3}\
\bigg|\pa^{w_3} \  \e^{-\frac{q^2+m^2}{\La '^2}}\bigg|_{q=0}\
\sqrt{|w_2|!}\
(n_2 + l_2- 2)!
\sum_{\lambda_2=0}^{l_{2}}
\frac{\log^{\lambda_2}({\ka'\over m})}
{2^{\lambda_2}\,\lambda_2!}
\end{split}
\een
We again use the inequality \eqref{w} and proceed as in \eqref{lambdaint} to estimate the $\Lambda'$ integral. As a result we arrive at the bound \eqref{boundCAG2} multiplied by $1/8 \ \mathcal{A}$, under the condition that $K$ satisfies the lower bound
\ben
\begin{split}
&6\,
K^{-D'(n+2l)(n+2l-1)+(4n+8l-7)-(|w+\tilde{w}|+D'-D)}\, 5k \\
&\times \sum_{\begin{array}{c}\\[-.8cm]_{l_1+l_2=l,}\\[-.2cm]
_{n_1+n_2=n+1}
 \end{array} } \!\!\!\!\!\!
(n_2+l_2)!\,  4\,n_1\, (2l+n)\ 2^{2l+n}\
(2+\sqrt{2})^{|w+\tilde w|}
  \le
\frac18 
\end{split}
\een
The inequality $(n+2l)(n+2l-1)\geq (4n+8l-6)$ ensures that $K$ always appears with a negative exponent on the right side, which is helpful in order to see that such a $K$ can be found.

\paragraph*{Third term on the r.h.s. of the FE :}

Using lemma \ref{lemmakint2N2}, we find
\ben\label{3rdFEN2}
\begin{split}
&\Bigg| \int_0^{\La} d\La' \partial_{\vec p}^{w+\tilde{w}} \int_k {\cal
  L}^{\La',\Lao}_{2n_1,l_1}(\O_A; k, 0,\dots,0)\,
{\dot C}^{\La'}(k)\,\,
{\cal L}^{\La',\Lao}_{2n_2,l_2}(\O_B; -k, 0,\dots,0)
\Bigg|\\
&\leq\ M\ K_0^{(4n+8l-4)(|w|+|\tilde{w}|+D'-D)}\
 K_0^{D'(n+2l)^3 }\ \sqrt{ |v_{1}|! \ |v_{2}|!}\, \frac{\sqrt{(|w|+|\tilde{w}|+D'-D)!}}{|x|^{D'-D}} \\
&\times \ \ \int^{\La}_0 d\La' \La'^{D-2n-|w|-|\tilde{w}|-1} \
\e^{-m^2/\Lambda'^2}
\sum_{\lambda=0}^{2l+n-1}
\frac{\log^{\lambda}(\frac{\kappa'}{m}) + \sqrt{\lambda !}}
{2^{\lambda}\,\lambda!} \ .
\end{split}
\een
For the integral over $\Lambda'$ we use the inequality \cite{Hollands:2011gf}
\ben
\begin{split}
&\int^{\La}_0 d\La' \La'^{D-2n-|w|-|\tilde{w}|-1} \ \e^{-m^2/\Lambda'^2} \
\sum_{\la=0}^{2l+n-1} \frac{\log^\lambda(\frac{\kappa'}{m}) + \sqrt{\lambda !}}
{2^{\lambda}\,\lambda!}\\
&\le\ \La^{D-2n-|w|-|\tilde{w}|} \sum_{\la=0}^{2l+n-1} \frac{1}{{2^{\lambda}\,\lambda!}}
\begin{cases}
\frac{1}{\lambda+1} \log^{\lambda+1}(\frac{\kappa}{m}) +
\sqrt{\lambda!}\log(\frac{\kappa}{m}) + 1 & \text{if $D-2n-|w|=|\tilde{w}|$,}\\
\frac{\log^{\lambda}(\frac{\kappa}{m}) + \sqrt{\lambda !}}{D-2n-|w|-|\tilde{w}|} &
\text{if $D-2n-|w|>|\tilde{w}|$}
\end{cases}\\
&\le\ 6(2l+n) \La^{D-2n-|w|-|\tilde{w}|}
\sum_{\lambda=0}^{2l+n} \frac{\log^{\lambda}(\frac{\kappa}{m})}
{2^{\lambda}\,\lambda!}\ .
\end{split}
\een
Recalling that we have to multiply \eqref{3rdFEN2} by $4n_{1}n_{2}$ and that we also have to apply the symmetrization operator $\mathbb{S}$ and sum over the indices $n_{i}$ and $l_{i}$, we reproduce the inductive bound under the condition
\ben
\begin{split}
4 \cdot 6(l+1)(2l+n)\ (n+1)^3 M K_0^{D'(n+2l)^3}&K_0^{(4n+8l-4)(|w|+|\tilde{w}|+D'-D)}\\
&\le \frac{1}{8} K^{D'(n+2l)^3}K^{(4n+8l-3)(|w|+|\tilde{w}|+D'-D)},
\end{split}
\een
on $K$. Again, this condition can be satisfied by a sufficiently large  $K$.

\vspace{.5cm}

\noindent \textbf{(c) Relevant terms ($2n+|w|\leq D$) at arbitrary momentum:}

In order to proceed to non-zero momentum we now make use of the Taylor series.
\ben\label{Taylornonzero}
\begin{split}
&|\partial_{\vec{p}}^{w}\L^{\Lambda,\Lambda_{0}}_{D,2n,l}(\O_{A_{1}}\otimes\O_{A_{2}}; \vec{p})|\, =\, \Big|\! \!\! \!\sum_{|\tilde{w}| \le D-2n-|w|}\! \! \frac{{\vec p}^{\,\tilde{w}}}{\tilde{w}!}\,
\,\partial_{\vec p }^{\tilde{w}+w} {\cal L}_{D,2n,l}^{\Lambda,\Lambda_{0}}(\O_{A_{1}}\otimes\O_{A_{2}};{\vec 0})\\
 &+\! \!\! \!\! \!\sum_{|\tilde{w}|=D+1-2n-|w|} \! \!\! \!
{\vec p}^{\tilde{w}} \int_0^1 \frac{(1-\tau)^{|\tilde{w}|-1}}{(|\tilde{w}|-1)!}\,
\,\partial_{\tau\vec p }^{\tilde{w}+w}
{\cal L}_{D,2n,l}^{\Lambda,\Lambda_{0}}(\O_{A_{1}}\otimes\O_{A_{2}};\tau {\vec p})\ \Big|
\end{split}
\een
On the right hand side we can now use the bounds previously derived in \textbf{(a)} and \textbf{(b)}.
\ben\label{insertA}
\begin{split}
&|\partial_{\vec{p}}^{w}\L^{\Lambda,\Lambda_{0}}_{D,2n,l}(\O_{A_{1}}\otimes\O_{A_{2}}; \vec{p})| \\
&\leq\sum_{|\tilde{w}|\leq D-2n-|w|} \Lambda^{D-2n-|w|} \left(\frac{|\vec{p}|}{\Lambda}\right)^{|\tilde{w}|}\, \frac{\sqrt{|v_{1}|! |v_{2}|! (|w|+|\tilde{w}|+D'-D)!}}{\tilde{w}!\, |x|^{D'-D}}\sum_{\lambda=0}^{2l+n}\frac{\log^{\lambda}(\frac{\kappa}{m})}{2^{\lambda}\lambda!}\\
&\qquad\times\Big[ 1/4\ \mathcal{A}(D',D,n,l,w+\tilde{w})\, K^{(4n+8l-3)(|w|+|\tilde{w}|+D'-D)+D'(2l+n)^{3}}\\
&\qquad\quad + 24 M (2l+n)(l+1)(n+1)^{3} \, K_{0}^{(4n+8l-4)(|w|+|\tilde{w}|+D'-D)+D'(2l+n)^{3}}\Big]\\
&+\sum_{|\tilde{w}|= D+1-2n-|w|} \Lambda^{D-2n-|w|} \left(\frac{|\vec{p}|}{\Lambda}\right)^{|\tilde{w}|}\, \frac{\sqrt{|v_{1}|! |v_{2}|! (|w|+|\tilde{w}|+D'-D)!}}{|\tilde{w}|!\, |x|^{D'-D}}\ |\tilde{w}|\\
&\quad\times \int_{0}^{1}\d\tau (1-\tau)^{|\tilde{w}|-1} \sum_{\mu=0}^{d(2,n,l,w+\tilde{w},D')} \frac{1}{\sqrt{\mu!}}\left(\frac{\tau|\vec{p}|}{\Lambda}\right)^{\mu}\ \sum_{\lambda=0}^{2l+n}\frac{\log^{\lambda}(\sup(\frac{\tau|\vec{p}|}{\kappa}, \frac{\kappa}{m}))}{2^{\lambda}\lambda!}\\
&\qquad\times\Big[ 1/4\ \mathcal{A}(D',D,n,l,w+\tilde{w})\, K^{(4n+8l-3)(|w|+|\tilde{w}|+D'-D)+D'(2l+n)^{3}}\\
&\qquad\quad + 20 M (l+1)(n+1)^{3} \, K_{0}^{(4n+8l-4)(|w|+|\tilde{w}|+D'-D)+D'(2l+n)^{3}}\Big]\\
\end{split}
\een
Here $K_{0}$ is the constant (called $K$ there) from the bound on the CAG's with one insertion, \eqref{boundCAG1m},  and $M$ is the constant defined in the statement of lemma \ref{lemmakint2N2}.
To obtain the r.h.s. of the inequality above, we used the fact that the first and second term on the r.h.s. of the FE satisfy the bound \eqref{boundCAG2} multiplied by $1/8 \ \mathcal{A}$ in the cases \textbf{(a)} and \textbf{(b)}. Hence the expressions include the factor $\mathcal{A}$. For the  contribution from the third term in the FE we also used the bounds derived above for the cases \textbf{(a)} and \textbf{(b)}, but expressed in terms of the constant $K_{0}$ instead of $K$, see the inequalities \eqref{sourceAbound} and \eqref{3rdFEN2} and the discussion following them.

To simplify this expression, we make use of:
\ben
\sqrt{(|w|+|\tilde{w}|+D'-D)!}\leq \sqrt{|\tilde{w}|! (|w|+D'-D)!}\ 2^{(|w|+|\tilde{w}|+D'-D)/2}
\een
\ben\label{binom}
\sum_{|\tilde w|=j} \frac{|\tilde{w}|!}{\tilde{w}!}= (8n)^{j}
\een
\begin{equation}
\frac{1}{\sqrt{|\tilde{w}|!}}\sum_{\mu=0}^{d}\frac{1}{\sqrt{\mu!}}\left(\frac{|\vec{p}|}{\Lambda}\right)^{\mu+|\tilde{w}|}\leq \sum_{\mu=0}^{d+\tilde{w}}\frac{1}{\sqrt{\mu!}}\left(\frac{|\vec{p}|}{\Lambda}\right)^{\mu}\ 2^{\mu/2}
\end{equation}
\ben\label{dprop1}
d(2,n,l,\tilde{w}+w, D')+|\tilde{w}| \leq d(2,n,l,w,D') \quad \text{ for } |\tilde{w}|\leq D'-2n-|w|+1
\een
We then obtain the bound
\ben\label{insertA2}
\begin{split}
&|\partial_{\vec{p}}^{w}\L^{\Lambda,\Lambda_{0}}_{D,2n,l}(\O_{A_{1}}\otimes\O_{A_{2}}; \vec{p})| \\
&\leq \Lambda^{D-2n-|w|} \, \frac{\sqrt{|v_{1}|! |v_{2}|! (|w|+D'-D)!}}{ |x|^{D'-D}} \sum_{\mu=0}^{d(2,n,l,w,D')} \frac{1}{\sqrt{\mu!}}\left(\frac{|\vec{p}|}{\Lambda}\right)^{\mu}   \sum_{\lambda=0}^{2l+n}\frac{\log^{\lambda}(\sup(\frac{|\vec{p}|}{\kappa},\frac{\kappa}{m}))}{2^{\lambda}\lambda!}\\
&\quad\times\Big(\sum_{|\tilde{w}|\leq D-2n-|w|}  \Big[ 1/4 \  \mathcal{A}(D',D,n,l,w+\tilde{w})\, K^{(4n+8l-3)(|w|+|\tilde{w}|+D'-D)+D'(2l+n)^{3}}\\
&\quad\quad + 24 M (2l+n)(l+1)(n+1)^{3} \, K_{0}^{(4n+8l-4)(|w|+|\tilde{w}|+D'-D)+D'(2l+n)^{3}}\Big]2^{(|w|+|\tilde{w}|+D'-D)/2}\,
\frac{(8n)^{|\tilde{w}|}}{\sqrt{\tilde{w}!}} \Big)\\
&\quad\times\Big( \sum_{|\tilde{w}|\leq D+1-2n-|w|} \Big[ 1/4 \  \mathcal{A}(D',D,n,l,w+\tilde{w})\, K^{(4n+8l-3)(|w|+|\tilde{w}|+D'-D)+D'(2l+n)^{3}}\\
&\quad\quad + 20 M (l+1)(n+1)^{3} \, K_{0}^{(4n+8l-4)(|w|+|\tilde{w}|+D'-D)+D'(2l+n)^{3}}\Big] 2^{(|w|+|\tilde{w}|+D'-D)/2}\, |\tilde{w}|\, 2^{d/2}\Big) \\
\end{split}
\een
It follows, using also~\eqref{binom}, that this contribution satisfies the bound claimed in the theorem, \eqref{boundCAG2}, provided that $K$ is chosen sufficiently large such that the following two conditions are satisfied:
\ben\label{Aappl}
\begin{split}
\sum_{j \le D+1-2n-|w|}\!\!\!\!\!\!\!\!
(8n)^{j}\, j\,   2^{\frac{d(2,n,l,w,D')}{2}} &
\, 2^{\frac{j+|w|+D'-D}{2}} K^{(4n+8l-3)j} \  {\cal A}(D',D,n,l,|w|+j) \le 1
\end{split}
\een
\begin{equation}\label{eq:Aappl2}
\begin{split}
\sum_{j \le D+1-2n-|w|} &48 M (2l+n)(l+1)(n+1)^{3} \, K_{0}^{(4n+8l-4)(|w|+j+D'-D)+D'(2l+n)^{3}}\\
&\times (8n)^{j}\, j\,   2^{\frac{d(2,n,l,w,D')}{2}}
\, 2^{\frac{j+|w|+D'-D}{2}} \leq K^{(4n+8l-3)(|w|+D'-D)+D'(2l+n)^{3}}
\end{split}
\end{equation}
It can be seen that it is indeed possible to choose $K$ such that the conditions \eqref{Aappl} and \eqref{eq:Aappl2} are satisfied, which finishes the proof of theorem \ref{thmCAG2}.
\end{proof}

The following corollary is a consequence of thm.~\ref{thmCAG2}.
\begin{corollary}\label{CAG2cor}
For $\Lambda\leq m$ there exists a constant $K>0$ such that
\ben
\begin{split}
& |\partial^{w}_{\vec{p}}\L^{\Lambda,\Lambda_{0}}_{D, 2n,l}(\O_{A_{1}}(x)\otimes\O_{A_{2}}(0); \vec{p})|\leq m^{D-2n-|w|}\,  K^{(4n+8l-3)(|w|+D'-D)}\, K^{D'(n+2l)^{3}}\, \sqrt{|v_1|!|v_2|!} \\
& \hspace{2cm} \times   \frac{\sqrt{(|w|+D'-D)! (2n+|w|-D)_{+}!}}{|x|^{D'-D}} \sup\left(1, \frac{|\vec{p}|}{m}\right)^{d(2,n,l,w,D')}\, \sum_{\lambda=0}^{2l+n}\frac{\log_{+}^{\lambda}(\frac{|\vec{p}|}{m})}{2^{\lambda}\lambda!}
\end{split}
\een
with $d(N,n,l,w,D'):= 2D'(n+l+2(N-1))+\sup(D'+1-2n-|w|, 0)$ and $A_i \equiv \{n_i, v_i\}$.
\end{corollary}
\noindent {\bf Remark:} This bound implies convergence of the Taylor expansion of $\partial^{w}_{\vec{p}}\L^{\Lambda,\Lambda_{0}}_{D, 2n,l}(\O_{A_{1}}(x)\otimes\O_{A_{2}}(0))$ with respect to $x$ in a neighborhood of $x\neq 0$, for any degree of regularization $D\leq [A_{1}]+[A_{2}]$, and uniformly in $\Lambda,\Lambda_{0}$.  This can be seen from the fact that the bound grows like
\ben
 j! \Big(\sup(|\vec{p}|/m, 1)^{2(n+l+2)+1} \tilde{K}/|x|\Big)^{j}
\een
if we take $j$ derivatives with respect to $x$ and apply the Lowenstein rule \eqref{loew2}.
Hence the CAG's with two insertions are real analytic in $x$ for $x \neq 0$, and similarly
for the OPE coefficients, related to them by defn.~\eqref{defOPE}.

\begin{proof}
We can insert the bounds from theorem \ref{thmCAG2} into the FE \eqref{FE2} once more and integrate over $\Lambda$ from $0$ to $m$. Note that there is a damping factor $\exp(-m^{2}/\Lambda^{2})$ in each of the terms on the r.h.s. of the FE, so we can bound negative powers of $\Lambda$ through the estimate
\ben
\begin{split}
\int_{0}^{m}\d\Lambda'\, \exp(-m^{2}/\Lambda'^{2})\, \frac{\Lambda'^{D-2n-|w|-\mu-1}}{\sqrt{\mu!}}\, \leq \, m^{D-2n-|w|-\mu}\, \frac{\sqrt{(2n+|w|+\mu-D)_{+}!}}{\sqrt{\mu}!}&\\
\leq\, 2^{(2n+|w|+\mu-D)_{+}/2}\, m^{D-2n-|w|-\mu}\, \sqrt{(2n+|w|-D)_{+}!}&
\end{split}
\een
Choosing a somewhat larger constant $K$ to accommodate for the additional powers of $2$ we obtain the corollary.
\end{proof}

\subsection{CAG's with three insertions}

 We next derive a bound on the CAG's $L^{\Lambda,\Lambda_{0}}_{\bfD}(\bigotimes_{i=1}^{3}\O_{A_{i}})$ with three insertions and a set $\bfD = (D_{I} \mid I\in\I )$ of regulating parameters,  as defined in sec.~\ref{sec:regCAG}. Here $\I$ always includes the set $\{1,2,3\}$, and it can include any number of two element subsets. We thus have the choices:
 \ben
 \bfD =
 \begin{cases}
 (D_{\{1,2,3\}}, D_{\{1,2\}}) & \text{or}\\
 (D_{\{1,2,3\}}, D_{\{1,3\}}) & \text{or}\\
 (D_{\{1,2,3\}}, D_{\{2,3\}}) & \text{or}\\
(D_{\{1,2,3\}}, D_{\{1,2\}}, D_{\{2,3\}}, D_{\{1,3\}})& .
 \end{cases}
 \een
 Since the first three are obviously related simply by a permutation of the insertions, and since the regulated CAG's
 have appropriate symmetry properties (from their FE and boundary conditions) we only need to consider one, say the first, case. A formula relating different choices of the set of regulating parameters is
 \ben\label{subsum}
 L^{\La,\Lao}_{(D_{\{1,2,3\}}, D_{\{1,2\}}, D_{\{2,3\}}, D_{\{1,3\}})} =
 L^{\La,\Lao}_{(D_{\{1,2,3\}}, D_{\{1,2\}})} +
 L^{\La,\Lao}_{(D_{\{1,2,3\}}, D_{\{2,3\}})} +
 L^{\La,\Lao}_{(D_{\{1,2,3\}}, D_{\{1,3\}})} \ .
 \een
 This formula can be proved by noting that the left and right sides satisfy the same
 FE and boundary conditions.
As a consequence, the bounds that we will derive for the quantities on the right side will
imply bounds on the quantity on the left side. Our bound is:

\begin{prop}\label{propCAG3}
Let $\bfD = (D_{\{1,2,3\}}, D_{\{1,2\}})$, with
regularization parameters such that $\delta_{1}:=[A_{1}]+[A_{2}]-D_{\{1,2\}}\geq 0$, $\delta_{2}:= D_{\{1,2\}}+[A_{3}]-D_{\{1,2,3\}}\geq 0$ and $D_{\{1,2,3\}}\leq D'=[A_{1}]+[A_{2}]+[A_{3}]$.
There exists a constant $K>0$ such that
\ben
\begin{split}
& \vspace{1cm} |\L^{\Lambda,\Lambda_{0}}_{\bfD, 2n,l}(\bigotimes_{i=1}^{3}\O_{A_{i}}(x_{i}); \vec{p})|\leq\,  \, K^{(4n+8l-3)(D'-D_{\{1,2,3\}})}\, K^{D'(n+2l)^{3}}\\
& \vspace{2cm} \times \ \frac{\Lambda^{D_{\{1,2,3\}}-2n}}{|x_{1}-x_{2}|^{\delta_{1}}\,{\rm max}(|x_{1}-x_3|, |x_{2}-x_3|)^{\delta_{2}}} \\
& \times \ \sqrt{|v_1|!|v_2|!|v_3|!(D'-D_{\{1,2,3\}})!}\
\sum_{\mu=0}^{{d}(3,n,l,0,D')}\frac{1}{\sqrt{\mu!}}\left(\frac{|\vec{p}|}{\Lambda}\right)^{\mu}\, \sum_{\lambda=0}^{2l+n+1}\frac{\log^{\lambda}(\sup(\frac{|\vec{p}|}{\kappa}), \frac{\kappa}{m})}{2^{\lambda}\lambda!}
\end{split}
\een
with $d(N,n,l,w,D'):= 2D'(n+l+2(N-1))+\sup(D'+1-2n-|w|, 0)$, and  $A_{i} \equiv \{n_{i},v_i\}$.
\end{prop}
\noindent
{\bf Remarks:} Here we can see explicitly how the regularization parameters allow us to remove (sub)\-divergences.
\begin{enumerate}
\item  The divergence on the total diagonal is of order $\delta_{1}+\delta_{2}=D'-D_{\{1,2,3\}}$. Thus, $D_{\{1,2,3\}}$ controls the degree of divergence on the total diagonal.

\item It follows by definition that as we increase $D_{\{1,2\}}$, we decrease $\delta_{1}$ and increase $\delta_{2}$. We see that in the bound this cures the divergence on the partial diagonal $x_{1}=x_{2}$ without changing the degree of divergence on the total diagonal. Hence, $D_{\{1,2\}}$ is indeed a regularization parameter associated to the partial diagonal $x_{1}=x_{2}$.

\item Note also that the bound above is regular on the other partial diagonals, $x_{2}=x_{3}$ and $x_{1}=x_{3}$ (where $x_{1}\neq x_{2}$). In other words, the decomposition of the CAG in eq.\eqref{subsum} is essentially a decomposition into contributions with a different nesting of subdivergences.
\end{enumerate}

\vspace{.5cm}
\noindent The proposition is an obvious consequence of the following theorem:

\begin{thm}\label{thmCAG3sub}
Let $x_{3}=0$ and $\bfD = (D_{\{1,2,3\}}, D_{\{1,2\}})$, with
regularization parameters such that $\delta_{1}:=[A_{1}]+[A_{2}]-D_{\{1,2\}}\geq 0$, $\delta_{2}:= D_{\{1,2\}}+[A_{3}]-D_{\{1,2,3\}}\geq 0$ and $D_{\{1,2,3\}}\leq D'=[A_{1}]+[A_{2}]+[A_{3}]$.
There exists a constant $K>0$ such that

\ben\label{boundCAG3}
\begin{split}
& \vspace{1cm} |\partial^{w}_{\vec{p}}\L^{\Lambda,\Lambda_{0}}_{\bfD, 2n,l}(\bigotimes_{i=1}^{3}\O_{A_{i}}(x_{i}); \vec{p})|\leq\,  \, K^{(4n+8l-3)(|w|+D'-D_{\{1,2,3\}})}\, K^{D'(n+2l)^{3}}\\
& \vspace{2cm} \times \ \frac{\Lambda^{D_{\{1,2,3\}}-2n-|w|}}{|x_{1}-x_{2}|^{\delta_{1}}\,{\rm max}(|x_{1}|, |x_{2}|)^{\delta_{2}}} \\
& \times \ \sqrt{|v_1|!|v_2|!|v_3|!(|w|+D'-D_{\{1,2,3\}})!}\
\sum_{\mu=0}^{{d}(3,n,l,w,D')}\frac{1}{\sqrt{\mu!}}\left(\frac{|\vec{p}|}{\Lambda}\right)^{\mu}\, \sum_{\lambda=0}^{2l+n+1}\frac{\log^{\lambda}(\sup(\frac{|\vec{p}|}{\kappa}), \frac{\kappa}{m})}{2^{\lambda}\lambda!}
\end{split}
\een
with $d(N,n,l,w,D'):= 2D'(n+l+2(N-1))+\sup(D'+1-2n-|w|, 0)$, and  $A_{i} \equiv \{n_{i},v_i\}$.
\end{thm}
\noindent\textbf{Remark:} Note that in the case $w=0$ we can make use of the translation covariance properties of the CAG's, eq.\eqref{CAGtrans}, in order to remove the condition $x_{3}=0$ and thus to obtain proposition \ref{propCAG3}. However, the translation operator $\exp(i x_{3} (p_{1}+\ldots+p_{2n}))$ evidently does not commute with the momentum derivatives $\partial^{w}_{\vec{p}}$, which is the reason why we did not include any momentum derivatives in proposition \ref{propCAG3}.

\begin{proof}[Proof of theorem \ref{thmCAG3sub}]
The FE with additional momentum derivatives for the CAG's under consideration is in the case at hand
[cf.~\eqref{FE3reg} and recall that we have $\I = (\{1,2,3\}, \{1,2\})$ in the case at hand]
\ben\label{FE3}
\begin{split}
&\hspace{3cm} \partial_{\Lambda}\partial_{\vec{p}}^{w}\L^{\Lambda,\Lambda_{0}}_{\bfD, 2n,l}(\bigotimes_{i=1}^{3}\O_{A_{i}}; p_{1},\ldots,p_{2n})=\\
&\vspace{0.05cm}\\
&= \left(\atop{2n+2}{2}\right)  \int_{k} \dot{C}^{\Lambda}(k)\partial_{\vec{p}}^{w}\L^{\Lambda,\Lambda_{0}}_{\bfD, 2n+2,l-1}(\bigotimes_{i=1}^{3}\O_{A_{i}}; k, -k,  p_{1},\ldots,p_{2n})\\
&-4\!\!\!\!\!\!\sum_{\atop{l_{1}+l_{2}=l}{n_{1}+n_{2}=n+1}}\!\!\!\!\! n_{1}n_{2}\,  \mathbb{S}\, \Bigg[\! \sum_{w_{1}+w_{2}+w_{3}=w}\!\!\!\!\!\!\!\!\! c_{\{w_{j}\}} \partial_{\vec{p}}^{w_{1}}\L^{\Lambda,\Lambda_{0}}_{\bfD, 2n_{1},l_{1}}(\bigotimes_{i=1}^{3}\O_{A_{i}}; q,  p_{1},\ldots,p_{2n_{1}-1})\\
&\hspace{8cm}\times \partial_{\vec{p}}^{w_{2}}\dot{C}^{\Lambda}(q)\,  \partial_{\vec{p}}^{w_{3}} \L^{\Lambda,\Lambda_{0}}_{2n_{2},l_{2}}(  p_{2n_{1}},\ldots,p_{2n}) \\
&+\!\!\! \sum_{w_{1}+w_{2}=w}\!\!\!\! c_{\{w_{j}\}} \int_{k}\, \partial_{\vec{p}}^{w_{1}}\L^{\Lambda,\Lambda_{0}}_{D_{\{1,2\}}, 2n_{1},l_{1}}(\O_{A_{1}}\otimes\O_{A_{2}}; k,  p_{1},\ldots,p_{2n_{1}-1})\\
&\hspace{8cm}\times  \dot{C}^{\Lambda}(k)\,   \partial_{\vec{p}}^{w_{2}}\L^{\Lambda,\Lambda_{0}}_{2n_{2},l_{2}}(\O_{A_{3}}; -k,   p_{2n_{1}},\ldots,p_{2n}) \Bigg]
\end{split}
\een
The proof utilizes the same general strategy as the proof of theorem \ref{thmCAG2}. In fact, to prove that the bound \eqref{boundCAG3} is satisfied by the first two terms on the r.h.s. of the FE~\eqref{FE3}, we can mimic the proof of theorem  \ref{thmCAG2}, making the following adjustments:
\begin{itemize}
\item Replace $D$ by $D_{\{1,2,3\}}$.
\item Replace $1/|x|^{D'-D}$ with $1/|x_{1}-x_{2}|^{\delta_{1}}\cdot  1/{\rm max}(|x_{1}|, |x_{2}|)^{\delta_{2}}$
\item Replace $d(2,n,l,w,D')$ by $d(3,n,l,w,D')$.
\end{itemize}
The first two changes do not affect the proof given in theorem \ref{thmCAG2}. To see that the third one does not cause any problems either, we note that $d(3,n,l,w,D')$ satisfies the following inequalities:
\ben
d(3,n+1,l-1,w,D')\leq d(3,n,l,w,D')
\een
\ben
 d(3,n_{1},l_{1},w,D') \leq d(3,n,l,w,D') \quad \text{ if } n_{1}+n_{2}=n+1, l_{1}+l_{2}=l \text{ and }n_{2}+l_{2} \geq 1
\een
\ben
d(3,n,l,w+\tilde{w},D')+|\tilde{w}| \leq d(3,n,l,w,D') \quad \text{ for }|\tilde{w}|\leq D'+1-2n-|w|
\een
Those were the only properties of $d$ that we needed in the part of the proof of theorem \ref{thmCAG2} regarding the first two terms on the r.h.s. of the FE \eqref{FE2}. Thus, following the same computational steps as in the proof of theorem \ref{thmCAG2}, we find that the contributions from the first two terms on the r.h.s. of the FE \eqref{FE3} satisfy the bound \eqref{boundCAG3}. We spare the reader the lengthy but straightforward repetitions of these calculations and
instead give the details how to find a suitable estimate for the source terms, i.e. the last line of the FE~\eqref{FE3}.

\begin{lemma}\label{lemmakint}
Let $x_{3}=0$ and fix regularization parameters such that $\delta_{1}:=[A_{1}]+[A_{2}]-D_{\{1,2\}}\geq 0$ and $\delta_{2}:= D_{\{1,2\}}+[A_{3}]-D_{\{1,2,3\}}\geq 0$.
Then we have the bound
\ben\label{lemma2bound}
\begin{split}
&\left| \partial^{w}_{\vec{p}}\int_{k}\, \Lsc_{D_{\{1,2\}}, 2n_{1},l_{1}}(\O_{A_{1}}\otimes\O_{A_{2}};k,p_{1},\ldots, p_{2n_{1}-1})\, \dot{C}^{\Lambda}(k)\,\Lsc_{ 2n_{2},l_{2}}(\O_{A_{3}};-k,p_{2n_{1}},\ldots, p_{2n})\right|\\
\leq&\quad  M\, K_{0}^{(4n+8l-3)(|w|+D'-D_{\{1,2,3\}})} K_{0}^{D'(n+2l)^{3}}\, \prod_{i=1}^{3}\sqrt{|v_i|!}\, \Lambda^{D_{\{1,2,3\}}-2n-|w|-1}\, \e^{-m^{2}/\Lambda^{2}}\\
&\times \frac{\sqrt{(|w|+D'-D_{\{1,2,3\}})!}}{|x_{1}-x_{2}|^{\delta_{1}}\,\max(|x_{1}|, |x_{2}|)^{\delta_{2}}}\sum_{\mu=0}^{{d}(N=3,n,l,w,D')}\frac{1}{\sqrt{\mu}!}\, \left(\frac{|\vec{p}|}{\Lambda}\right)^{\mu}\, \sum_{\lambda=0}^{2l+n}\frac{\log^{\lambda}(\sup(\frac{|\vec{p}|}{\kappa}), \frac{\kappa}{m})+\sqrt{\lambda!}}{2^{\lambda}\lambda!}
\end{split}
\een
where $M=
5^{|w|+\delta_{2}}\, 2^{2l+n+1}2^{2d}(2l+n+1)\ d\ 2^{[|w|+\delta_{2}]/2}$ and where $K_{0}$ is the constant (called $K$ there) from our bound on the CAG's with two insertions, see thm \ref{thmCAG2}.
\end{lemma}

\begin{proof}[Proof of lemma \ref{lemmakint}]
The proof is similar to that of lemma \ref{lemmakint2N2} above. Assume for the moment that $|x_{2}|>|x_{1}|$.  To begin with, we make use of the translation covariance property [see eq.\eqref{CAGtrans}] of the CAG's and perform partial integrations in order to bound the left hand side of \eqref{lemma2bound} by
\ben\label{PICAG3}
\begin{split}
&\Big| \partial^{w}_{\vec{p}}\int_{k}\, \Lsc_{D_{\{1,2\}}, 2n_{1},l_{1}}(\O_{A_{1}}(x_{1})\otimes\O_{A_{2}}(x_{2});k,p_{1},\ldots, p_{2n_{1}-1})\\
&\hspace{5cm} \times\ \dot{C}^{\Lambda}(k)\,\Lsc_{ 2n_{2},l_{2}}(\O_{A_{3}}(0);-k,p_{2n_{1}},\ldots, p_{2n}) \Big|\\
\leq&\,\Big|\int_{k} \sum_{w_{1}+w_{2}+w_{3}=w} c_{\{w_{i}\}} \  \e^{ikx_{2}} \ \e^{ix_{2}(p_{1}+\ldots+p_{2n_{1}-1})}  \\
&  \partial_{k}^{w_{3}}\Big( \partial_{\vec{p}}^{w_{1}} \Lsc_{D_{\{1,2\}}, 2n_{1},l_{1}}(\O_{A_{1}}(x_{1}-x_{2})\otimes\O_{A_{2}}(0);k,p_{1},\ldots, p_{2n_{1}-1})\\
&\hspace{5cm} \times\ \dot{C}^{\Lambda}(k)\, \partial_{\vec{p}}^{w_{2}}\Lsc_{ 2n_{2},l_{2}}(\O_{A_{3}}(0);-k,p_{2n_{1}},\ldots, p_{2n}) \Big) \Big|
\end{split}
\een
In analogy to the discussion following eq.\eqref{lemmakint2PI} we can multiply by $(2\|x_{2}\|)^{\delta_{2}}$, where we recall the notation $\| x \|=\max_{\alpha\in\{1,\ldots, 4\}}|x_{\alpha}|$, and write this factor as a $k$-derivative of $\e^{ikx_{2}}$. This $\delta_{2}$-fold derivative can then be moved onto the terms in the second line of \eqref{PICAG3} by partial integration. We can then further substitute the bounds  for the CAG's with one [cf. \eqref{boundCAG1}] and two insertions (cf. theorem \ref{thmCAG2}) into the resulting expression. After some bookkeeping, this yields the estimate
\ben
\begin{split}
&\text{''r.h.s. of \eqref{PICAG3}''} \ \times \ (2\|x_{2}\|)^{\delta_{2}} \le M_{0}\, K_{0}^{(4n+8l-3)(|w|+\delta_{1}+\delta_{2})} K_{0}^{D'(n+2l)^{3}}\, \sqrt{(|w|+\delta_{1}+\delta_{2})!|v_1|!|v_2|!}  \\ & \times\, \frac{\Lambda^{D_{\{1,2,3\}}-2n-|w|-5}}{|x_{1}-x_{2}|^{\delta_{1}}}\ \e^{-m^{2}/\Lambda^2} \
\int_{k} e^{-k^{2}/2\Lambda^{2}} \sum_{\mu=0}^{d_{12}+d_{3}}\frac{1}{\sqrt{\mu}!}\, \left(\frac{|\vec{p}|+|k|}{\Lambda}\right)^{\mu}\, \sum_{\lambda=0}^{2l+n}\frac{\log^{\lambda}(\sup(\frac{|\vec{p}|+|k|}{\kappa}), \frac{\kappa}{m})}{2^{\lambda}\lambda!}
\end{split}
\een
where $M_{0}=
5^{|w|+\delta_{2}}\, 2^{2l+n+1}2^{d_{12}+d_{3}}(2l+n+1)\ d \ 2^{[|w|+\delta_{2}]/2}$ and where we used the shorthand $d_{12}=d(2, n_{1},l_{1},w_{1}, [A_{1}]+[A_{2}])$ and $d_{3}=d(1, n_{2},l_{2},w_{2}, [A_{3}])$. We can replace the upper limit of summation over $\mu$ by $d(3,n,l,w,D')$, since
\ben
d(3,n,l,w,D')= 2D'(n+l+4)+\sup(D'+1-2n-|w|, 0)\geq 2D'(n+l+3)+2D'\geq d_{12}+d_{3}\, .
\een
The momentum integral can then be estimated as in \eqref{momentumint}, and after division by $(2\|x_{2}\|)^{\delta_{2}}$ and recalling that $|x_{2}|\leq 2 \|x_{2}\|$ we obtain the bound claimed in lemma  \ref{lemmakint} for $|x_{2}|>|x_{1}|$. If instead $|x_{2}|<|x_{1}|$ we can just repeat the discussion above with the indices $1$ and $2$ exchanged, which completes the proof.
\end{proof}
With this lemma at hand we can verify that the bound claimed in theorem \ref{thmCAG3sub} is satisfied by the $\Lambda'$ integral over the source terms. Indeed this integral can be bounded in the same manner as in the discussion of the source terms for the CAG's with two insertions:

\vspace{.5cm}
\noindent \textbf{(a) Irrelevant terms ($2n+|w|>D_{\{1,2,3\}}$):} Following the same computational steps as in the proof of theorem \ref{thmCAG2} [cf. \eqref{sourceAbound}] we find that the contribution from the source term is bounded by \eqref{boundCAG3}, provided $K$ is chosen large enough that
\ben
40 (l+1)(n+1)^3 \,M\, K_0^{(4n+8l-3)(|w|+D'-D_{\{1,2,3\}})}\, K_0^{D'(n+2l)^3} \le \frac{1}{8} K^{(4n+8l-3)(|w|+D'-D_{\{1,2,3\}})}\, K^{D'(n+2l)^3}\ .
\een

\vspace{.5cm}
\noindent \textbf{(b) Relevant terms ($2n+|w|\leq D_{\{1,2,3\}}$) at $\vec{p}=\vec{0}$:} Proceeding as in \eqref{3rdFEN2} and the subsequent discussion, we find that the contribution from the source term is bounded by \eqref{boundCAG3}, provided $K$ is chosen large enough that
\ben
\begin{split}
24 (l+1)(n+1)^3 (2l+n+1) \,M\, &K_0^{(4n+8l-3)(|w|+D'-D_{\{1,2,3\}})}\, K_0^{D'(n+2l)^3} \\
&\le \frac{1}{8} K^{(4n+8l-3)(|w|+D'-D_{\{1,2,3\}})}\, K^{D'(n+2l)^3}\ .
\end{split}
\een

\vspace{.5cm}
\noindent \textbf{(c) Relevant terms ($2n+|w|\leq D_{\{1,2,3\}}$) at $\vec{p}\neq \vec{0}$:} Using the Taylor expansion with remainder technique as in the proof of theorem \ref{thmCAG2} [cf. eq.\eqref{Taylornonzero} and subsequent discussion], we find that the contribution from the source term is bounded by \eqref{boundCAG3}, provided $K$ is chosen large enough that
\ben
\begin{split}
\sum_{j \le D+1-2n-|w|} &48 M (2l+n)(l+1)(n+1)^{3} \, K_{0}^{(4n+8l-3)(|w|+j+D'-D_{\{1,2,3\}})+D'(2l+n)^{3}}\\
&\times (8n)^{j}\, j\,   2^{\frac{d(2,n,l,w,D')}{2}}
\, 2^{\frac{j+|w|+D'-D_{\{1,2,3\}}}{2}} \leq K^{(4n+8l-3)(|w|+D'-D_{\{1,2,3\}})+D'(2l+n)^{3}}
\end{split}
\een
Since all these inequalities hold for a sufficiently large choice of $K$, we have finished the proof of the theorem.
\end{proof}
\noindent Since we eventually want to remove the cutoffs, i.e. take the limits $\Lambda\to 0, \Lambda_{0}\to\infty$, the following result will be useful:

\begin{corollary}\label{corCAG3}
Let $x_{3}=0$ and choose $\bfD, \delta_{1}$ and $\delta_{2}$ as in proposition \ref{propCAG3} and let $\Lambda\leq m$.
There exists a constant $K>0$ such that
\ben
\begin{split}
& \vspace{1cm} |\partial^{w}_{\vec{p}}\L^{\Lambda,\Lambda_{0}}_{\bfD, 2n,l}(\bigotimes_{i=1}^{3}\O_{A_{i}}(x_{i}); \vec{p})|\leq\,  \, K^{(4n+8l-3)(|w|+D'-D_{\{1,2,3\}})}\, K^{D'(n+2l)^{3}}\\
& \vspace{2cm} \times \ \frac{m^{D_{\{1,2,3\}}-2n-|w|}}{|x_{1}-x_{2}|^{\delta_{1}}\,{\rm max}(|x_{1}-x_3|, |x_{2}-x_3|)^{\delta_{2}}}\ \sqrt{(2n+|w|-D_{\{1,2,3\}})_{+}!} \\
& \times \ \sqrt{|v_1|!|v_2|!|v_3|!(|w|+D'-D_{\{1,2,3\}}) !}\
\sum_{\mu=0}^{{d}(3,n,l,w,D')}\left(\frac{|\vec{p}|}{m}\right)^{\mu}\, \sum_{\lambda=0}^{2l+n+1}\frac{\log_{+}^{\lambda}(\frac{|\vec{p}|}{m})}{2^{\lambda}\lambda!}
\end{split}
\een
with $d(N,n,l,w,D'):= 2D'(n+l+2(N-1))+\sup(D'+1-2n-|w|, 0)$, and  $A_{i} \equiv \{n_{i},v_i\}$.
\end{corollary}
\noindent\textbf{Remark:} This bound  implies convergence of the multivariate Taylor expansion of the CAG's with three insertions in a neighborhood of  $\vec{x}=(x_{1},x_{2},x_{3})$ (where $x_{i}\neq x_{j}$) for any choice of regularization $\bfD$. This can be seen from the fact that the bound grows like
\ben
j! \Big(\sup(|\vec{p}|/m, 1)^{2(n+l+4)+1} \tilde{K}/\min_{i\neq j}|x_{i}-x_{j}|\Big)^{j}
\een
if we take $j$ derivatives with respect to $\vec{x}=(x_{1},x_{2},x_{3})$.
[Here we also used the Lowenstein rule \eqref{loew2}.] Consequently, the CAG's
with or without cutoffs are analytic functions of their spacetime arguments for
non-coinciding configurations. Since the OPE coefficients are derived from the CAG's [see
defn.~\ref{defOPE}],
the same holds true for them.

\medskip
\noindent
Corollary \ref{corCAG3} follows by inserting the bound from theorem \ref{propCAG3} once more into the FE and integrating over $\Lambda$ between $0$ and $m$ (cf. proof of corollary \ref{CAG2cor}).

\subsection{Amputated Greens functions with $N=3$ insertions}

We now give bounds for the regulated AG's with three insertions, $G^{\La,\Lao}_D(\O_{A_1} \otimes \O_{A_2} \otimes \O_{A_3})$, see eq.~\eqref{GDdef}, or more precisely, certain combinations of Taylor coefficients of their moments
${\mathcal G}^{\La,\Lao}_{n,l,D}$, where we recall the definition of the Taylor expansion operator in
eq.~\eqref{defT}. The result is:

\begin{thm}\label{thmGTaylor}

Let $D=[A_1]+[A_2]+[A_3]+\Delta$, where $\Delta \ge 0$, and let $\Lambda\leq m$. Then
\ben\label{GDbound}
\begin{split}
&\Big|
(1-\sum_{j\leq\Delta}\T^{j}_{(x_{1},x_{2},x_{3})\to(x_{3},x_{3},x_{3})})\, \G^{\Lambda,\Lambda_{0}}_{D,2n,l}(\bigotimes_{i=1}^{3}\O_{A_{i}}(x_i); \vec{p})   \Big| \\
\le &\ \ m^{ -2n-1}\
\frac{(\tilde{K}\  m \ {\rm max}(|x_1-x_2|,|x_2-x_3|,|x_1-x_3|))^{\sum [A_i]+1+\Delta}}{  {\rm min}(|x_1-x_2|,|x_2-x_3|,|x_1-x_3|)^{\sum [A_i]+1}}\\
& \ \ \times \ \sup(1,\frac{|\vec{p}|}{m})^{d(3,n,l,w=0,D+D')}\
\frac{\prod_{i=1}^{3}[A_{i}]!}{\sqrt{\Delta!}} \sum_{\lambda=0}^{2l+n+1}\frac{\log^{\lambda}_{+}(\frac{|\vec{p}|}{m})}{2^{\lambda}\lambda!}
\end{split}
\een
where $\tilde{K}>0$ is a constant depending on $n$ and $l$, and where $A_i \equiv \{n_i, v_i\}$.
\end{thm}

\begin{proof}
Note that, by eq.~\eqref{GDdef}, the quantities $G^{\La,\Lao}_D(\O_{A_1} \otimes \O_{A_2} \otimes \O_{A_3})$
and $F^{\La,\Lao}_D(\O_{A_1} \otimes \O_{A_2} \otimes \O_{A_3})$ only differ by products of CAG's with one
insertion, for which we already have given corresponding bounds in sec.~\ref{sec:CAG1}.
Thus, the task boils down to bounds on $
(1-\sum_{j\leq\Delta}\T^{j}_{(x_{1},x_{2},x_{3})\to(x_{3,}x_{3},x_{3})}) \F^{\La,\Lao}_{n,l,D}$. Although this quantity
satisfies a FE, we have not been able to derive our bounds straightforwardly from that FE in the manner described above. The technical obstruction seems to be that one cannot make use of the partial integration trick that has been applied in all previous proofs to bound the momentum integrals, see e.g. eq.\eqref{lemmakint2PI}. This is due to the fact that now we have more than just two CAG's in the source terms, but only two of them are integrated over.

To get around this, we will prove a decomposition of $F^{\La,\Lao}_D$ (and correspondingly its moments $\F^{\La,\Lao}_{D,n,l}$),
into quantities that we already know how to estimate.

\begin{lemma}[Decomposition Lemma]\label{lem:decomp}

Let $x_{3}=0$. For any $D\geq -1$ the functionals  $F^{\La,\Lao}_D$ can be decomposed as
\ben\label{eqdecomp}
(1-\sum_{j\leq D-D'}\T^{j}_{(x_{1},x_{2})\to(0,0)})\ F_{D}^{\Lambda,\Lambda_{0}}(\bigotimes_{i=1}^{3}\O_{A_{i}})=F_1+F_2+F_3 \ ,
\een
where $F_1,F_2,F_3$ are the functions of $x_1,x_2$ given by
\ben\label{eqdecomp1}
\begin{split}
&F_1= \\
&(1-\sum_{j\leq D-[A_{1}]}\hspace{-.3cm} \T_{(x_1,x_2)\to (0,x_2)}^{j})\\
&\quad\times\left[ L^{\Lambda,\Lambda_{0}}_{(D_{\{1,2,3\}}=D,\,  D_{\{2,3\} } =-1)}(\bigotimes_{i=1}^{3}\O_{A_{i}}) - L^{\Lambda,\Lambda_{0}}(\O_{A_2} \otimes \O_{A_3})\, L^{\Lambda,\Lambda_{0}}(\O_{A_{{1}}}) \right]\\
&+\sum_{j_{1}\leq D-[A_{1}]} \T_{(x_1,x_2)\to (0,x_2)}^{j_{1}} (1-\sum_{j_{2}\leq D-D'-j_{1}}\T_{(x_1,x_2)\to (x_1,0)}^{j_{2}}) \\
&\quad\times\left[ L^{\Lambda,\Lambda_{0}}_{(D_{\{1,2,3\}}=D,\, D_{\{2,3\}}=D-[A_{1}]-j_{1})}(\bigotimes_{i=1}^{3}\O_{A_{i}}) - L^{\Lambda,\Lambda_{0}}_{D-[A_{1}]-j_{1}}(\O_{A_2} \otimes \O_{A_3})L^{\Lambda,\Lambda_{0}}(\O_{A_{1}}) \right] \ .
\end{split}
\een
The function $F_2$ is given by the same expression replacing $A_1 \to A_2, D_{\{2,3\}}
\to D_{\{1,3\}}$ and $\T_{(x_1,x_2)\to (0,x_2)}\to \T_{(x_1,x_2)\to (x_{1},0)}$. The function $F_3$ is given by
\ben\label{eqdecomp2}
\begin{split}
&F_3= \\
&(1-\sum_{j\leq D-[A_{1}]}\hspace{-.3cm} \T_{(\frac{x_1+x_2}{\sqrt{2}},\frac{x_1-x_2}{\sqrt{2}})\to (0,\frac{x_1-x_2}{\sqrt{2}})}^{j})\\
&\quad\times\left[ L^{\Lambda,\Lambda_{0}}_{(D_{\{1,2,3\}}=D,\,  D_{\{1,2\} } =-1)}(\bigotimes_{i=1}^{3}\O_{A_{i}}) - L^{\Lambda,\Lambda_{0}}(\O_{A_1} \otimes \O_{A_2})\, L^{\Lambda,\Lambda_{0}}(\O_{A_{{3}}}) \right]\\
&+\sum_{j_{1}\leq D-[A_{3}]} \T_{(\frac{x_1+x_2}{\sqrt{2}},\frac{x_1-x_2}{\sqrt{2}})\to (0,\frac{x_1-x_2}{\sqrt{2}})}^{j_{1}} (1-\sum_{j_{2}\leq D-D'-j_{1}}\T_{(\frac{x_1+x_2}{\sqrt{2}},\frac{x_1-x_2}{\sqrt{2}})\to (\frac{x_1+x_2}{\sqrt{2}},0)}^{j_{2}}) \\
&\quad\times\left[ L^{\Lambda,\Lambda_{0}}_{(D_{\{1,2,3\}}=D,\, D_{\{1,2\}}=D-[A_{3}]-j_{1})}(\bigotimes_{i=1}^{3}\O_{A_{i}}) - L^{\Lambda,\Lambda_{0}}_{D-[A_{3}]-j_{1}}(\O_{A_1} \otimes \O_{A_2})L^{\Lambda,\Lambda_{0}}(\O_{A_{3}}) \right] \ .
\end{split}
\een
Here it is understood that we express the CAG's as functions of $(x_{1}+x_{2}, x_{1}-x_{2})$.
\end{lemma}
\noindent \textbf{Remark:} The CAG's are of course translation invariant; the restriction to $x_{3}=0$ is made only for convenience here.

\vspace{.5cm}

\noindent
The lemma can be proved by an induction in $D$, using at each step that both sides satisfy
the same FE and boundary conditions, and using~\eqref{subsum}. We omit the proof here. Continuing the proof of theorem~\ref{thmGTaylor}, we now expand both sides of the lemma
into a power series in $\varphi$ and $\hbar$.
Recall also that the remainder of a Taylor expansion can be expressed as the following integral:
\ben\label{Rfacttaylor}
\begin{split}
&(1-\sum_{j\leq\Delta}\T^{j}_{\vec x\to\vec y})f(x_1, \dots, x_N) =\\
&\sum_{|v|=
\Delta+1}\frac{(\vec{x}-\vec{y})^{v}}{\Delta!}\int_{0}^{1}\d\tau\, (1-\tau)^{\Delta}\, \partial^{v}\Big[
f(y_{1}+\tau (x_1-y_{1}), \dots, y_{N}-\tau( x_N-y_{N})) \Big] \ .
\end{split}
\een
Using that formula in the decomposition lemma and applying the Lowenstein rules, eqs. \eqref{loew1}, \eqref{loew2} and \eqref{loew3}, as well as the previous bounds on the CAG's with insertions, i.e. \eqref{boundCAG1m} and corollaries \ref{CAG2cor}  and \ref{corCAG3},
we obtain straightforwardly the bound stated in the theorem, but on $\F^{\La,\Lao}_{n,l,D}$, rather than
${\mathcal G}^{\La,\Lao}_{n,l,D}$. As we said, these only differ by products of CAG's with one insertion,
for which we have the bound \eqref{boundCAG1} again. This yields the statement of the theorem, after using also translation invariance in order to proceed to $x_{3}\neq 0$.
As an illustrative example, we estimate a term in $F_3$. Let $\mathbf{D}=(D_{\{1,2,3\}}= D, D_{\{1,2\}}=-1)$.
\begin{equation}\label{eq:thm5ex}
\begin{split}
&\Big|(1-\sum_{j\leq D-[A_{3}]} \T^{j}_{(\frac{x_{1}+x_{2}}{\sqrt{2}},\frac{x_{1}-x_{2}}{\sqrt{2} })\to (0,\frac{x_{1}-x_{2}}{\sqrt{2}})} )\,\L^{\Lambda,\Lambda_{0}}_{\mathbf{D},2n,l}(\otimes_{i=1}^{3}\O_{A_{i}}; \vec{p})\Big|\\
\leq&\frac{|x_{1}+x_{2}|^{D-[A_{3}]+1}}{(D-[A_{3}])!}\sum_{|w|=D+1-[A_{3}]} \Big| \int_{0}^{1}\d\tau (1-\tau)^{|w|-1}\partial_{\tau\left(\frac{x_{1}+x_{2}}{2}\right)}^{w}\\
&\times  \left(e^{i\tau\frac{x_{1}+x_{2}}{2}(p_{1}+\ldots+p_{2n})}  \L^{\Lambda,\Lambda_{0}}_{\mathbf{D},2n,l}(\O_{A_{1}}(\frac{x_{1}-x_{2}}{2})\otimes \O_{A_{2}}(-\frac{x_{1}-x_{2}}{2})\otimes \O_{A_{3}}(-\tau\frac{x_{1}+x_{2}}{2}); \vec{p})\right)\Big|\\
\leq&\, \frac{|x_{1}+x_{2}|^{D-[A_{3}]+1}}{(D-[A_{3}])!}\!\!\!\!\!\!\! \sum_{|w_{1}+w_{2}|=D+1-[A_{3}]}\!\!\!\!\!\!\!\!\!\!\!\!\! c_{\{w_{j}\}} |\vec{p}|^{|w_{1}|}\, \Big| \L^{\Lambda,\Lambda_{0}}_{\mathbf{D},2n,l}(\O_{A_{1}}(\frac{x_{1}-x_{2}}{2})\otimes \O_{A_{2}}(-\frac{x_{1}-x_{2}}{2})\otimes \partial^{w_{2}}\O_{A_{3}}(0); \vec{p})\Big|
\end{split}
\end{equation}
Here we expressed the remainder of the Taylor expansion through eq.\eqref{Rfacttaylor} and made use of the covariance property of the CAG's, eq.\eqref{CAGtrans}. Using the bound on the CAG's with three insertions given in corollary \ref{corCAG3}, we obtain:
\begin{equation}\label{eq:thm5ex2}
\begin{split}
&\Big|(1-\sum_{j\leq D-[A_{3}]} \T^{j}_{(\frac{x_{1}+x_{2}}{\sqrt{2}},\frac{x_{1}-x_{2}}{\sqrt{2} })\to (0,\frac{x_{1}-x_{2}}{\sqrt{2}})} )\,\L^{\Lambda,\Lambda_{0}}_{\mathbf{D},2n,l}(\otimes_{i=1}^{3}\O_{A_{i}}; \vec{p})\Big|\\
\leq&\, \frac{|x_{1}+x_{2}|^{D-[A_{3}]+1}}{(D-[A_{3}])!}\, 2^{D+1-[A_{3}]}\, K^{(D'+D-[A_{3}]+1)(2l+n)^{3}+(4n+8l-3)([A_{1}]+[A_{2}]+1)}\, \frac{m^{D-2n}\, \sqrt{(2n-D)_{+}!} }{|x_{1}-x_{2}|^{[A_{1}]+[A_{2}]+1}}\\
\times& \sqrt{|v_{1}|!|v_{2}|!(|v_{3}|+D-[A_{3}]+1)! ([A_{1}]+[A_{2}]+1)!}\, (\sup(1,\frac{\vec{p}}{m}))^{d(3,n,l,0,D'+D-[A_{3}]+1)}\sum_{\lambda=0}^{2l+n+1}\frac{\log_{+}^{\lambda}(\frac{|\vec{p}|}{m})}{2^{\lambda}\lambda!}\\
\leq &\, m^{-2n-1} \frac{\Big(\tilde{K} m \max(|x_{1}|,|x_{2}|)\Big)^{D+1}}{\min(|x_{1}-x_{2}|, |x_{1}|,|x_{2}|)^{D'+1}}
\frac{\prod_{i=1}^{3}[A_{i}]!}{\sqrt{(D-[A_{3}])!}}\\
&\times (\sup(1,\frac{\vec{p}}{m}))^{d(3,n,l,0,D'+D-[A_{3}]+1)}\sum_{\lambda=0}^{2l+n+1}\frac{\log_{+}^{\lambda}(\frac{|\vec{p}|}{m})}{2^{\lambda}\lambda!}\
\end{split}
\end{equation}
where $\tilde{K}$ is some constant depending on $n$ and $l$. Using translation covariance we obtain a bound compatible with inequality \eqref{GDbound} also for $x_{3}\neq 0$. The other terms on the r.h.s. of the decomposition lemma can be bounded in a similar fashion.
\end{proof}

\section{Proof of thm.~\ref{OPEbound}}\label{secOPEconv}

\setcounter{thm}{0}
We repeat the statement of the theorem for convenience:
\begin{thm}\label{OPEbound}
Let $f_p(x)$ be any smooth function on $x \in \mr^4$ such that the support of the
Fourier transform $\hat f_p(q)$ is contained in a ball $|p-q| \le 1$. Define the
smeared spectator fields by
$$
\varphi(f_p) \equiv \int d^4 x \ \varphi(x)  \ f_p(x) \ .
$$
Then the remainder of the OPE, carried out up to operators of dimension $D=\sum_{i=1}^{3}[A_{i}]+\Delta$, at $l$-loops, is bounded by
\ben
\begin{split}
&\Big| \Big\bra \O_{A_{1}}(x_{1})\O_{A_{2}}(x_{2})\O_{A_{3}}(x_3)\, \varphi(f_{p_{1}})\cdots\varphi(f_{p_{n}}) \Big\ket\\
 &\qquad\qquad- \sum_{C:[C] \leq D}\C_{A_{1}A_{2} A_{3}}^{C}(x_{1},x_{2},x_{3}) \Big\bra \O_{C}(x_3)\, \varphi(f_{p_{1}})\cdots\varphi(f_{p_{n}}) \Big\ket  \Big|\\
 \leq\quad& m^{n-1}\, \prod_{i=1}^{3}[A_{i}]!\ \prod_{j} \sup|\hat{f}_{p_{j}}|\, \sup(1,\frac{|\vec{p}|_{n}}{m})^{(4\sum [A_i] + 2\Delta)(n+l+9/2)+3n}
 \\ &\quad \times
 \sum_{\lambda=0}^{2l+n/2+1}\frac{\log^{\lambda}\sup(1,\frac{|\vec{p}|_{n}}{m})}{2^{\lambda}\lambda!}
 \\
  & \quad \times \frac{1}{\sqrt{\Delta!}}
  \ \ \frac{(\tilde{K}\ m\ {\rm max}(|x_1-x_2|,|x_2-x_3|,|x_1-x_3|) )^{\sum [A_{i}]+1+\Delta}}{( {\rm min}(|x_1-x_2|,|x_2-x_3|,|x_1-x_3|) )^{\sum [A_i]+1}}
\end{split}
\label{ope3}
\een
in $g\varphi^4$-theory.
Here 
$[A]$
denotes the canonical dimension of a composite field ${\mathcal O}_A$ as in eq.~\eqref{compop}.
$\tilde{K}$ is a constant depending on $n$ and $l$.
\end{thm}
\begin{proof}
A part of the proof can be generalized to $N$ field insertions without any additional effort, so we will keep $N\geq 2$ arbitrary for as long as possible.
Define the remainder functional
\ben\label{Rdef}
R_{D}^{\Lambda,\Lambda_{0}}(\bigotimes_{i=1}^{N}\O_{A_{i}}(x_{i}))\, :=\, G^{\Lambda,\Lambda_{0}}(\bigotimes_{i=1}^{N}\O_{A_{i}}(x_{i}))-\sum_{C:[C]\leq D}\, \C_{A_{1}\ldots A_{N}}^{C}(x_{1},\ldots, x_{N-1})\, L^{\Lambda,\Lambda_{0}}(\O_{C}(x_{N}))
\een
which allows us to write (for the theory with UV and IR cutoffs $\Lao$ and $\La$)
\ben\label{OPER}
\begin{split}
&\Big| \Big\bra \O_{A_{1}}(x_{1})\cdots\O_{A_{N-1}}(x_{N-1})\O_{A_{N}}(x_{N})\, \varphi(f_{p_{1}})\cdots\varphi(f_{p_{n}}) \Big\ket\\
 &\qquad\qquad- \sum_{C:[C]-D'\leq \Delta}\C_{A_{1}\ldots A_{N}}^{C}(x_{1},\ldots,x_{N}) \Big\bra \O_{C}(x_{N})\, \varphi(f_{p_{1}})\cdots\varphi(f_{p_{n}}) \Big\ket  \Big|\\
 =&\sum_{\atop{I_{1}\cup\ldots\cup I_{j}=\{1,\ldots, n\}}{l_{1}+\ldots+l_{j}=l}} \int_{\vec q}  \R^{\La,\Lao}_{D,|I_{1}|,l_{1}}(\bigotimes_{i=1}^{N}\O_{A_{i}}(x_{i}); \vec{q}_{I_{1}})\, \bar\L^{\La,\Lambda_{0}}_{|I_{2}|, l_{2}}(\vec{q}_{I_{2}})\cdots \bar\L^{\La,\Lambda_{0}}_{|I_{j}|, l_{j}}(\vec{q}_{I_{j}}) \, \prod_{i=1}^{n}C^{\La,\Lambda_{0}}(q_{i}) \hat f_{p_i}(q_i)
\end{split}
\een
where $\bar\L^{\La,\Lambda_{0}}_{n,l}$ are the expansion coefficients of the generating functional $\bar{L}^{\La,\Lambda_{0}}(\varphi)=L^{\La,\Lambda_{0}}(\varphi)+\frac{1}{2}\bra \varphi,\, (C^{\La,\Lambda_{0}})^{-1}\star\varphi \ket$ without the momentum conservation delta functions taken out.
We wish to find a bound for the above expression.
Since we already have bounds on $\Lsc_{n,l}$ from \eqref{propout}, and since $C^{\La,\Lambda_{0}}$ can be estimated trivially as $C^{\Lambda,\Lambda_{0}}(p)\leq [\sup(m,|p|)]^{-2}$, we will be concerned with $\R_{D,n,l}^{\La,\Lambda_{0}}$ in the following. The following lemma will allow us to express $\R_{D,n,l}^{\La,\Lambda_{0}}$ in terms of quantities with known bounds as given in the previous sections:

\begin{lemma}\label{remainder}
The remainder functionals satisfy
\ben\label{RGD}
\begin{split}
R_{D}^{\Lambda,\Lambda_{0}}(\bigotimes_{i=1}^{N}\O_{A_{i}}(x_{i}))
&= (1-\sum_{j\leq\Delta}\T^{j}_{\vec{x}\to(x_{N},\dots, x_{N})})\,  {G}^{\Lambda,\Lambda_{0}}_{D} (\bigotimes_{i=1}^{N}\O_{A_{i}}(x_{i}))
\end{split}
\een
with $\Delta=D-\sum_{i=1}^{N}[A_{i}]$.
\end{lemma}
This lemma reduces to lemma 4.1 of \cite{Hollands:2011gf} in the case $N=2$. The proof is given in Appendix~\ref{app:remainder}.

Hence, in order to bound the remainder of the OPE we have to find an estimate on the remainder of the Taylor expansion of the regularized AG's. Here we have theorem \ref{thmGTaylor} in the $N=3$ case. Substituting this bound along with the bound~\eqref{boundCAG0m} on the CAG's without operator insertions into eq.\eqref{OPER}, and also using $C^{\Lambda,\Lambda_{0}}(p)\leq [\sup(m,|p|)]^{-2}$, we obtain the statement of the theorem (note that the resulting bound is independent of $\Lambda\leq m$ and $\Lambda_{0}$, so the cutoffs can be removed safely).
\end{proof}

\section{Proof of thm.~\ref{thmfact}}\label{secFactor}

For convenience, we repeat the statement of this theorem as well:
\begin{thm}\label{thmfact}
Up to any arbitrary but fixed loop order $l$ in $g\varphi^4$-theory, the identity
\ben\label{ope3}
\C_{A_1 A_2 A_3}^B(x_1, x_2, x_3) = \sum_C \C_{A_1 A_2}^C(x_1, x_2) \ \C_{CA_3}^B(x_2, x_3) \ .
\een
holds for all configurations satisfying
$$\frac{|x_{1}-x_2|}{|x_2-x_{3}|} < \frac{1}{\tilde K}$$
for some (sufficiently large) constant $\tilde K>0$ (depending on $l,B$).
\end{thm}

\setcounter{thm}{5}

\noindent
{\em Proof:}
The proof has two main steps, which will be discussed in more detail below:

\begin{enumerate}
\item Show that the OPE still converges on the spacetime domain $\frac{|x_{1}-x_2|}{|x_2-x_{3}|} < \frac{1}{\tilde K}$ when performed in successive steps.
\item Show that this implies the relation \eqref{ope3} on the level of OPE coefficients.
\end{enumerate}
In the next two sections we will focus on these two issues.

\subsection{Partial OPE}

We have already shown that the remainder
\ben
G^{\Lambda,\Lambda_{0}}(\bigotimes_{i=1}^{3}\O_{A_{i}})-\sum_{[C]\leq D} \C_{A_{1}A_{2} A_{3}}^{C}\, L^{\Lambda,\Lambda_{0}}(\O_{C})
\een
 goes to zero as $D\to\infty$. Instead of expanding all three operator insertions, we now consider a similar expansion in just two of these operators, say $\O_{A_{1}}$ and $\O_{A_{2}}$, and leave the other one, $\O_{A_{3}}$, untouched, namely the expression
 \ben\label{partialremainder}
 G^{\Lambda,\Lambda_{0}}(\bigotimes_{i=1}^{3}\O_{A_{i}}(x_{i}))-\!\! \sum_{[C]\leq D}\!\! \C_{A_{1}A_{2}}^{C}(x_{1}-x_{2})\, G^{\Lambda,\Lambda_{0}}(\O_{C}(x_{2})\otimes \O_{A_{3}}(x_{3})) \ .
 \een
 The following lemma will allow us to bound the remainder of this partial OPE.

\begin{lemma} \label{partremainder}
Fix $D \ge -1$, and define regularization parameters as
$$\bfD = (D_{\{1,2,3\}}=-1,D_{\{2,3\}}=-1, D_{\{1,3\}}=-1,D_{\{1,2\}}=D) \ .$$
The remainder of the partial OPE can be expressed as
\ben\label{eq:partrem}
\begin{split}
&G^{\Lambda,\Lambda_{0}}(\bigotimes_{i=1}^{3}\O_{A_{i}}(x_{i}))-\!\! \sum_{[C]\leq D}\!\! \C_{A_{1}A_{2}}^{C}(x_{1}-x_{2})\, G^{\Lambda,\Lambda_{0}}(\O_{C}(x_{2})\otimes \O_{A_{3}}(x_{3}))\\
&=(1-\sum_{j=0}^{D-[A_{1}]-[A_{2}]} \T^{j}_{(x_{1},x_{2},x_{3})\to (x_{2},x_{2},x_{3})})\\
 &\quad\times\Big[ L^{\Lambda,\Lambda_{0}}_{\bfD}(\bigotimes_{i=1}^{3}\O_{A_{i}})- L^{\Lambda,\Lambda_{0}}_{D}(\O_{A_{1}}\otimes\O_{A_{2}})L^{\Lambda,\Lambda_{0}}(\O_{A_{3}})-L^{\Lambda,\Lambda_{0}}(\O_{A_{2}}\otimes\O_{A_{3}})L^{\Lambda,\Lambda_{0}}(\O_{A_{1}})\\
&\qquad- L^{\Lambda,\Lambda_{0}}(\O_{A_{1}}\otimes\O_{A_{3}})L^{\Lambda,\Lambda_{0}}(\O_{A_{2}})+L^{\Lambda,\Lambda_{0}}(\O_{A_{1}})L^{\Lambda,\Lambda_{0}}(\O_{A_{2}})L^{\Lambda,\Lambda_{0}}(\O_{A_{3}})\Big] \ .
\end{split}
\een
\end{lemma}
\noindent
The proof of this lemma can be found in appendix \ref{app:partrem}.
Lemma~\ref{partremainder}, combined with our bounds on the CAG's [see \eqref{boundCAG1}, theorem~\ref{thmCAG2} and theorem~\ref{thmCAG3sub}], allows us to estimate the remainder of the partial OPE.

\begin{thm}\label{partOPEbound}
Assume $|x_{1}-x_{2}|\leq |x_{2}-x_{3}|$ and let $\Lambda\leq m$. There exists a constant $\tilde K>0$ depending on $n$ and $l$, such that for all $D-[A_{1}]-[A_{2}]=\Delta$
\ben\label{partOPEboundeq}
\begin{split}
&\Big|\partial_{\vec{p}}^{w}
\Big( \G_{2n,l}^{\Lambda,\Lambda_{0}}(\bigotimes_{i=1}^{3}\O_{A_{i}}(x_{i});\vec{p})-\!\! \sum_{[C]\leq D}\!\! \C_{A_{1}A_{2}}^{C}(x_{1}-x_{2})\, \G_{2n,l}^{\Lambda,\Lambda_{0}}(\O_{C}(x_{2})\otimes \O_{A_{3}}(x_{3}); \vec{p}) \Big)
\Big|\\
&\leq \, m^{-2n-|w|-1 }\, \tilde{K}^{|w|}|w|! \ \prod_{i=1}^{3}[A_{i}]!\ \sum_{\lambda=0}^{2l+n+1}\frac{\log^{\lambda}_{+}(\frac{|\vec{p}|}{m})}{2^{\lambda}\lambda!}\ \sup(1,m |\vec{x}|)^{|w|}\\
&\qquad\times \left(\frac{\tilde{K}\ \sup(\frac{|\vec{p}|}{m},1)^{2n+2l+5}}{\min(|x_{2}-x_{3}|,|x_{1}-x_{3}|, 1/m)}\right)^{D'+1}  \\
&\qquad\times \frac{\Big(\tilde{K}\,  \sup(\frac{|\vec{p}|}{m},1)^{2n+2l+5} \,|x_{1}-x_{2}| \Big)^{\Delta} }{ \min\Big(|x_{2}-x_{3}|^{\Delta-1}\cdot \min(|x_{2}-x_{3}|,|x_{1}-x_{3}|), \sqrt{\Delta!}/m^{\Delta}\Big)}\,
\end{split}
\een
where $D'=[A_{1}+[A_{2}]+[A_{3}]$.
\end{thm}
\begin{proof}
We use lemma \ref{partremainder} to express the l.h.s. as the remainder of a Taylor expansion in $x_{1}-x_{2}$. We write this Taylor expansion as $(1-\sum_{j\leq \Delta}\T^{j})=(1-\sum_{j\leq \Delta-1}\T^{j})-\T^{\Delta}$ and apply again the well known formula \eqref{Rfacttaylor} for the term in brackets on the r.h.s.. For the expressions on the r.h.s. of lemma \ref{partremainder} we can insert our bounds from \eqref{boundCAG1m}, corollary \ref{CAG2cor} and corollary \ref{corCAG3}. Note that these bounds are always given for configurations where one spacetime argument is zero. Hence, in order to be able to use our bounds we also have to make use of the translation properties of the CAG's. For example, we can write
\ben
\L^{\Lambda,\Lambda_{0}}_{2n,l}(\O_{A_{1}}(x_{1})\otimes\O_{A_{3}}(x_{3}); \vec{p})=\e^{i x_{3}(p_{1}+\ldots+p_{2n})}\L^{\Lambda,\Lambda_{0}}_{2n,l}(\O_{A_{1}}(x_{1}-x_{3})\otimes\O_{A_{3}}(0); \vec{p})
\een
and apply corollary \ref{CAG2cor} for the term on the right. Note that the momentum derivatives can either act on the CAG's or on the exponential factors. Derivatives on the latter simply give rise to  powers of the spacetime variables. This is accounted for by the factor $\sup(1,m |\vec{x}|)^{|w|}$ in the bound \eqref{partOPEboundeq}. When estimating the $\tau$-integral in the formula for remainder of the Taylor expansion we also make use of our assumption $|x_{1}-x_{2}|\leq|x_{2}-x_{3}|$ to replace expressions like
\ben\label{xbehav1}
\frac{(1-\tau)\, |x_{1}-x_{2}|}{|x_{2}-x_{3}+\tau (x_{1}-x_{2})|}=\frac{|x_{1}-x_{2}|}{|x_{2}-x_{3}|}\, \frac{1-\tau}{|\frac{x_{2}-x_{3}}{|x_{2}-x_{3}|}+\tau \frac{x_{1}-x_{2}}{|x_{2}-x_{3}|}|}\, \leq\, \frac{|x_{1}-x_{2}|}{|x_{2}-x_{3}|}
\een
under the integral from $\tau=0$ to $\tau=1$. As an example, consider the following contribution to the l.h.s. of \eqref{partOPEboundeq}:
\ben
\begin{split}
&\Big|\partial_{\vec{p}}^{w}(1-\sum_{j=0}^{D-[A_{1}]-[A_{2}]}\T^{j}_{(x_{1},x_{2},x_{3})\to (x_{2},x_{2},x_{3})})\\
&\qquad\qquad\times\L_{2n_{1},l_{1}}^{\Lambda,\Lambda_{0}}(\O_{A_{1}}\otimes\O_{A_{3}}; {p}_{1},\ldots,p_{2n_{1}-1})\L^{\Lambda,\Lambda_{0}}_{2n_{2},l_{2}}(\O_{A_{2}}; p_{2n_{1}},\ldots,p_{2n})
\Big|\\
&\leq \, 2\, \frac{|x_{1}-x_{2}|^{\Delta}}{(\Delta-1)!}\sum_{|v|=\Delta}\sup_{0\leq\tau\leq 1}\Big|\partial_{\vec{p}}^{w}(1-\tau)^{\Delta-1} e^{ix_{3}(p_{1}+\ldots+p_{2n_{1}-1})} e^{ix_{2}(p_{2n_{1}}+\ldots+p_{2n})} \\
&  \L_{2n_{1},l_{1}}^{\Lambda,\Lambda_{0}}(\partial^{v}\O_{A_{1}}(x_{2}-x_{3}+\tau(x_{1}-x_{2}))\otimes\O_{A_{3}}(0); {p}_{1},\ldots,p_{2n_{1}-1})\L^{\Lambda,\Lambda_{0}}_{2n_{2},l_{2}}(\O_{A_{2}}(0); p_{2n_{1}},\ldots,p_{2n}) \Big|\\
&\leq \sup_{0\leq\tau\leq 1} \left( \frac{|1-\tau|^{\Delta-1}\,|x_{1}-x_{2}|^{\Delta}}{|x_{2}-x_{3}+\tau(x_{1}-x_{2})|^{\Delta+[A_{1}]+[A_{3}]+1}}\right)\, \frac{\prod_{i=1}[A_{i}]!}{\sqrt{[A_{2}]!}}\\
&\ \times m^{[A_{2}]-2n-|w|-1}\, \tilde{K}^{|w|+D'+\Delta}|w|!  \sup(\frac{|\vec{p}|}{m},1)^{{d}(3,n,l,w,D'+\Delta)}\, \sum_{\lambda=0}^{2l+n+1}\frac{\log^{\lambda}_{+}(\frac{|\vec{p}|}{m})}{2^{\lambda}\lambda!}\ \sup(1,m |\vec{x}|)^{|w|}\Big|\\
&\leq \, m^{-2n-|w|-1 }\, \tilde{K}^{|w|}|w|! \ \prod_{i=1}[A_{i}]!\ \sum_{\lambda=0}^{2l+n+1}\frac{\log^{\lambda}_{+}(\frac{|\vec{p}|}{m})}{2^{\lambda}\lambda!}\ \sup(1,m |\vec{x}|)^{|w|}\\
&\times \left(\frac{\tilde{K}\ \sup(\frac{|\vec{p}|}{m},1)^{2n+2l+5}}{\min(|x_{2}-x_{3}|,|x_{1}-x_{3}|)}\right)^{[A_{1}]+[A_{3}]+1} \left(\frac{\tilde{K}\ \sup(\frac{|\vec{p}|}{m},1)^{2n+2l+5}}{ [A_{2}]!^{1/2[A_{2}]}/m)}\right)^{[A_{2}]} \\
&\times \left( \frac{\tilde{K}\,  \sup(\frac{|\vec{p}|}{m},1)^{2n+2l+5} \,|x_{1}-x_{2}|}{ |x_{2}-x_{3}|}\right)^{\Delta-1}\, \left( \frac{\tilde{K}\,  \sup(\frac{|\vec{p}|}{m},1)^{2n+2l+5} \,|x_{1}-x_{2}|}{ \min(|x_{1}-x_{3}|, |x_{2}-x_{3}|)}\right)
\end{split}
\een
Here we have used the known bounds for the CAG's with one and two insertions as well as the inequality \eqref{xbehav1}.   In a similar manner one checks that all the other terms on the r.h.s. of lemma \ref{partremainder} satisfy the bound \eqref{partOPEboundeq} as well.
\end{proof}

\noindent \textbf{Remark:} Obviously, the r.h.s. of \eqref{partOPEboundeq} vanishes as we take the limit $\Delta\to\infty$ provided that
\ben\label{convS}
\left( \frac{\tilde{K}\,  \sup(\frac{|\vec{p}|}{m},1)^{n+2l+5} \,|x_{1}-x_{2}|}{|x_{2}-x_{3}|}\right)\, <\, 1
\een
which defines an open spacetime region $(x_{1},x_{2},x_{3})\in\mathcal{S}(n,l,\vec{p},m)$ for any finite values of $n,l,m$ and $\vec{p}$. Hence, the partial OPE converges in that region, which is the upshot of this section.

\subsection{Proof of Factorization}

We are now ready to give the proof of the factorization identity,  theorem \ref{thmfact}. 
%
We have found in eq.\eqref{Rdef} for $N=3$ that
\ben\label{OPEalt1}
\begin{split}
G^{0,\Lambda_{0}}(\bigotimes_{i=1}^{3}\O_{A_{i}}(x_{i}))&=\sum_{[C]\leq D} \C_{A_{1}A_{2} A_{3}}^{C}(x_{1},x_{2},x_{3})\, L^{0,\Lambda_{0}}(\O_{C}(x_{3}))+R_{D}^{0,\Lao}(\bigotimes_{i=1}^{3}\O_{A_{i}}(x_{i}))\quad .
\end{split}
\een
We may also write
\ben\label{OPEalt2}
\begin{split}
G^{0,\Lambda_{0}}(\bigotimes_{i=1}^{3}\O_{A_{i}}(x_{i}))&=\sum_{[C_{1}]\leq D_{1}}\!\!\! \C_{A_{1}A_{2}}^{C_{1}}(x_{1},x_{2})\ G^{0,\Lambda_{0}}(\O_{C_{1}}(x_{2})\otimes\O_{A_{3}}(x_{3}))+ (\text{partial remainder})
\\
&=\sum_{[C_{1}]\leq D_{1}}\sum_{[C_{2}]\leq D_{2}} \C_{A_{1}A_{2}}^{C_{1}}(x_{1},x_{2})\ \C_{C_{1}A_{3}}^{C_{2}}(x_{2},x_{3}) L^{0,\Lambda_{0}}(\O_{C_{2}}(x_{3}))\\
&+ \sum_{[C_{1}]\leq D_{1}}\!\!\! \C_{A_{1}A_{2}}^{C_{1}}(x_{1},x_{2}) R^{0,\Lambda_{0}}_{D_{2}}(\O_{C_{1}}(x_{2})\otimes\O_{A_{3}}(x_{3}))  \\
&+(\text{partial remainder})  \ ,
\end{split}
\een
where the `partial remainder', i.e. the remainder of the partial OPE, is the expression~\eqref{partialremainder} with $D$ replaced by $D_1$.
The key point is that this partial remainder is bounded by theorem~\ref{partOPEbound}.

We will exploit these two different ways of writing $G^{0,\Lambda_{0}}$. To do this, we
firstly note the relation
\ben\label{DBRD}
\D^{B}\, \{R_{D}^{0,\Lao}(\bigotimes_{i=1}^{N}\O_{A_{i}}) \}=0
\een
for $D\geq [B]$. To see this, we use lemma \ref{remainder} to express the remainder $R_{D}^{0,\Lao}$ in terms of the regularized AG's, $G^{0,\Lambda_{0}}_{D}$. Then we rewrite $G^{0,\Lambda_{0}}_{D}$ using \eqref{GDdef}. According to the boundary conditions for the $F_{D}^{\Lambda,\Lambda_{0}}$ functionals, eq.\eqref{Gbound1}, we find that $\D^{B}F_{D}^{0,\Lambda_{0}}$ vanishes for $D\geq [B]$. It remains to show that the contribution in \eqref{GDdef} to the remainder from the product of CAG's with one insertion vanishes, too. It reads explicitly
\ben
\begin{split}
&(1-\sum_{j\leq\Delta}\T^{j}_{\vec{x}\to \vec{y}=(x_{N},\dots, x_{N})})\,\prod_{i=1}^{N}\ L^{0,\Lambda_{0}}(\O_{A_{i}}(x_{i}))\\
&=\sum_{|v|=\Delta+1}\frac{(\vec{x}-\vec{y})^{v}}{(|v|-1)!}\int_{0}^{1}\d\tau\, (1-\tau)^{|v|-1}\, \partial^{v}_{\vec{y}+\tau(\vec{x}-\vec{y})} \prod_{i=1}^{N}\ L^{0,\Lambda_{0}}(\O_{A_{i}}(x_{N}+\tau( x_{i}-x_{N})))
\end{split}
\een
where again we expressed the remainder of the Taylor expansion through an integral formula. Using the Lowenstein rule \eqref{loew1} to pull the derivatives into the CAG's, it follows that this expression vanishes as we apply $\D^{B}$ due to the boundary condition \eqref{BCL1}. This yields eq.\eqref{DBRD}.

Secondly, the boundary condition \eqref{BCL1} implies
\ben
\D^{B}\, \sum_{[C]\leq D} \C_{A_{1}\ldots A_{N}}^{C}\, L^{0,\Lambda_{0}}(\O_{C}) = \D^{B}\, \sum_{[C]< [B]} \C_{A_{1}\ldots A_{N}}^{C}\, L^{0,\Lambda_{0}}(\O_{C})\, +\,  \C_{A_{1}\ldots A_{N}}^{B}
\een
for $D\geq [B]$. Thus, applying the operator $\D^{B}$ to~\eqref{OPEalt1} and~\eqref{OPEalt2} and choosing $D,D_{2}\geq [B]$, we obtain
\ben
\begin{split}
&\D^{B}\sum_{[C]\leq [B]} \  \C_{A_{1}A_{2} A_{3}}^{C}(x_{1},x_{2},x_{3})\, L^{0,\Lambda_{0}}(\O_{C}(x_{3}))\, +\,  \C_{A_{1}A_{2} A_{3}}^{B}(x_{1},x_{2},x_{3})\\
=\, &\D^{B}\sum_{[C_{1}]\leq D_{1}}\sum_{[C_{2}]\leq [B]} \C_{A_{1}A_{2}}^{C_{1}}(x_{1},x_{2})\ \C_{C_{1}A_{3}}^{C_{2}}(x_{2},x_{3}) L^{0,\Lambda_{0}}(\O_{C_{2}}(x_{3}))\\
&+\sum_{[C_{1}]\leq D_{1}}\C_{A_{1}A_{2}}^{C_{1}}(x_{1},x_{2})\ \C_{C_{1}A_{3}}^{B}(x_{2},x_{3})+\D^{B}(\text{partial remainder})
\end{split}
\een
where we used eq.\eqref{DBRD} in order to get rid of the other two remainder terms. Now we take $D_{1}\to\infty$ in this equation. Assuming that our spacetime arguments  satisfy the condition $\frac{|x_{1}-x_2|}{|x_2-x_{3}|} <1/\tilde{K}$, it follows from theorem \ref{partOPEbound} that the expression $\D^{B}$(partial remainder) tends to zero (here it is crucial that we have a bound also for arbitrary momentum derivatives of the remainder of the partial OPE, since $\D^{B}$ also includes derivatives $\partial_{\vec{p}}^{w}$). Note that, according to theorem \ref{partOPEbound}, the constant $\tilde{K}$ depends on the loop order $l$ and on $n_{B}$, where $B=\{n_{B}, w_{B}\}$. We get
\ben\label{factinduct}
\begin{split}
&\D^{B}\{ \sum_{[C]< [B]} \C_{A_{1}A_{2} A_{3}}^{C}\, L^{0,\Lambda_{0}}(\O_{C})\} +  \C_{A_{1}A_{2} A_{3}}^{B}\\
&\qquad =\lim_{D_{1}\to\infty}\Big[ \D^{B}\{\!\! \sum_{[C_{1}]\leq D_{1}}\sum_{[C_{2}]< [B]} \C_{A_{1}A_{2} }^{C_{1}}\, \C_{C_{1} A_{3}}^{C_{2}}\, L^{0,\Lambda_{0}}(\O_{C_{2}}) \} + \sum_{[C_{1}]\leq D_{1}}\C_{A_{1}A_{2} }^{C_{1}}\, \C_{C_{1} A_{3}}^{B}\Big]\quad .
\end{split}
\een
The proof can now be finished by an induction in $[B]$:

\paragraph*{Induction start ($[B]=0$):} Let  $\frac{|x_{1}-x_2|}{|x_2-x_{3}|} <1/\tilde{K}$. In this case  equation \eqref{factinduct} yields
\ben
\C_{A_{1} A_{2} A_{3}}^{B} = \lim_{D\to\infty}\sum_{[C]\leq D}\C_{A_{1} A_{2}}^{C}\, \C_{C A_{3}}^{B}
\een
as claimed in theorem \ref{thmfact}.

\paragraph*{Induction step:} Assume theorem \ref{thmfact} holds for all $\C_{A_{1}A_{2} A_{3}}^{C}$ with $[C]< D_{0}\in\mathbb{N}$.  Then for any $B$ with $[B]=D_{0}$ we  have according to equation \eqref{factinduct} (assuming also $\frac{|x_{1}-x_2|}{|x_2-x_{3}|} <1/\tilde{K}$, which depends on $n_{B}$)
\ben\label{Csumcancel}
\begin{split}
 \sum_{[C] < [B]}& \C_{A_{1}A_{2} A_{3}}^{C}\,\D^{B} L^{\Lambda,\Lambda_{0}}(\O_{C}(0))+\C_{A_{1}A_{2} A_{3}}^{B} \\
  =&\lim_{D_{1}\to\infty}\sum_{[C_{1}]\leq D_{1}} \C_{A_{1} A_{2}}^{C_{1}}\Big(\sum_{[C]< [B]} \C_{C_{1} A_{3}}^{C}\,\D^{B} L^{\Lambda,\Lambda_{0}}(\O_{C}(0))+ \C_{C_{1}A_{3}}^{B}\Big)
\end{split}
\een
Recall, however, that our induction hypothesis $\C_{A_{1}A_{2} A_{3}}^{C} = \lim_{D_{1}\to\infty}\sum_{[C_{1}]\leq D_{1}}\C_{A_{1} A_{2}}^{C_{1}}\, \C_{C_{1} A_{3}}^{C}$ holds for all $C$ with $[C]<[B]=D_{0}$ and for $\frac{|x_{1}-x_2|}{|x_2-x_{3}|} <1/\tilde{K}_{C}$, where $\tilde{K}_{C}$ depends on $C$. Assuming  $\frac{|x_{1}-x_2|}{|x_2-x_{3}|}<\min_{[C]< [B]}(1/\tilde{K}_{C})$, we find that the sums over $C$ on both sides of equation \eqref{Csumcancel} cancel, and we are left with the claim of theorem \ref{thmfact}. \qed

\section{Summary and outlook}

We have first established convergence of the OPE in Euclidean $\varphi^{4}$-theory for the expansion of $N=3$ fields, generalizing the statement for $N=2$ proved in \cite{Hollands:2011gf}. Obviously one would like to further generalize this result to an arbitrary number $N$ of fields, which would yield strong support to the viewpoint put forward in \cite{Hollands:2008vq, Hollands:2009fc, Hollands:2008wr} that quantum field theory can be defined in terms of OPE coefficients and one-point functions as fundamental objects. This generalization appears to be possible, following a similar argument as in the present paper. However, it also requires substantial amounts of additional bookkeeping on top of the already quite heavy derivation presented here. Therefore, we will leave this topic to a future publication.

Concerning our second main result, the factorization identity satisfied by the three-point OPE coefficients, some possible lines of future work come to mind immediately. First, one would again like to generalize this result to the factorizations of an $N$-point coefficient. Further, it would be interesting to try and improve our bounds in order to prove factorization of the OPE not only for configurations $|x_{1}-x_{2}|\ll |x_{2}-x_{3}|$, as in theorem \ref{thmfact}, but for any $|x_{1}-x_{2}| < |x_{2}-x_{3}|$. This would yield the stronger associativity/consistency conditions proposed in \cite{Hollands:2008vq}. However, this task seems significantly more challenging, and it is certainly possible that one can not improve on the factorization property presented here.

\paragraph{Acknowledgements:} S.H. acknowledges support through ERC grant QC~\&~C~259562. J.H. acknowledges support by Leverhulme Trust grant F/00407/BM. We are grateful to Ch. Kopper for comments and discussions.

\appendix

\section{Proof of lemma \ref{remainder}}\label{app:remainder}

The proof follows the same strategy as the proof of lemma 4.1 in \cite{Hollands:2011gf}. Let us assume $x_{N}=0$ for the moment. To begin with, consider the telescopic sum
\ben\label{Ftelescop}
G^{\Lambda,\Lambda_{0}}(\bigotimes_{i=1}^{N}\O_{A_{i}})=G^{\Lambda,\Lambda_{0}}_{D}(\bigotimes_{i=1}^{N}\O_{A_{i}})+\sum_{j=0}^{D}[G^{\Lambda,\Lambda_{0}}_{j-1}(\bigotimes_{i=1}^{N}\O_{A_{i}})-G^{\Lambda,\Lambda_{0}}_{j}(\bigotimes_{i=1}^{N}\O_{A_{i}})]
\een
with $ D<D'=\sum_{i}[A_{i}]$. Further, note that for any $j\in\mathbb{N}$ we have
\ben\label{Gdiff}
G^{\Lambda,\Lambda_{0}}_{j-1}(\bigotimes_{i=1}^{N}\O_{A_{i}})-G^{\Lambda,\Lambda_{0}}_{j}(\bigotimes_{i=1}^{N}\O_{A_{i}})=\sum_{C:[C]=j}\D^{C}\{F^{0,\Lambda_{0}}_{j-1}(\bigotimes_{i=1}^{N}\O_{A_{i}})\}\, L^{\Lambda,\Lambda_{0}}(\O_{C}(0))\ ,
\een
which can be seen by checking that both sides of the equation satisfy the same linear homogeneous FE and the same boundary conditions, which are
\begin{align*}
\partial^{w}_{\vec{p}}\left(\G^{0,\Lambda_{0}}_{n,l,j-1}(\bigotimes_{i=1}^{N}\O_{A_{i}};\vec{0})-\G^{0,\Lambda_{0}}_{n,l,j}(\bigotimes_{i=1}^{N}\O_{A_{i}};\vec{0})\right)=&\, 0 \quad &\text{ for }n+|w|< j \\
\partial^{w}_{\vec{p}}\left(\G^{0,\Lambda_{0}}_{n,l,j-1}(\bigotimes_{i=1}^{N}\O_{A_{i}};\vec{0})-\G^{0,\Lambda_{0}}_{n,l,j}(\bigotimes_{i=1}^{N}\O_{A_{i}};\vec{0})\right)=&\, \partial^{w}_{\vec{p}}\F^{0,\Lambda_{0}}_{n,l,j-1}(\bigotimes_{i=1}^{N}\O_{A_{i}};\vec{0})\, &\text{ for }n+|w|= j \\
\partial^{w}_{\vec{p}}\left(\G^{\Lambda_{0},\Lambda_{0}}_{n,l,j-1}(\bigotimes_{i=1}^{N}\O_{A_{i}};\vec{p})-\G^{\Lambda_{0},\Lambda_{0}}_{n,l,j}(\bigotimes_{i=1}^{N}\O_{A_{i}};\vec{p})\right)=&\, 0 \, &\text{ for }n+|w|>j \\
\end{align*}
in both cases. Further, we will need the identity
\ben\label{TG}
\begin{split}
&\T^{\Delta+1}_{\vec{x}\to\vec{0}}\, {G}_{D+1}^{\Lambda,\Lambda_{0}}(\bigotimes_{i=1}^{N}\O_{A_{i}})\\
&=(-1)^{N+1}\sum_{[C]=D+1}\D^{C}\left\{\T^{\Delta+1}_{\vec{x}\to\vec{0}}\prod_{i=1}^{N}L^{0,\Lambda_{0}}(\O_{A_{i}})\right\}\, L^{\Lambda,\Lambda_{0}}(\O_{C}(0))\\
&=(-1)^{N+1}\sum_{[C]=D+1}\D^{C}\left\{(1-\sum_{j\leq \Delta}\T^{j}_{\vec{x}\to\vec{0}})\prod_{i=1}^{N}L^{0,\Lambda_{0}}(\O_{A_{i}})\right\}\, L^{\Lambda,\Lambda_{0}}(\O_{C}(0))\, ,
\end{split}
\een
where $\Delta=D-D'$.
The first equality follows again directly by comparison of the FE and boundary conditions of both sides of the equation. In the last line we applied the formula for the Taylor expansion with remainder, eq.\eqref{Rfacttaylor}. The Taylor expansion terms of degree $j>\Delta+1$ vanish due to the boundary conditions of the CAG's with one insertion.

  Note also that the boundary conditions for the CAG's with one insertion, eq.\eqref{BCL1}, imply
\ben\label{1insBC}
\D^{C}\{G^{0,\Lambda_{0}}_{[C]-1}(\bigotimes_{i=1}^{N}\O_{A_{i}})\}=\D^{C}\{F^{0,\Lambda_{0}}_{[C]-1}(\bigotimes_{i=1}^{N}\O_{A_{i}})\}\quad \text{ for }D< D'\, .
\een

We now prove lemma \ref{remainder} by induction in $D$:
\begin{enumerate}

\item Induction start: For $D=-1$ the sum in eq.\eqref{Rdef} vanishes and we obtain the lemma for $D=-1$, $R_{D=-1}^{\Lambda,\Lambda_{0}}(\bigotimes_{i=1}^{N}\O_{A_{i}})=G_{D=-1}^{\Lambda,\Lambda_{0}}(\bigotimes_{i=1}^{N}\O_{A_{i}})$, trivially.

\item Induction step: Assume the lemma holds up to order $D$, i.e. assume
\ben
R_{\tilde{D}}^{\Lambda,\Lambda_{0}}(\bigotimes_{i=1}^{N}\O_{A_{i}})=(1-\sum_{j\leq\tilde{D}-D'}\T^{j}_{\vec{x}\to\vec{0}})\, {G}^{\Lambda,\Lambda_{0}}_{\tilde{D}} (\bigotimes_{i=1}^{N}\O_{A_{i}}(x_{i}))
\een
for all $\tilde{D}\leq D$. Using again eq.\eqref{Rdef}, we then get
\ben
\begin{split}
&R_{D+1}^{\Lambda,\Lambda_{0}}(\bigotimes_{i=1}^{N}\O_{A_{i}})=R_{D}^{\Lambda,\Lambda_{0}}(\bigotimes_{i=1}^{N}\O_{A_{i}})-\sum_{[C]=D+1}\C_{A_{1}\ldots A_{N}}^{C}\, L^{\Lambda,\Lambda_{0}}(\O_{C})\\
=& (1-\sum_{j\leq\Delta}\T^{j}_{\vec{x}\to\vec{0}})\, {G}^{\Lambda,\Lambda_{0}}_{D} (\bigotimes_{i=1}^{N}\O_{A_{i}}(x_{i}))
-\sum_{[C]=D+1}\C_{A_{1}\ldots A_{N}}^{C}\, L^{\Lambda,\Lambda_{0}}(\O_{C})\\
=& (1-\sum_{j\leq\Delta+1}\T^{j}_{\vec{x}\to\vec{0}})\,{G}^{\Lambda,\Lambda_{0}}_{D+1} (\bigotimes_{i=1}^{N}\O_{A_{i}}(x_{i})) \\
+&(1-\sum_{j\leq\Delta}\T^{j}_{\vec{x}\to\vec{0}})\left\{ {G}^{\Lambda,\Lambda_{0}}_{D} (\bigotimes_{i=1}^{N}\O_{A_{i}}(x_{i})) -  {G}^{\Lambda,\Lambda_{0}}_{D+1} (\bigotimes_{i=1}^{N}\O_{A_{i}}(x_{i}))\right\}\\
+&\,\T^{\Delta+1}_{\vec{x}\to\vec{0}}\, {G}_{D+1}^{\Lambda,\Lambda_{0}}(\bigotimes_{i=1}^{N}\O_{A_{i}})
-\sum_{[C]=D+1}\C_{A_{1}\ldots A_{N}}^{C}\, L^{\Lambda,\Lambda_{0}}(\O_{C})
\end{split}
\een
where $\Delta=D-D'$. Using eqs.\eqref{Gdiff} and \eqref{TG} to replace the corresponding terms in the last two lines and also recalling the definition of the OPE coefficients $\C_{A_{1}\ldots A_{N}}^{C}$, eq.\eqref{OPEhigh}, we find that the last three terms cancel out (in the case $\Delta<0$ one also has to take into account eq.\eqref{1insBC} to see this), leaving the claim of the lemma at order $D+1$ in the case $x_{N}=0$.

 The case $x_{N}\neq 0$ then follows by translation covariance.
\hfill\qedsymbol

\end{enumerate}

\section{Proof of lemma \ref{partremainder}}\label{app:partrem}

We will divide the problem into two parts. Consider first the following contribution to the r.h.s. of eq.\eqref{eq:partrem}:
\ben\label{lemmaParthalf1}
\begin{split}
&\left[ L^{\Lambda,\Lambda_{0}}(\O_{A_{1}})L^{\Lambda,\Lambda_{0}}(\O_{A_{2}}) -  L^{\Lambda,\Lambda_{0}}(\O_{A_{1}}\otimes\O_{A_{2}}) \right] L^{\Lambda,\Lambda_{0}}(\O_{A_{3}})-\sum_{[C]\leq D}\C_{A_{1}A_{2}}^{C}\, L^{\Lambda,\Lambda_{0}}(\O_{C})L^{\Lambda,\Lambda_{0}}(\O_{A_{3}})\\
&=-R^{\Lambda,\Lambda_{0}}_{D}(\O_{A_{1}}\otimes\O_{A_{2}})L^{\Lambda,\Lambda_{0}}(\O_{A_{3}})\\
&=-(1-\!\!\!\!\!\!\!\sum_{j=0}^{D-[A_{1}]-[A_{2}]}\!\!\! \T^{j}_{(x_{1},x_{2},x_{3})\to (x_{2},x_{2},x_{3})}) \Big[ L^{\Lambda,\Lambda_{0}}_{D}(\O_{A_{1}}\otimes\O_{A_{2}})-L^{\Lambda,\Lambda_{0}}(\O_{A_{1}})L^{\Lambda,\Lambda_{0}}(\O_{A_{2}})\Big]L^{\Lambda,\Lambda_{0}}(\O_{A_{3}})
\end{split}
\een
Here we have used lemma \ref{remainder} in the last step. 
Now consider
\ben\label{Partunreg}
L^{\Lambda,\Lambda_{0}}(\bigotimes_{i=1}^{3}\O_{A_{i}})- L^{\Lambda,\Lambda_{0}}(\O_{A_{1}}\otimes\O_{A_{3}})L^{\Lambda,\Lambda_{0}}(\O_{A_{2}})- L^{\Lambda,\Lambda_{0}}(\O_{A_{2}}\otimes\O_{A_{3}})L^{\Lambda,\Lambda_{0}}(\O_{A_{1}})
\een
The FE for this expression is of the form
\ben
\begin{split}
\partial_{\Lambda} [\text{eq.}\eqref{Partunreg}]=&\frac{1}{2}\bra \varp \, ,\, \dot{C}\star \varp  \ket\, [\text{eq.}\eqref{Partunreg}] -\bra \varp  [\text{eq.}\eqref{Partunreg}] \, ,\, \dot{C}\star \varp L^{\Lambda,\Lambda_{0}} \ket\\
-& \bra \varp  G^{\Lambda,\Lambda_{0}}(\O_{A_{1}}\otimes\O_{A_{2}}) \, ,\, \dot{C}\star \varp L^{\Lambda,\Lambda_{0}}(\O_{A_{3}}) \ket\ ,
\end{split}
\een
and the boundary conditions are trivial
\ben\label{PartregBC}
\partial_{\vec{p}}^{w}[\text{eq.}\eqref{Partunreg}]^{\Lambda_{0},\Lambda_{0}}_{n,l}(\vec{p})=0\quad \text{for all }n,l,w.
\een
Note that $\sum_{[C]\leq D}\C_{A_{1}A_{2}}^{C} L^{\Lambda,\Lambda_{0}}(\O_{C}\otimes\O_{A_{3}})$ obeys the same boundary conditions. Subtracting this expression from eq.\eqref{Partunreg} and taking the derivative with respect to $\Lambda$ then yields
\ben
\begin{split}
\partial_{\Lambda}[\bullet]=&\frac{1}{2}\bra \varp \, ,\, \dot{C}\star \varp  \ket\, [\bullet] -\bra \varp  [\bullet] \, ,\, \dot{C}\star \varp L^{\Lambda,\Lambda_{0}} \ket\\
-& \bra \varp  \left[G^{\Lambda,\Lambda_{0}}(\O_{A_{1}}\otimes\O_{A_{2}})-\sum_{[C]\leq D}\C_{A_{1}A_{2}}^{C} L^{\Lambda,\Lambda_{0}}(\O_{C}) \right] ,\, \dot{C}\star \varp L^{\Lambda,\Lambda_{0}}(\O_{A_{3}})\\
=&\frac{1}{2}\bra \varp \, ,\, \dot{C}\star \varp  \ket\, [\bullet] -\bra \varp  [\bullet] \, ,\, \dot{C}\star \varp L^{\Lambda,\Lambda_{0}} \ket\\
-&(1-\sum_{j=0}^{D-[A_{1}]-[A_{2}]}\T^{j}_{(x_{1},x_{2},x_{3})\to(x_{2},x_{2},x_{3})})\\
&\qquad\bra \varp  \left[ L^{\Lambda,\Lambda_{0}}_{D}(\O_{A_{1}}\otimes\O_{A_{2}})- L^{\Lambda,\Lambda_{0}}(\O_{A_{1}})L^{\Lambda,\Lambda_{0}}(\O_{A_{2}}) \right] ,\, \dot{C}\star \varp L^{\Lambda,\Lambda_{0}}(\O_{A_{3}})\\
\end{split}
\een
where $[\bullet]$ stands for the difference of eq.\eqref{Partunreg} and $\sum_{[C]\leq D}\C_{A_{1}A_{2}}^{C} L^{\Lambda,\Lambda_{0}}(\O_{C}\otimes\O_{A_{3}})$ and where again we used lemma \ref{remainder} to simplify the remainder of the 2-point OPE. But this is exactly the FE satisfied by the expression
\ben\label{partregBCtr}
\begin{split}
(1-\sum_{j=0}^{D-[A_{1}]-[A_{2}]}\T^{j}_{(x_{1},x_{2},x_{3})\to(x_{2},x_{2},x_{3})})\Big[L^{\Lambda,\Lambda_{0}}_{\bf D}(\bigotimes_{i=1}^{3}\O_{A_{i}})-& L^{\Lambda,\Lambda_{0}}(\O_{A_{1}}\otimes\O_{A_{3}})L^{\Lambda,\Lambda_{0}}(\O_{A_{2}})\\
-& L^{\Lambda,\Lambda_{0}}(\O_{A_{2}}\otimes\O_{A_{3}})L^{\Lambda,\Lambda_{0}}(\O_{A_{1}})\Big]
\end{split}
\een
with $\bfD = (D_{\{1,2,3\}}=-1,D_{\{2,3\}}=-1, D_{\{1,3\}}=-1,D_{\{1,2\}}=D)$. The terms in eq.\eqref{partregBCtr}  also satisfy the trivial boundary conditions of the form \eqref{PartregBC}. Thus, we conclude that eq.\eqref{Partunreg} minus  $\sum_{[C]\leq D}\C_{A_{1}A_{2}}^{C} L^{\Lambda,\Lambda_{0}}(\O_{C}\otimes\O_{A_{3}})$ is equal to eq.\eqref{partregBCtr}. Putting this together with eq.\eqref{lemmaParthalf1} we have the proof of the lemma.\hfill \qedsymbol



\end{document}